\documentclass[pra,twocolumn, preprintnumbers,amsmath,amssymb, superscriptaddress,nofootinbib]{revtex4}

\usepackage{graphicx, color, graphpap}
\usepackage{amsmath}
\usepackage{amssymb}
\usepackage{amsthm}
\usepackage{enumerate}
\usepackage{aliascnt}
\usepackage[usenames,dvipsnames]{xcolor}
\usepackage{cleveref}
\usepackage{multirow}
\usepackage[T1]{fontenc}

\long\def\ca#1\cb{} 

\newcommand{\ketbra}[2]{| \hspace{1pt} #1 \rangle \langle #2 \hspace{1pt} |}
\newcommand{\outp}[1]{| \hspace{1pt} #1 \rangle \langle #1 \hspace{1pt} |}

\newcommand{\ketbraq}[1]{\ketbra{#1}{#1}}
\newcommand{\bramatket}[3]{\langle #1 \hspace{1pt} | #2 | \hspace{1pt} #3 \rangle}
\newcommand{\bramatketq}[2]{\bramatket{#1}{#2}{#1}}

\newcommand{\nbox}[2][9]{\hspace{#1pt} \mbox{#2} \hspace{#1pt}}

\newcommand{\ket}[1]{|#1\rangle}               
\newcommand{\bra}[1]{\langle #1|}              

\renewcommand{\geq}{\geqslant}
\renewcommand{\leq}{\leqslant}

\DeclareMathOperator{\betamax}{\beta_{max}}
\DeclareMathOperator{\calP}{\mathcal{P}}
\DeclareMathOperator{\calM}{\mathcal{M}}
\DeclareMathOperator{\Herm}{Herm}
\DeclareMathOperator{\Pos}{Pos}
\newcommand{\calE}{\mathcal{E}}
\DeclareMathOperator{\calT}{\mathcal{T}}

\newcommand{\calR}{\mathcal{R}}

\newcommand{\calI}{\mathcal{I}}
\newcommand{\calD}{\mathcal{D}}

\newcommand*{\id}{\openone}

\newcommand*{\Hmin}{H_{\min}}
\newcommand*{\Hmax}{H_{\max}}

\newtheoremstyle{defi}
     {10pt}
     {10pt}
     {\sl}
     {}
     {\bfseries}
     {:}
     {.5em}
     {}

\theoremstyle{defi}
\newtheorem{thm}{Theorem}
\newtheorem{prop}[thm]{Proposition}
\newtheorem{lemma}[thm]{Lemma}
\newtheorem{cor}[thm]{Corollary}
\newtheorem*{notation}{Notation}

\newtheorem*{lprog}{Linear Program}
\newtheorem{defi}[thm]{Definition}

\newcommand{\dt}[1]{\textbf{#1}}

\newtheorem{rmk}[thm]{Remark}
\newtheorem{ex}[thm]{Example}

\newenvironment{linprog}[2]{
\smallskip
\begin{tabular}{ll}
#1 & #2\\
subject to
}
{
\end{tabular}
\smallskip
}

\crefname{prop}{Proposition}{Propositions}
\crefname{defi}{Definition}{Definitions}
\crefname{lemma}{Lemma}{Lemmas}
\crefname{thm}{Theorem}{Theorems}
\crefname{section}{Section}{Sections}
\crefname{exercise}{Exercise}{Exercises}
\crefname{rmk}{Remark}{Remarks}
\crefname{figure}{Figure}{Figures}
\crefname{equation}{equation}{equations}
\crefname{lprog}{Linear Program}{Linear Programs}
\crefname{ex}{Example}{Examples}

\newcommand{\tr}{{\mathsf{Tr}}}
\newcommand{\R}{\mathcal{R}}

\newcommand{\hmin}{{{\rm H}_{\rm min}}}
\newcommand{\dec}{{{\rm Dec}}}
\newcommand{\Dec}{{{\rm Dec}}}
\newcommand{\calH}{\mathcal{H}}
\newcommand{\calS}{\mathcal{S}}

\newcommand{\abs}[1]{\left\vert #1 \right\vert}

\begin{document}
\title{Understanding nature from experimental observations: a theory independent test for gravitational decoherence}
\author{C. Pfister}
\affiliation{QuTech, Delft University of Technology, Lorentzweg 1, 2628 CJ Delft, Netherlands}
\affiliation{Centre for Quantum Technologies, 3 Science Drive 2, 117543 Singapore}
\author{J. Kaniewski}
\affiliation{QuTech, Delft University of Technology, Lorentzweg 1, 2628 CJ Delft, Netherlands}
\affiliation{Centre for Quantum Technologies, 3 Science Drive 2, 117543 Singapore}
\author{M. Tomamichel}
\affiliation{University of Sydney, School of Physics, NSW 2006 Sydney, Australia}
\affiliation{Centre for Quantum Technologies, 3 Science Drive 2, 117543 Singapore}
\author{A. Mantri}
\affiliation{Centre for Quantum Technologies, 3 Science Drive 2, 117543 Singapore}
\author{R. Schmucker}
\affiliation{Centre for Quantum Technologies, 3 Science Drive 2, 117543 Singapore}
\author{N. McMahon}
\affiliation{ARC Centre for Engineered Quantum Systems,
School of Mathematics and Physics,
The University of Queensland, St Lucia, QLD 4072, Australia}
\author{G. Milburn}
\affiliation{ARC Centre for Engineered Quantum Systems,
School of Mathematics and Physics,
The University of Queensland, St Lucia, QLD 4072, Australia}
\author{S. Wehner}
\email{steph@locc.la}
\affiliation{QuTech, Delft University of Technology, Lorentzweg 1, 2628 CJ Delft, Netherlands}
\affiliation{Centre for Quantum Technologies, 3 Science Drive 2, 117543 Singapore}

\begin{abstract}
Quantum mechanics and the theory of gravity are presently not compatible. A particular question is whether gravity causes decoherence - an unavoidable source of noise.
Several models for gravitational decoherence have been proposed, not all of which can be described quantum mechanically \cite{Dio89,Penrose1996,Diosi2011}. In parallel, several experiments have been proposed to test some of these models \cite{Bouwmeester2012,Bouwmeester2003,ignacioGrav}, where the data obtained by such experiments is analyzed assuming quantum mechanics. Since we may need to modify quantum mechanics to account for gravity, however, one may question the validity of using quantum mechanics as a calculational tool to draw conclusions from experiments concerning gravity.

Here we propose an experiment to estimate gravitational decoherence whose conclusions hold even if quantum mechanics would need to be modified.
We first establish a general information-theoretic notion of decoherence which reduces to the standard measure within quantum mechanics. Second, drawing on ideas from quantum information, we propose a very general experiment that allows us to obtain a \emph{quantitative} 
estimate of decoherence of any physical process for any physical theory satisfying only very mild conditions.
Finally, we propose a concrete experiment using optomechanics to estimate gravitational decoherence in any such theory, including quantum mechanics as a special case.

Our work raises the interesting question whether other properties of nature could similarly be established from experimental observations alone - that is, without already having a rather well formed theory of nature like quantum mechanics to make sense of experimental data. 
\end{abstract}
\maketitle

Experiments aiming to test the presence - and amount - of gravitational decoherence generally go beyond established theory. Many theoretical models for gravitational decoherence have been proposed~\cite{Diosi2011,diosi1, diosi2, Dio89, KTM14,hu1,hu2,hu3,kay,breuer,wang}, and it is wide open if one of these proposals is correct. As such, experiments are of a highly exploratory nature, aiming to establish data points to which one may tailor future theoretical proposals.
This task is made even more difficult by the fact that quantum mechanics and gravity do not go hand in hand, and indeed quantum mechanics may need to be modified in a yet unknown way in order to account for gravitational effects such as decoherence.
We are thus compelled to design an experiment that provides a guiding light for the search for the right theoretical model - or indeed new physical theory - whose conclusions do not rely on quantum mechanics. 

\section*{Decoherence in QM}

As an easy warmup, let us first focus on the concept of decoherence \emph{within} quantum mechanics. We first show how the protocol given in Figure~\ref{fig:abstract} allows us to estimate quantum mechanical decoherence without knowing the decoherence process, and without doing quantum tomography to determine it. 
Traditionally, the presence of decoherence within quantum mechanics is related to the change of state due to measurement and the ''collapse of the wavefunction''. There are two complimentary ways to view this based on unconditional and conditional states. Given some pure quantum state, $\alpha \ket{0} + \beta \ket{1}$, and an arbitrarily accurate measurement of the variable diagonal in this basis,  the post-measurement conditional states are   $\ket{0}$ or $\ket{1}$ conditioned on the measurement outcome. On the other hand if this measurement has taken place but the results are unknown, the resulting unconditional state is given by the a quantum density matrix $\rho=|\alpha|^2|0\rangle\langle 0|+|\beta|^2|1\rangle\langle 1|$. The vanishing of the off-diagonal matrix elements in the measurement basis for the post measurement unconditional state forms decoherence. If the measurement is not arbitrarily accurate (i.e. weak) the off-diagonal matrix elements are reduced but do not vanish.  More general forms of decoherence correspond to a decay of off-diagonal terms in the density operator $\rho$ with respect to any basis, and can occur due to the interaction of the system with an environment that may not be a measurement procedure of any kind. It is clear that this way of thinking about decoherence is entirely tied to the quantum mechanical matrix formalism, and also offers little in the way of quantifying the amount of decoherence in an operationally meaningful way.

The modern way of understanding decoherence in quantum mechanics in a quantitative way is provided by quantum information theory. One thereby thinks of a decoherence process as a interaction of a system $A'$ with an environment as described in Figure~\ref{fig:deco}, resulting in a quantum channel $\Gamma_{A' \rightarrow B}$. 
The amount of of decoherence can now be quantified by the channel's ability to transmit quantum information, i.e., its quantum capacity (see e.g.~\cite{Wil13} and the appendix for further background). Concretely, one considers $n$ uses of the channel given by $\Gamma_{A'\rightarrow B}^{\otimes n}$, and asks how many qubits $k$ we can send in relation to $n$ using an error-correcting encoding. Of particular interest is thereby the so-called single shot capacity, which determines the largest rate $k/n$ up to a given error parameter for any choice of $n$~\footnote{The asymptotic quantity often considered in information theory arises in the limit of $n\rightarrow \infty$, but the single-shot capacity gives more refined statements which are valid for any $n$.}. This single-shot capacity is determined by the so-called min-entropy $\hmin(A|E)$~\cite{OneShotDecouple, nilanjana}.
\begin{figure}
\includegraphics{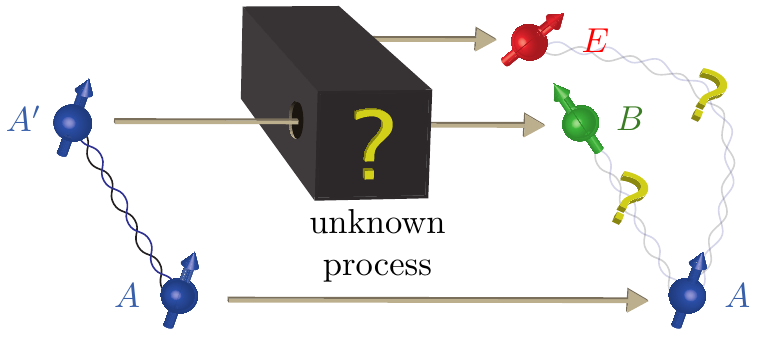}
\caption{Any decoherence process - also known as a \emph{(quantum) channel} - can be thought of as an interaction $U_I$ of the system $A'$ with an environment $E_{\rm in}$. 
In quantum mechanics, the resulting state is the output of the channel $\rho_B = \Gamma_{A'\rightarrow B}(\rho_{A'}) = \tr_{E}[U_I \rho_{A'} \otimes \outp{\Psi_{E_{\rm in}}} U_I^\dagger]$. In general, $B$ may be a smaller or larger system than $A'$. In the examples below, however, we will focus on the case where $A'$ and $B'$ have the same dimension, corresponding to the case where a fixed system $A'=B$ experiences some interaction with another system $E_{\rm in} = E$. The channels ability to preserve quantum information - that is, the amount of decoherence - can be characterized by how well it preserves entanglement between an outside system $A$ and $A'$.
We note that our treatment of theories that go beyond standard quantum mechanics makes no statement
 whether the environment is an actual physical system, or merely a mathematical Gedankenexperiment possibly used to describe an intrinsic decoherence process.
In full generality, the experiment consists of a Bell experiment in which a source of decoherence is introduced deliberately. For simplictly, we consider an experiment for the CHSH inequality, although our analysis could easily be extended to any other Bell inequality. In each run, a source prepares the maximally entangled state $\Phi_{AA'}$, where $A'$ is subsequently exposed to the decoherence process to be tested. We then perform the standard CHSH measurements: system $A$ is measured with probability $1/2$ using observables $A_0=\sigma_X$ and $A_1=\sigma_Z$ respectively. System $B$ is measured using observables $B_0=(\sigma_X-\sigma_Z)/\sqrt{2}$ and $B_1=(\sigma_X+\sigma_Z)/\sqrt{2}$ with probability $1/2$ each. 
Performing the experiment many times allows an estimate of $\beta = \tr[\rho_{AB} (A_0\otimes B_0 + A_0\otimes B_1 + A_1\otimes B_0 - A_1 \otimes B_1)]$.
}
\label{fig:deco}
\label{fig:abstract}
\end{figure}

Apart from its information-theoretic significance, the min-entropy has a beautiful operational interpretation that also makes its role as a decoherence measure intuitively apparent. 
Very roughly, the amount of decoherence can be understood as a measure of how correlated $E$ becomes with $A$. 
Suppose we start with a maximally entangled test state $\Phi_{AA'}$ where the decoherence process is applied to $A'$. This results in a state $\ket{\Psi_{ABE}}$ (see Figure~\ref{fig:deco}).
If no decoherence occurs, the output state will be of the form $\Phi_{AB} \otimes \outp{0}_E$ where $A'=B$. That is, $A$ and $B$ are maximally entangled, but $A$ and $E$ are completely uncorrelated. 
The strongest decoherence, however, produces an output state of the form $\Phi_{AE_1} \otimes \rho_{E_2} \otimes \outp{0}_B$ where $A'=E_1$ and $E=E_1E_2$. That is, $A$ is now maximally entangled with $E_1$, whereas $A$ and $B$ are completely uncorrelated. What about the intermediary regime?
The min-entropy can be written as 
\begin{align}
\hmin(A|E) = - \log d_A\ \dec(A|E)\ ,
\end{align}
where $d_A$ is the dimension of $A$, and~\cite{KRS09} 
\begin{align}
\dec(A|E) = \max_{\R_{E\rightarrow A'}} F^2(\Phi_{AA'},\id_A \otimes \R_{E\rightarrow A'}(\rho_{AE}))\ .
\end{align}
The maximization above is taken over all quantum operations $\R_{E\rightarrow A'}$ on the system $E$, which aim to bring the state $\rho_{AE}$ as close as possible to the maximally entangled state $\Phi_{AA'}$. Intuitively, $\dec(A|E)$ can thus be understood as a measure of how far the output $\rho_{AE}$ is from the setting of maximum decoherence (where $\rho_{AE} = \Phi_{AE}$ is the maximally entangled state).
If there is no decoherence, we have $\rho_{AE} = \id/d_A \otimes \rho_E$ giving $\dec(A|E) = 1/d_A^2$ and $\hmin(A|E) = \log d_A$.
If there is maximum decoherence, we have $\rho_{AE_1} = \Phi_{AA'}$ giving $\dec(A|E) = 1$ and $\hmin(A|E) = - \log d_A$ where $\R_{E\rightarrow A'} = \tr_{E_2}$ is simply the operation that discards the remainder of the environment $E_2$.
A larger value of $\dec(A|E)$ thus corresponds to a larger amount of decoherence.
In the quantum case, $\dec(A|E)$ can be computed using any semi-definite programming solver~\cite{Ren05, sedumi}.

We hence see that in quantum mechanics, the relevant measure of decoherence is simply $\dec(A|E)$. 
How can we estimate it an experiment? Our goal in deriving this estimate will be to rely on concepts that we can later extend beyond the realm of quantum theory, deriving a universally valid test.
It is clear that to estimate $\dec(A|E)$ we need to make a statement about the entanglement between $A$ and $E$ - yet $E$ is inaccesible to our experiment. 
A property of quantum mechanics known as the monogamy of entanglement~\cite{terhal:monogamy} nevertheless allows such an estimate: if $\rho_{AB}$ is highly entangled, then $\rho_{AE}$ is necessarily far from highly entangled. Since low entanglement in $\rho_{AE}$ means that $\dec(A|E)$ is low, a test that is able to detect entanglement between $A$ and $B$ should help us bound $\dec(A|E)$ from above. 
We note that whereas any experimental proposal demands that we specifiy concrete measurements to be performed, our conclusions remain valid even if we do not have full control over the measurements, possibly because they are also somehow affected by an gravitational interaction in an unknown fashion. Dealing with unknown states and measurements is the essence of so-called device independence~\cite{ABGM07} in quantum cryptography. Allowing arbitrary measurements again forms a crucial stepping stone, enabling us to extend our results beyond quantum mechanics. 

\begin{figure*}[htb]
  \begin{center}
    \includegraphics[scale=0.9]{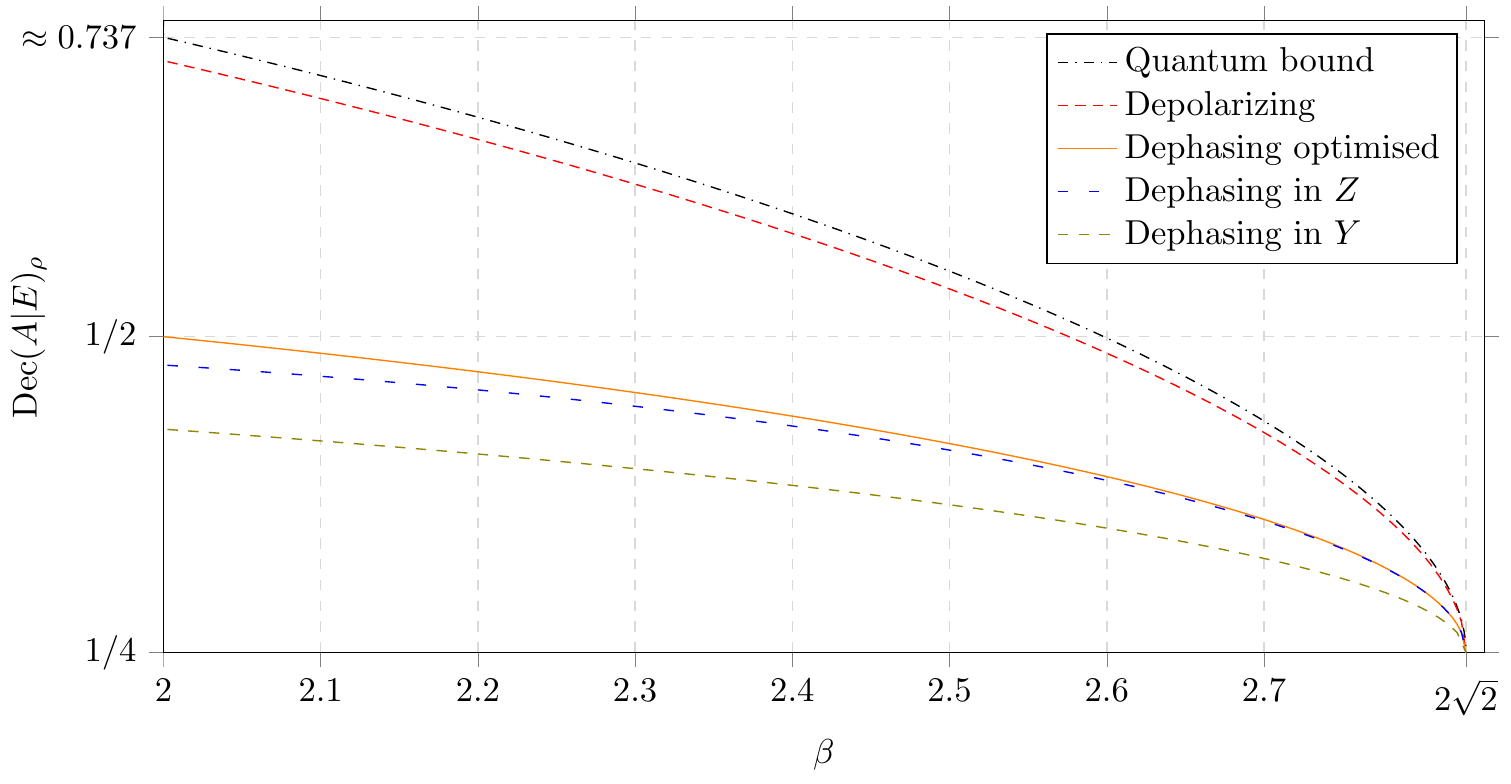}
    \caption{\textbf{Comparison of the quantum bound with the actual values of 
      $\Dec(A|E)_\rho$ for some example channels and measurements} (colors online). 
      The black dashdotted line on top shows the quantum bound, i.e. the maximal value of 
      $\Dec(A|E)_\rho$ that is compatible with a measured CHSH value $\beta$ in quantum 
      theory. The other four plots are parametric plots: The parameter that is varied is the 
      noise parameter of the channel. For each noise parameter, the value of 
      $\Dec(A|E)_\rho$ of the resulting state is calculated, as well as the CHSH value 
      $\beta$ that one would meausre for this state using the standard measurements in the 
      $X$-$Z$-plane that would be optimal for an EPR pair. This measurement happens to be 
      optimal for the resulting state for the depolarizing channel, but not for the 
      dephasing channels. The orange solid line also shows such a parametric plot for the 
      dephasing channel, but for that line, the CHSH value $\beta$ is not calculated for the 
      standard measurement for the EPR pair but for the measurement that is optimal for the 
      actual resulting state. This is done using a formula found in \cite{HHH95}. 
      The resulting curve is independent of the dephasing direction.
      \label{fig:channels-imported}}
  \end{center}
\end{figure*}

\section*{Beyond QM}
The real challenge is to show that the conclusions of our test remain valid even outside of quantum mechanics. 
Since we want to make as few assumptions as possible, we consider the most general probabilistic theory, in which we are only given a set of possible states $\Omega$ and measurements on these states. Every measurement is thereby a collection of effects $e_a: \Omega \rightarrow [0,1]$ satisfying $e_a(\omega) \geq 0$ and $\sum_a e_a(\omega) = 1$ for all $\omega \in \Omega$. We also refer to a measurement as an instrument $M = \{e_a\}_a$.  The label $a$ corresponds to a measurement outcome '$a$'. The notion of separated systems $A$, $B$ and $E$ is in general difficult to define uniquely. We thus again make the most minimal assumption possible in which we identify ''systems'' $A$, $B$ and $E$ by sets of measurements 
that can be performed. For simplicitly, we take measurements and operations in the sets $A$,$B$, and $E$ to commute, but do not impose any other strucuture. 
We thus merely use labels $A$ and $B$ and $E$ for commuting measurements.
This means that for maps going from a system $E$ to an output system $A'$ like $\R_{E \rightarrow A'}$ the map is really from $E$ to $E$ and we use $A'$ merely 
to remind ourselves we consider a restricted class of measurements on the output. Again, this is analogous to quantum mechanics where such measurements consist of operators on $A'$ and the identity elsewhere (see appendix for a discussion). 

The first obstacle consists of defining a general notion of decoherence. We saw that quantumly decoherence can be quantified by how well correlations between $A$ and $A'$ are preserved, 
and this can be measured by how well the decoherence process preserves the maximally correlated state. Indeed, we can also quantify classical noise in terms of how well it preserves correlations, where the maximally correlated state takes on the form $(1/d_A) \sum_a \outp{a}_A \otimes \outp{a}_{A'}$ for some classical symbols $a$. We hence start by defining the set of maximally correlated states, by observing a crucial and indeed defining property of the maximally correlated in quantum mechanics. Concretely, $A$ and $A'$ are maximally entangled if and only if for any von Neumann measurement on $A$, there exists a corresponding measurement on $A'$ giving the same outcome. Again, the same is also true classically but made trivial by the fact that only one measurement is allowed. In analogy, we thus define the set of maximally correlated states as
\begin{align}
\Psi_{AA'} &= \left\{\Phi \in \Omega_{AA'} \mid \forall M^A = \{e_a^A\}_a\ \exists M^B=\{e_a^B\}_a \right. \nonumber\\
&\qquad \left. \mbox{\ such that\ } \sum_a e_a^A e_a^B (\omega) = 1\right\}\ .
\end{align}
This set coincides with the set of maximally entangled states in quantum mechanics, where $A'$ can potentially contain an additional component $\sigma_{A'}$ in $\Phi_{AA'} \otimes \sigma_{A'}$
which is irrelevant to our discussion. 
We thus define
\begin{align}\label{eq:genDec}
\dec(A|E)_{\omega} = \sup_{\R_{E\rightarrow A'}} \sup_{\Phi_{AA'} \in \Psi_{AA'}} F^2 (\Phi_{AA'},\R_{E\rightarrow A'}(\omega_{AE}))\ ,
\end{align}
where $\omega_{AE}$ is the state shared between $A$ and $E$ according to the general physical theory.
The fidelity between two states $\omega_1$ and $\omega_2$ is thereby defined in full analogy to the quantum case~\cite{NC00} as
\begin{align}
F(\omega_1,\omega_2) = \inf_M F(M(\omega_1),M(\omega_2))\ ,
\end{align}
where the minimization is taken over all possible measurements $M$, and $M(\omega)$ denotes the probability distribution over the measurement outcomes of $M$. 
That is, the fidelity can be expressed as the minimum fidelity between probability distributions of classical measurement outcomes. 
How about the transformation $\R_{E \rightarrow A'}$?  In general, it is difficult to characterize the set of allowed transformations $\R_{E\rightarrow A'}$ in arbitrary physical 
theories, however we will not need make $\R_{E \rightarrow A'}$ explicit in order to bound $\dec(A|E)$.
Equation~\eqref{eq:genDec} 
gives us the familiar quantity within quantum mechanics, but provides us with an a very intuitive way to \emph{quantify} decoherence in any physical theory that admits 
maximally correlated states. We emphasize that with our general techniques the latter demand could be weakend to allow 
all theories, even those who only have (weak) approximations of maximally correlated states. However, as we are not aware of any physically motivated example of 
such a theory, we leave such an extension to future study for clarity of exposition.

The second challenge is to prove that our test actually provides a bound on $\dec(A|E)_{\omega}$. 
Note that without quantum mechanics to guide us, all that we could reasonably establish by performing measurements on $A$ and $B$ are the probabilities of outcomes $a$ and $b$ given measurement 
settings $x$ and $y$. That is, the probability 
\begin{align}
\Pr[a,b|x,y]_{\omega} = e_a^A e_b^B(\omega_{AB})\ ,
\end{align}
where $e_a^A \in M_x^A$ and $e_b^B \in M_y^B$.
Yet, given the system $E$ is entirely inaccessible to us we have no hope of measuring $\Pr[a,b,c|x,y,z]_{\omega}$ directly, 
where $z$ denotes a measurement setting on $E$ with outcome $c$. 
Nevertheless, similar to quantum entanglement, it is known that no-signalling distributions are again monogamous~\cite{Ton09} - and it is this fact that allows us to draw conclusions about $E$ by measuring only 
$A$ and $B$.
We will therefore make a non-trivial assumption about the physical theory, namely that no-signalling holds between $A$, $B$ and $E$. We emphasize weaker constraints on the amount of signalling could also lead to a bound - but we are not aware of any other concrete example to consider. 
Mathematically, no signalling means that the marginal distributions obey 
\begin{align}
\forall a,x,y,y',z,z'\ \Pr[a|x,y,z]_{\omega} = \Pr[a|x,y',z']_{\omega}\ ,
\end{align}
that is, the choice of measurement settings $y,y'$ and $z,z'$ does not influence the distribution of the outcomes $a$. 
A set of distributions is no-signalling if such conditions hold for all marginal distributions.

\begin{figure}[htb]
  \begin{center}
    \includegraphics[scale=0.9]{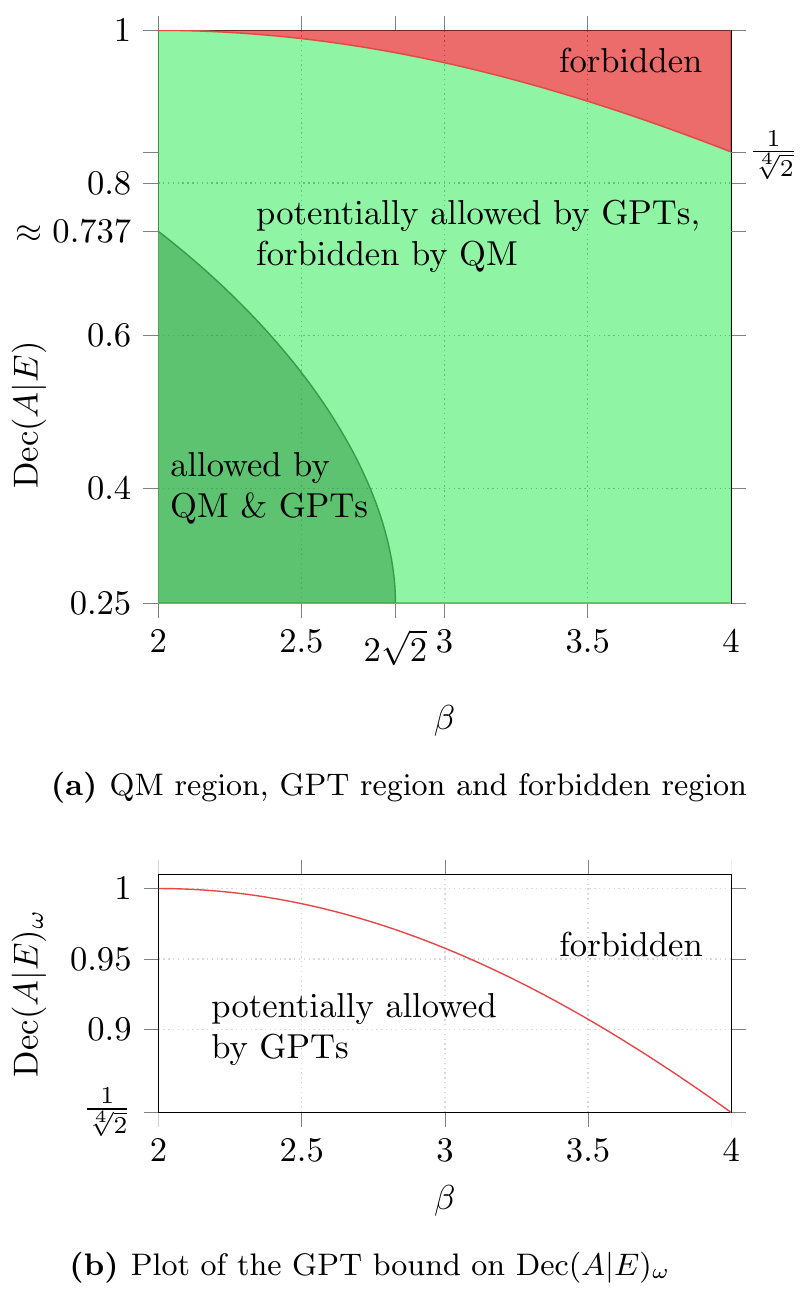}
    \caption{\textbf{Allowed values of the decoherence quantity for measured CHSH values} 
      (colors online).  
      Part \textbf{(a)} shows what values of the decoherence quantity are compatible with 
      some measured CHSH value $\beta$, assuming either quantum theory or any other 
      probabilistic theory. The dark green region consists of all points 
      $(\beta, \Dec(A|E)_\rho)$ for which there exists a quantum state $\rho_{AB}$ and two pairs 
      $(A_0, A_1)$ and $(B_0, B_1)$ of observables with the according values, i.e. the bound 
      is tight. The red region shows pairs $(\beta, \Dec(A|E)_\omega)$ that cannot be 
      realized in any non-signalling probabilistic theory, but not necessarily all of them. 
      In other words, the curve between the light green area and the red area is a bound on 
      $\Dec(A|E)_\omega$ which is valid for all non-signalling probabilistic theories, but 
      the light green area might still contain forbidden pairs. What is certain is that a 
      bound for any specific non-signalling probabilistic theory runs below the red region.
      Part \textbf{(b)} shows a zoomed-in plot of the border line between the forbidden 
      region and the region which is potentially allowed by GPTs, making it easier to see 
      where the line runs for values of $\beta$ close to 2.
      In a world constrained only by no-signalling, $\beta = 4$ is possible~\cite{PR,PR1,PR2,PR3}.
      \label{fig:regions}}
  \end{center}
\end{figure}

\section*{Abstract experiment}

Our method is fully general and can in principle be used to measure the decoherence of any physical process. Figure~\ref{fig:abstract} illustrates the general procedure. We create an entangled pair, and use half of this entangled pair to probe the unknown decoherence process. To estimate $\dec(A|E)$ we will make use of the fact that in QM entanglement is monogamous, or more generally - when considering theories beyond QM - that no-signalling correlations are monogamous. This allow us to make statements about the correlations between $A$ and $E$, even though we can only perform measurements on $A$ and $B$. A test that allows us to bound $\dec(A|E)$ from observations made on $A$ and $B$ alone is given by a Bell inequality~\cite{Bell64,survey}. For the purpose of illustration, we consider creating an entangled state $\Phi_{AA'}$ and perform a test based on the CHSH inequaltity~\cite{CHSH69} (see Figure~\ref{fig:abstract}). We emphasize that our methods are fully general and could be used in conjunction with other inequalities and higher dimensional entangled states.

As an easy warmup, let us first again consider what happens in quantum mechanics. For now, we assume that the measurement devices have no memory. That is, the experiment behaves
the same in each round, independent on the previous measurements. It is relatively straight forward to obtain an upper bound on $\dec(A|E)$ 
by extending techniques from quantum key distribution (QKD)~\cite{ABGM07}.
In essence, we maximize $\dec(A|E)$ over all states that are consistent with the observed CHSH correlator $\beta$ (see Figure~\ref{fig:abstract}). This maximization problem
is simplified by the inherent symmetries of the CHSH inequality, allowing us to reduce this optimization problem to consider only states that are diagonal in the Bell basis.
We proceed to establish properties of min and max entropies for Bell diagonal states, leading to an upper bound.
Concretely, we show in the Appendix (Theorem~\ref{thm:feasible-region}) that
\begin{align}
\dec(A|E) \leq h(\beta)\ ,
\end{align}
where $h(\beta)$ is an easy optimization problem that can be solved using Lagrange multipliers. We have chosen not to weaken this bound by an analytical bound that is strictly larger, as it is indeed easily evaluated (see Figure~\ref{fig:regions}).
If the devices are allowed memory, then a variant of this test and some more sophisticated techniques from QKD nevertheless can nevertheless be shown to give a bound.

\begin{figure}[h!] 
   \includegraphics[scale =0.7]{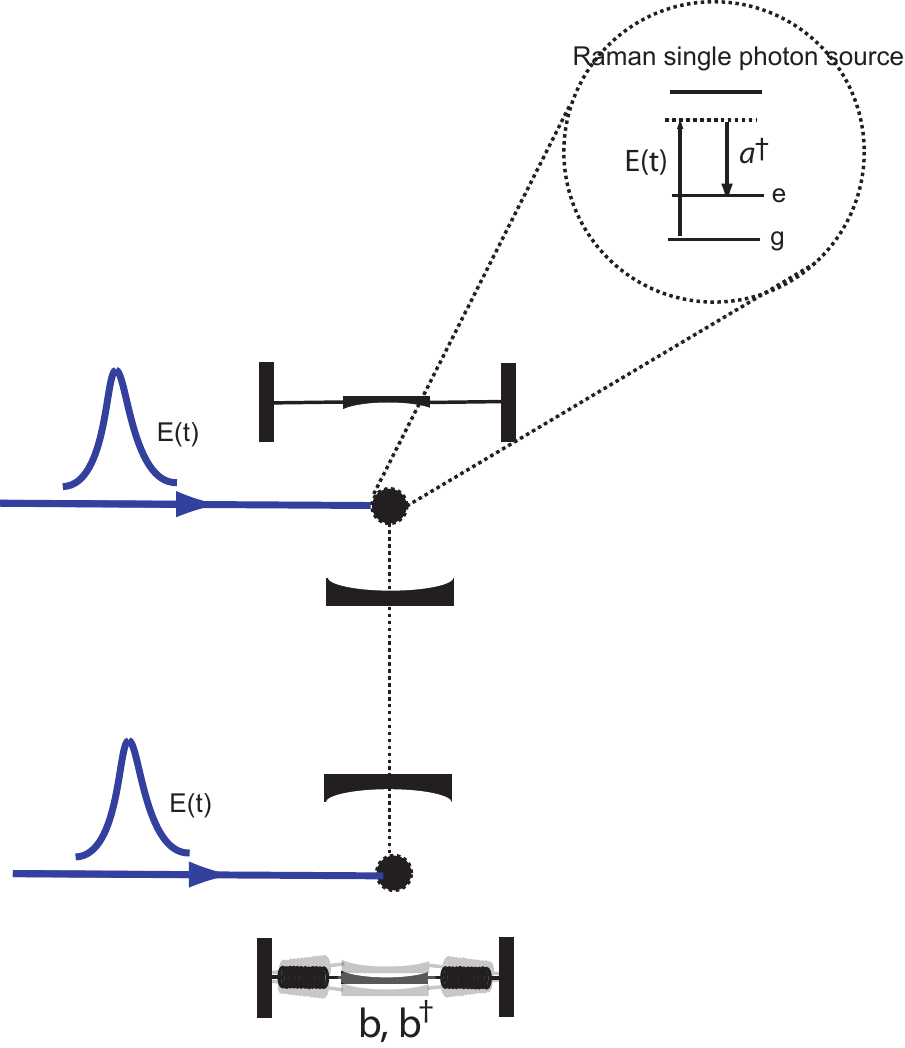} 
      \centering
   \caption{{\textbf{Probing an optomechanical system.}} Our goal is to create two entanglement between two opto-mechanical cavities. One cavitiy thereby has a movable
mirror that introduces gravitational decoherence. Two cavities each contain a Raman single photon source controlled by an external laser `write field' $E(t)$~\cite{NDLK11}. This write-field is used to map excitations in the atomic sources to single photon excitations in the cavities.  The top cavity has fixed end mirrors while the bottom cavity has one mirror that is harmonically bound along the cavity axis and can move in response to the radiation pressure force of light in the cavity. The Raman sources are first prepared in an entangled 
state. This setup is a modification of the one proposed by Bouwmeester \cite{Bouwmeester2003} in which an itinerant single photon pulse is injected into a cavity rather than created intra-cavity as here. Our modification avoids the problem that the time over which the photons interact with the mechanical element is stochastic and determined by the random times at which the photons enter and exit the cavity through an end mirror. In the new scheme, the cavities are assumed to have almost perfect mirrors --- very narrow line width~\cite{NIST} (see appendix for details). 
}
   \label{fig:exp}
\end{figure}

How can we hope to attain an estimate outside of quantum mechanics? Let us first give a very loose intuition, why performing a Bell experiment on $A$ and $B$, may  
allows us to bound $\dec(A|E)_\omega$. It is well known~\cite{Ton09} that non-signalling correlations are also monogamous. That is, if we observe a violation 
of the CHSH inequality as captured by the measured parameter $\beta$, then we know that the violation between $A$ and $E$ and also between $E$ and $B$ must be low. 
Note that the expectation values $\tr[\rho_{AB} (A_x \otimes B_y)]$ in terms of quantum observables $A_x$ and $B_y$ can be expresssed in terms of probabilities as
\begin{align}
&\tr[\rho_{AB}(A_x \otimes B_y)]\nonumber\\
&\qquad = \sum_{a \in \{\pm 1\}} \Pr[a,a|x,y]_{\omega} - \Pr[a,-a|x,y]_{\omega}\ ,
\end{align}
where
we have again used $\omega_{AB}$ in place of $\rho_{AB}$ to remind ourselves that we may be outside of QM.
In fact, if $\beta$ is larger than what a classical theory allows ($\beta > 2$), then $E$ and $B$ cannot violate the CHSH inequality at all. Let us now assume by contradiction
that the state $\omega_{AE}$ shared between $A$ and $E$ would be close to maximally correlated. Then by definition of the maximally correlated state, for every measurement
on $A$, there exists some measurement on $E$ which yields (almost) the same outcome. Hence, if $\omega_{AE}$ would be close to maximally correlated, then we would expect
that $E$ and $B$ can achieve a similar CHSH violation than $A$ and $B$ - because $E$ can make measurements that reproduce the same correlations that $A$ can achieve with $B$.
Yet, we know that this cannot be since CHSH correlations are monogamous.

In the appendix, we make this rough intuition precise. While we do not follow the steps suggested by this intuition, we employ a technique that has also been used for
studying monogamy of CHSH correlations~\cite{Ton09}. Specifically, we use linear programming as a technique to obtain bounds.
We thereby first relate the fidelity to the statistical distance, which is a linear
functional. We are then able to optimize this linear functional over probability distributions $\Pr[a,b,c|x,y,z]_{\omega}$ satisfying linear constraints.
The first such constraint is given by the fact that we consider only no-signalling distribtions. The second by the fact that the marginal distribution $\Pr[a,b|x,y]_{\omega}$
leads to the observed Bell violation $\beta$. The last one stems from the fact that maximal correlations can also be expressed using a linear constraint.
Solving this linear program for an observed violation $\beta$ leads to Figure~\ref{fig:regions}.

\begin{figure*}[htb]
  \begin{center}
    \includegraphics[scale=0.9]{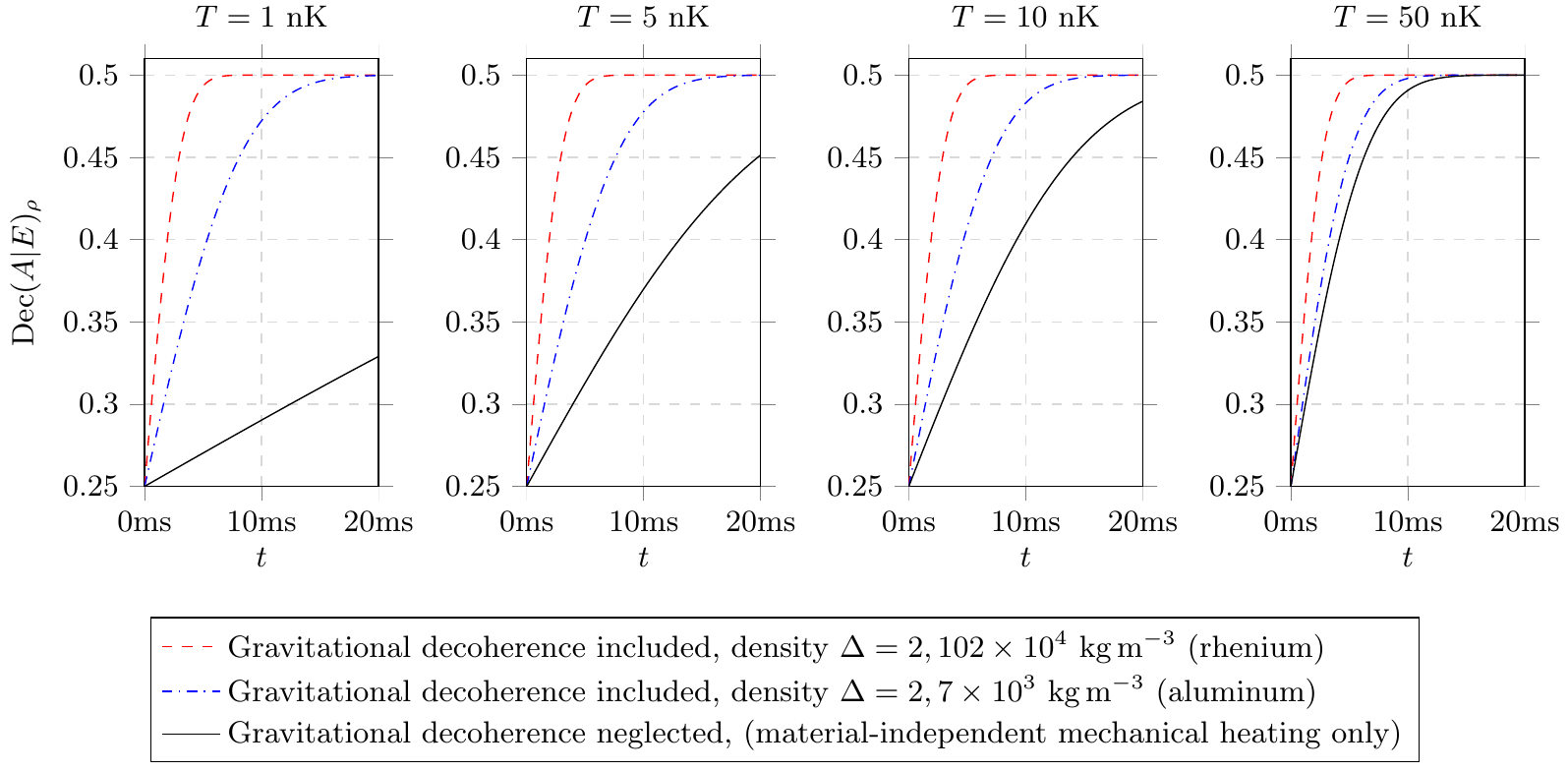}
    \caption{\textbf{Predicted values of the decoherence quantity in the optomechanical 
      experiment.} 
      This figure shows the predicted values of $\Dec(A|E)_\rho$ as a function of the 
      running time of the optomechanical experiment for different temperatures 
      and for different materials of the mechanical element as calculated in the proposed 
      model for gravitational decoherence. In addition, $\Dec(A|E)_\rho$ is plotted for the 
      case where gravitational decoherence is not taken into account. When the gap between 
      the predicted values with and without gravitational decoherence is large enough, the 
      decoherence estimation formalism allows for a test that potentially falsifies the 
      proposed model for gravitational decoherence.
      The calculations have been made for the example experimental parameters $g_0 = 1 \ \text{s}^{-1}$, $\omega_m = 1 \ \text{s}^{-1}$ and $\gamma_m = 10^{-10} \ \text{s}^{-1}$. 
      \label{fig:opto-decs}}
  \end{center}
\end{figure*}

\section*{Optomechanics experiment}

To gain insights into the significance of gravitational decoherence, we examine
Diosi's theory of gravitational decoherence~\cite{Dio89}. This is equivalent to the decoherence model introduced in Kafri et al. \cite{KTM14}.
It can be applied to an optomechanical cavity in which one mirror is free to move in a harmonic potential with frequency $\omega_m$ as in Figure~\ref{fig:exp}. 
The master equation for a massive particle moving in a harmonic potential, including gravitational decoherence is 
\begin{equation}
\frac{d\rho}{dt}=-i\omega_m[b^\dagger b,\rho]-\Lambda[b+b^\dagger,[b+b^\dagger,\rho]]
\end{equation}
where 
\begin{equation}
b=\sqrt{\frac{m\omega_m}{2\hbar}}\hat{x}+i\frac{1}{\sqrt{2\hbar m\omega_m}}\hat{p}
\end{equation}
with $\hat{x},\hat{p}$ the usual canonical position and momentum operators for the moving mirror. 
We have that
\begin{align}
\Lambda = \Lambda_{\rm grav} + \Lambda_{\rm heat}\ ,  
\end{align}
where the gravitational decoherence rate $\Lambda_{\text{grav}}$ is given by 
\begin{equation}
\Lambda_{\text{grav}}=\frac{2\pi}{3} \frac{G\Delta}{\omega_m}
\end{equation}
with $G$ the Newton gravitational constant and $\Delta$ the density of the moving mirror. As one might expect $\Lambda_{\text{grav}}$ is quite small, of the order of $10^{-8}$ s$^{-1}$ for suspended mirrors with $\omega_m\sim 1$. The term 
\begin{align}
\Lambda_{\rm heat} = \frac{k_B T}{\hbar Q}\ ,
\end{align}
with $Q = \omega/\gamma_m$ corresponds to mechanical heating. To see effect of the gravitational term stand out next to the mechanical heating
we thus need to make the temperature $T$ low. 
A calculation shows that this model leads to a dephasing channel $\Gamma(\rho) = p \rho + (1-p) Z \rho Z^\dagger$ where $p$ is a function of the density $\Delta$, 
and the other parameters. In the appendix, we show that for this model 
\begin{widetext}
\begin{align} 
  \Dec(A|E)_\rho = \frac{1}{4} \left( 1+\sqrt{1-\exp\left(-4\left(1+2\left(\frac{4\pi G}{3} 
  \frac{1}{\gamma_m \omega_m} \Delta + \frac{2k_B}{\hbar} \frac{1}{\omega_m} T\right)\right)
  \frac{g_0^2}{\omega_m^2} \sin^2 \left(\frac{\omega_m t}{2}\right) \right)} \right)\ ,
\end{align}
\end{widetext}
where $G$ is the Newton gravitational constant, $k_B$ is the Boltzman constant, and $\hbar$ the Planck constant (see Figure~\ref{fig:opto-decs} for the other parameters).
(see Figure~\ref{fig:opto-decs} for parameters)

\section*{Discussion}

What have we actually learned when performing such an experiment?
We first observe that the measured $\beta$ always gives an upper bound on the amount of decoherence observed - for \emph{any} no-signalling theory.
This means that even if quantum mechanics would indeed need to be modified we can still draw conclusions from the data we obtain. As such, the observations made
in such an experiment establish a fundamental limit on decoherence no matter what the theory might actually look like in detail. 
It is clear, however, that the bound thus obtained is much weaker than if we had assumed quantum mechanics. No-signalling is but one of many 
principles obeyed by quantum mechanics, and these other features put stronger bounds on the values that $\dec(A|E)$ can take. Our motivation
for considering theories which are only constrained by no-signalling is to demonstrate even such weak demands still allow us to draw meaningful conclusions
from such an experiment. One can easily adapt our approach by introducing further constraints on the probabilities $\Pr[a,b,c|x,y,z]$ - but not all of quantum mechanics - 
in order to get stronger bounds. In this case, one can similarly obtain an upper bound on $\dec(A|E)$ from the measured data - this time for the more constrained theory. 
Also in a fully quantum mechanical world, our approach yields to a bound (see Figure~\ref{fig:regions}). If we assume quantum mechanics, we may of course also try and perform
process tomography in order to determine the decoherence process, and indeed any experiment should try and perform such a tomographic analysis whenever possible. The appeal of our
approach is rather that we can draw conclusions from the experimental data while making only very minimal assumptions about the underlying physical theory.

One may wonder, why we only upper bound $\dec(A|E)$. Note that from our experimental statistics we can only make statements about the overall decoherence observed in the experiment, namely the gravitational decoherence (if it exists) as well as any other decoherence introduced due to experimental imperfections. Finding that the Bell violation is low (and thus maybe $\dec(A|E)$ might be large) can thus not be attributed conclusively to the gravitational decoherence process, making a lower bound on $\dec(A|E)$ meaningless if our desire is to make statements about a particular decoherence process such as gravity. 

Second, we observe that our approach can rule out models of gravitational decoherence but not verify a particular one. It is important to note that a model for gravitational decoherence does not stand on its own, but is always part of a theory on what states, evolutions and measurements behave like. Given such a physical theory and a model for gravitational decoherence, we know enough to compute $\dec(A|E)$. In addition, we can compute an upper bound $b(\beta)$ on $\dec(A|E)$ specific to that theory, which may give a much stronger bound 
than no-signalling alone. Indeed, we see from Figure~\ref{fig:regions} that this is the case for quantum mechanics. Given the calculated $\dec(A|E)$ and the experimentally 
observed value for $b(\beta)$, we can then compare: If $\dec(A|E) > b(\beta)$, then the model (or indeed theory) we assumed must be wrong. However, if $\dec(A|E) \leq b(\beta)$, then
we know that the model and theory would be consistent with out experimental observations. We discuss this in more detail in the appendix with a candidate decoherence model that has been proposed and which - if it is valid - may be observed in the experiment suggested above. 

Our approach thus provides a guiding light in the search for gravitational decoherence models. It is very general, and could in principle be used in conjunction with other proposed experimental setups and decoherence models. In particular, it could also be used to probe decoherence models conjectured to arise from decoherence affecting macroscopic objects, where
there exist proposals to bring such objects into superposition~\cite{ignacioGrav}. 
Clearly, however, probing such models using entanglement is extremely challenging. 

It is a very interesting open question to improve our analysis and to apply it to other physical theories that are more constrained than no-signalling, but yet do not quite yield quantum mechanics. Candidates for this may come from the study of generalized probabilistic theories where e.g.~\cite{MM11,MMAPG12,CDP11,DB11,Udu12,PW13} introduced further constraints in order to recover quantum mechanics, but also from suggested ways to modify the Schr{\"o}dinger equation in order to account for non quantum mechanical noise. 
Since our approach could also be applied to higher dimensional systems, and other Bell inequalities, it is a very interesting open question whether other Bell inequalities could be used
to obtain stronger bounds on $\dec(A|E)$ from the resulting experimental observations.

\acknowledgments
We thank Markus P. M\"uller, Matthew Pusey, Tobias Fritz and Gary Steele for insightful discussions.
CP, JK, MT, AM, RS and SW were supported by MOE Tier 3A grant ''Randomness from quantum processes'', NRF CRP ''Space-based QKD''. SW was also supported by QuTech.
NM and GM were supported by ARC Centre of Excellence for Engineered Quantum Systems, CE110001013

\bibliographystyle{apsrev}
\bibliography{measureDeco}

\clearpage
\onecolumngrid
\begin{appendix}
  \numberwithin{equation}{section}
  \numberwithin{thm}{section}

\section*{APPENDIX}
\section*{Conventions}

For this document, we make the following conventions.
\begin{itemize}
  \item The logarithm is with respect to base 2, i.e. $\log \equiv \log_2$.
  \item Hilbert spaces are assumed to be finite-dimensional, unless otherwise stated.
  \item We denote the set of density operators (states) on a Hilbert space $\calH$ by 
    $\calH$.
  \item We identify operators on Hilbert spaces with their reordered versions resulting 
    from permutations of systems. For example, for Hilbert spaces $\calH_A$, $\calH_B$, 
    $\calH_E$ and states $\Phi_{AE} \in \calS(\calH_A \otimes \calH_E)$, $\sigma_B \in 
    \calS(\calH_B)$, we identify the state $\Phi_{AE} \otimes \sigma_B \in \calS(\calH_A
    \otimes \calH_E \otimes \calH_B)$ with the state in $\calS(\calH_A \otimes \calH_B 
    \otimes \calH_E)$ resulting from the application of the braiding map $\calH_A \otimes 
    \calH_E \otimes \calH_B \rightarrow \calH_A \otimes \calH_B \otimes \calH_E$ on $
    \Phi_{AE} \otimes \sigma_B$.
  \item For a state $\rho_{ABE} \in \calS(\calH_A \otimes \calH_B \otimes \calH_E)$, we
    denote its reduced states by according changes of the subscript, e.g. $\rho_A := 
    \tr_B(\rho_{AE})$, $\rho_A := \tr_B(\tr_E(\rho_{ABE}))$.
  \item For a state $\rho_{ABE} \in \calS(\calH_A \otimes \calH_B \otimes \calH_E)$, 
    entropies are evaluated for the according reduced states, e.g. $H(A|B)_\rho$ is the 
    conditional von Neumann entropy of $\rho_{AB} = \tr_E(\rho_{ABE})$ (c.f. 
    \cref{sec:hmin}).
\end{itemize}

\section{Background: Decoherence in quantum theory} \label{sec:background}

In this section, we give a short introduction to decoherence in quantum theory. It consists of concepts, results and quantities that are well-established in quantum information science \cite{Wil13}. The topics are chosen to facilitate the understanding of our contributions in \cref{sec:qm,sec:gpt} rather than to give a full introduction to the subject of decoherence. In \cref{sec:dec-trip}, we describe how the dynamical evolution of a system gives rise to a state of a tripartite system. This tripartite state plays a central role in our later analysis. In \cref{sec:hmin}, we explain why the min-entropy is the relevant quantity in the information theoretic analysis of decoherence. The min-entropy is the quantity that we use for our analysis in \cref{sec:qm}. It is also the quantity that serves as our motivation to define a decoherence quantity for generalized probabilistic theories in \cref{sec:gpt}. We note that
a generalization of quantum theory by, for example, introducing additional terms into the Schr{\"o}dinger equation fall under the regime of generalized theories in our discussion.

\subsection{Dynamical evolution and its tripartite purification} \label{sec:dec-trip}

\textbf{Interaction and non-unitary evolution:} Suppose that a system $S$, initially in a state described by a density operator $\rho_S \in \calS(\calH_S)$, undergoes a dynamical evolution over some time interval. If $S$ undergoes this evolution as a closed system, then according to one of the postulates of quantum mechanics, the state transforms as
\begin{align}
  \rho_S \mapsto U_{S \rightarrow S} \, \rho_S \, U_{S \rightarrow S}^\dagger
\end{align} 
for a unitary $U_{S \rightarrow S}: \calH_S \rightarrow \calH_S$ (see \cref{fig:clopen} (a)). In general, however, the system $S$ may be open, i.e. it may interact with another system $E$ that is called the \emph{environment}. We consider the environment $E$ to consist of all the systems that interact with system $S$. Taken together, the combined system $SE$ then forms a closed system and hence evolves as
\begin{align} \label{eq:se-evolution}
  \rho_S \otimes \rho_E \mapsto U_{SE \rightarrow SE} (\rho_S \otimes \rho_E)
  U_{SE \rightarrow SE}^\dagger \,,
\end{align}
where $\rho_E$ is the initial state of the environment and $U_{SE \rightarrow SE}: \calH_S \otimes \calH_E \rightarrow \calH_S \otimes \calH_E$ is a unitary. 

\begin{figure}[htb]
  \centering
  \includegraphics{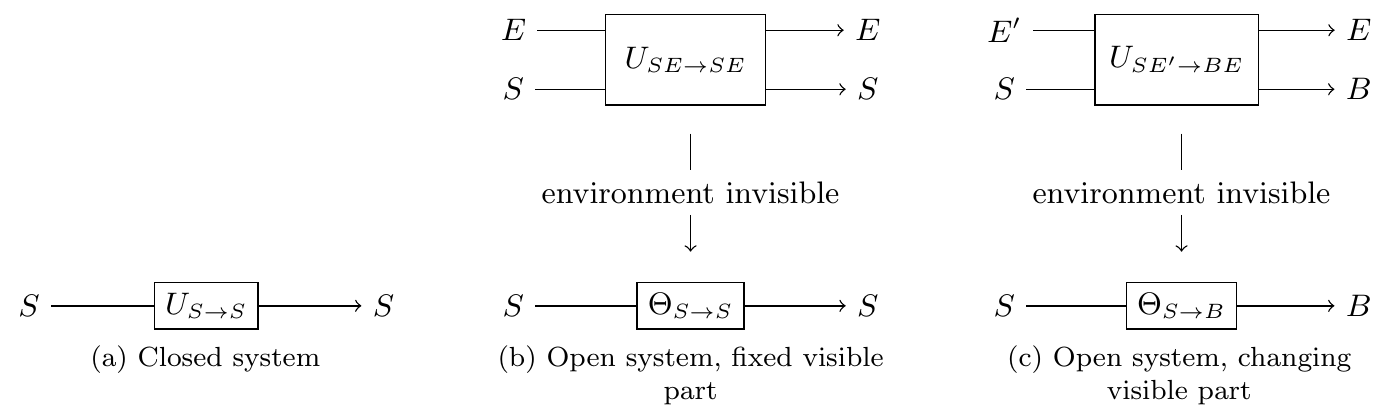}
  \caption{\textbf{Different kinds of dynamical evolution.} (a) If a system $S$ evolves 
  as a closed system, its state transformation is described by a unitary $U_{S \rightarrow 
  S}$. (b) If it evolves in interaction with another system $E$ that is not part of our 
  description, its state transformation is described by a TPCPM 
  $\Theta_{S \rightarrow S}$. (c) If the visible part of the overall system before and 
  after the evolution is not the same, we describe the state transformation by a TPCPM 
  $\Theta_{S \rightarrow B}$.
  \label{fig:clopen}}
\end{figure}

We may be ignorant about the environment $E$ and only have access to system $S$. Our description would then treat the state of the subsystem $S$ after the evolution as a function of the state $\rho_S$ of $S$ before the evolution. We arrive at this description by taking the partial trace over $E$ in expression \eqref{eq:se-evolution}:
\begin{align} \label{eq:traced-se-evolution}
  \rho_S \mapsto \tr_E \left( U_{SE \rightarrow SE} (\rho_S \otimes \rho_E) U_{SE 
  \rightarrow SE}^\dagger \right) =: \Theta_{S \rightarrow S}(\rho_S) \,.
\end{align}
A map $\Theta_{S \rightarrow S}$ of the form \eqref{eq:traced-se-evolution} is easily shown to be a trace-preserving completely positive map (TPCPM). Thus, the evolution of an open system $S$, when the environment $E$ is not visible, is described by a TPCPM $\Theta_{S \rightarrow S}$ (see \cref{fig:clopen} (b)). 

In a yet more general case, it may be that after the evolution of the system $SE$, we do not have access to system $S$ but to a different subsystem $B$ of $SE$. An example would be a two-particle system $S$ interacting with another two-particle system, where we only have access to one particle ($B$) of the four particles after the evolution. Mathematically speaking, the fact that we see a different subsystem before and after the evolution means that our factorization of the overall Hilbert space changes: Before the evolution, we write $\calH = \calH_S \otimes \calH_{E'}$ and after the evolution, we write $\calH = \calH_B \otimes \calH_E$ (see \cref{fig:clopen} (c)). Thus, the unitary evolution of the closed overall system is described by a unitary $U_{SE' \rightarrow BE}: \calH_S \otimes \calH_{E'} \rightarrow \calH_B \otimes \calH_E$. Describing only the accessible part before and after the evolution, we end up with a TPCPM
\begin{align} \label{eq:tre-sb}
  \rho_S \mapsto \tr_E \left( U_{SE' \rightarrow BE} (\rho_S \otimes \rho_{E'}) U_{SE' 
  \rightarrow BE}^\dagger \right) =: \Theta_{S \rightarrow B}(\rho_S) \,.
\end{align}
Thus, the evolution of a system $S$ to a system $B$, when being ignorant about the environment, is described by a TPCPM $\Theta_{S \rightarrow B}$.

\smallskip
\noindent\textbf{Stinespring dilation:} We have demonstrated that unitaries on two (or more) systems give rise to TPCPMs on one system. It is well-known that the converse is also true: Every TPCPM $\Theta_{S \rightarrow B}: \calS(\calH_S) \rightarrow \calS(\calH_B)$ can be extended to a unitary $U_{SE' \rightarrow BE}$ on a larger system in the following sense. For Hilbert spaces $\calH_{E'}$ and $\calH_E$ of appropriate dimensions, it holds that for every pure state $|0\rangle\langle0|_{E'}$ on $\calH_{E'}$, there is a unitary
\begin{align}
  U_{SE' \rightarrow BE}: \calH_S \otimes \calH_{E'} \rightarrow \calH_B \otimes \calH_E
\end{align} 
such that
\begin{align} \label{unit-dil-property}
  \tr_E(U_{SE' \rightarrow BE} (\rho_S \otimes |0\rangle\langle0|_{E'}) U^\dagger_{SE' 
  \rightarrow BE}) 
  = \Theta_{S \rightarrow B}(\rho_S) \quad \text{for all } \rho_S \in \calS(\calH_S) \,.
\end{align}
This (or an equivalent statement) is the \emph{Stinespring dilation theorem} \cite{Sti55}. For more details see \cite{Wil13}.

\smallskip
\noindent\textbf{Textbook definitions of decoherence:} From now on, we will take the viewpoint that the TPCPM $\Theta_{S \rightarrow B}$ is what we are given in the first place. Physically speaking, we assume that we are in the setting where all that we observe is a process in which a system $S$ in some state $\rho_S$ transforms into a state $\rho_B$ of some system $B$. We think of this as a 
\emph{channel} $\Theta_{S \rightarrow S}$, into which we input a system $S$ and get a system $B$ as an output. 
Our goal in \cref{sec:background} is to find a precise mathematical formulation of the following question in the quantum theoretical framework: How much does the channel \emph{decohere} the system? 

For the case where $S = B$, the standard quantum mechanics literature gives some simple descriptions of what the decoherence of a system under a dynamical evolution is. As an example, consider the case where $S$ is a spin-1/2 particle, initially in the spin ``up'' state in the $x$-direction,
\begin{align} \label{eq:x-up}
  \rho_S = |\uparrow_x\rangle\langle \uparrow_x| \,, \quad |\uparrow_x\rangle = \frac{|\uparrow_z\rangle + |\downarrow_z\rangle}{\sqrt{2}} \,.
\end{align} 
If the channel $\Theta_{S \rightarrow S}$ is given by a measurement of the spin in the $z$-direction, then, written in the $z$-basis, the state of the system transforms as
\begin{align} \label{eq:dec-ex}
  \rho_S = \left(\begin{array}{cc}1/2 & 1/2 \\1/2 & 1/2\end{array}\right) \mapsto
  \left(\begin{array}{cc}1/2 & 0 \\0 & 1/2\end{array}\right) 
  = \Theta_{S \rightarrow S}(\rho_S) \,.
\end{align}
One possible observation one can make in \eqref{eq:dec-ex} is that the spin measurement in the $z$-direction causes the off-diagonal terms of the density matrix to vanish. This is an extreme case of the \emph{dephasing channel} in the $z$-basis, which causes a loss of the phase information of the superposition \eqref{eq:x-up}. This loss of phase information is often equated with decoherence. Another feature of \eqref{eq:dec-ex} that is often said to be the characteristic of decoherence is that $\Theta_{S \rightarrow S}$ turns an initially pure state into a mixed state.

These descriptions of decoherence, valid in their own right, are not favored by us for mainly three reasons. Firstly, these are no quantitative measures of decoherence. Secondly, they lack a clear operational meaning. Thirdly, they rely on the quantum mechanical formalism, in which states are expressed as density operators. It is not clear how to express them in more general cases that are not described by quantum theory.

In quantum information science, it is very popular to think of the systems arising in the purified picture we just presented as being controlled by \emph{parties} with intentions and interests rather than just being dead physical objects. We will follow this spirit and from now on use the language of a game and speak of parties Alice, Bob and Eve, that we think of as agents controlling the systems $A$, $B$ and $E$.

\subsection{The min-entropy as a measure for decoherence} \label{sec:hmin} 

\textbf{The coherent information:} As mentioned in \cref{sec:dec-trip}, it has been realized in quantum information science that important quantitative measures of the channel are functions of the state $\rho_{ABE}$ that we described above. One such measure quantifying decoherence is the \emph{coherent information} \cite{SN96}. It is defined in terms of the conditional von Neumann entropy
\begin{align}
  H(A|B)_\rho := H(AB)_\rho - H(B)_\rho \,,
\end{align}
where $H(AB)_\rho = -\tr(\rho_{AB}\log(\rho_{AB}))$ and $H(B) = -\tr(\rho_B \log(\rho_B))$ is the von Neumann entropy of the reduced state $\rho_{AB}$ and $\rho_B$, respectively. The coherent information is defined as
\begin{align}
  I(A \rangle B)_\rho := - H(A|B)_\rho \,.
\end{align}
The coherent information $I(A \rangle B)_\rho$ has been shown to be related to the quantum channel capacity $Q(\Theta_{S \rightarrow B})$ of $\Theta_{S \rightarrow B}$, which is known as the Lloyd-Shor-Devetak (LSD) theorem \cite{Llo97,Sho02,Dev05}. It says that
\begin{align} \label{eq:lsd}
  Q(\Theta_{S \rightarrow B}) = \lim_{n \to \infty} \frac{1}{n} 
  \max_{\rho_{S^n} \in \calS(\calH_S^{\otimes n})} I(A^n \rangle B^n)_\rho \,,
\end{align}
where $I(A^n \rangle B^n)_\rho$ is the coherent information for $\rho_{A^n B^n} = \id_A^{\otimes n} \otimes \Theta_{S \rightarrow B}^{\otimes n}(\rho_{A^n S^n})$ and $\rho_{A^n S^n}$ is a purification of $\rho_{S^n}$. The state $\id_A^{\otimes n} \otimes \Theta_{S \rightarrow B}^{\otimes n}(\rho_{A^n S^n})$ results from the $n$-fold use of the channel $\Theta_{S \rightarrow B}$ to transmit $S^n$, i.e. $n$ copies of system $S$, while the purification $A^n$ of $S^n$ remains unchanged. Thus, the r.h.s. of \eqref{eq:lsd} is the coherent information in the limit of infinitely many channel uses. Likewise, the quantum capacity $Q(\Theta_{S \rightarrow B})$ is the limit of the achievable rate for quantum data transmission in the limit of infinitely many channel uses. One says that the quantum capacity, and therefore the coherent information, is an \emph{asymptotic quantity}. This has the disadvantage that from the coherent information, only very limited statements about finitely many uses of the channel can be made. 

\smallskip
\noindent\textbf{The min-entropy:} More insight about the behavior of the channel under finitely many uses can be gained by considering the corresponding \emph{single-shot} quantity. To formulate it, note that the state $\rho_{ABE}$ is pure, in which case the duality relation $H(A|B)_\rho = - H(A|E)_\rho$ for the conditional von Neumann entropy holds. This gives us
\begin{align}
  I(A \rangle B)_\rho = H(A|E)_\rho \,.
\end{align}
The corresponding single-shot quantity for the conditional von Neumann entropy $H(A|E)_\rho$ is the \emph{conditional min-entropy}, or just \emph{min-entropy}, $H_\text{min}(A|E)_\rho$ \cite{Ren05}. It is defined as
\begin{align} \label{hmin-definition}
  \Hmin(A|E)_\rho = \max_{\sigma_E} \sup\{ \lambda \in \mathbb{R} \mid 
  \rho_{AE} \leq 2^{-\lambda} \id_A \otimes \sigma_E \} \,,
\end{align}
where the maximum is taken over all subnormalized density operators on $\calH_E$, i.e. all positive operators on $\calH_E$ with trace between 0 and 1. 
The min-entropy quantifies the maximal size of a subsystem of $A$ that can be 
decoupled from $E$~\cite{OneShotDecouple}, and thus
tells us how many EPR pairs between Alice and Bob can be created~\cite{lsdDecouple} given a noisy output state $\rho_{AB}$. To obtain the single-shot capacity of $n$ channel uses we are - as in the asymptotic case - allowed to optimze over input states $\rho_{A^nS^n}$. Clearly, however, the resulting expression can be lower bounded using a particlar input state given by $n$ copies of the maximally entangled state. This is the test state we employ here, and hence our test also provides a bound on 
the single shot capacity. For instance if $A$ is a $2$ level system, then the min-entropy readily quantifies the number of EPR-pairs we can recover, given that we started with $n$ EPR pairs as an input. The min-entropy thus has a very appealing operational interpretation.

For our purposes, another expression for the min-entropy is more useful. In the following, we use the symbol $\simeq$ to denote that two Hilbert spaces are isomorphic, i.e. $\calH_A \simeq \calH_{A'}$ means that the two spaces have the same dimension. It has been shown \cite{KRS09} that the min-entropy can be expressed as 
\begin{align} \label{hmin}
  H_{\text{min}}(A|E)_\rho = - \log d_A \max_{\R_{E \rightarrow A'}} F^2(\Phi_{AA'}, \id_A \otimes \R_{E \rightarrow A'}(\rho_{AE})) \,,
\end{align}
where $d_A$ is the dimension of the Hilbert space $\calH_A$ of system $A$, $A'$ is a system with $\calH_{A'} \simeq \calH_A$, the maximization is carried out over all TPCPMs $\R_{E \rightarrow A'}$ from system $E$ to system $A'$, $F(\rho, \sigma) = \tr \sqrt{ \rho^{1/2} \sigma \rho^{1/2} }$ is the fidelity and $\Phi_{AA'}$ is a maximally entangled state on $AA'$, i.e. $\Phi_{AA'}$ is an element of the set
\begin{align} \label{gamma-def}
  \Gamma_{AA'} := \left\{ \Phi_{AA'} \in \calS(\calH_A \otimes \calH_{A'}) 
  \ \middle\vert \ 
  \parbox{0.55\textwidth}{
    There are bases $\{|i\rangle_A\}_i$, $\{|i\rangle_{A'}\}_i$ of $\calH_A$, $
    \calH_{A'}$ such that $\Phi_{AA'} = |\phi\rangle \langle\phi|_{AA'}$ with 
    $|\phi\rangle_{AA'} = \frac{1}{\sqrt{d_A}} \sum_i |i\rangle_A \otimes |i\rangle_{A'}$.
  } \right\} \,.
\end{align} 
The choice of $\Phi_{AA'} \in \Gamma_{AA'}$, i.e. the choice of bases for $\calH_A$ and $\calH_{A'}$, is irrelevant for the value of $\Hmin(A|E)_\rho$. Since every $\Phi_{AA'} = |\phi\rangle\langle\phi|_{AA'} \in \Gamma_{AA'}$ is pure, we have that $F(\Phi_{AA'}, \sigma_{AA'}) = \sqrt{\langle\phi|\sigma|\phi\rangle_{AA'}}$ for any state $\sigma_{AA'}$ on $AA'$.

The expression \eqref{hmin} provides an intuition for the min-entropy. We think of the system $ABE$, which is in the pure state $\rho_{ABE}$, as being distributed between Alice, Bob and Eve. Imagine that Eve tries to perform operations on her share of the system with the intention to bring the reduced state between her and Alice as close as possible to the maximally entangled state $\Phi_{AA'}$, where the square of the fidelity is the measure of closeness. The closer Eve can bring the state to the maximally entangled state, the smaller the min-entropy $\Hmin(A|E)_\rho$. The overall situation of our decoherence analysis is shown in \cref{fig:overall-qm}.

\begin{figure}[h!] 
  \begin{center}
  \includegraphics{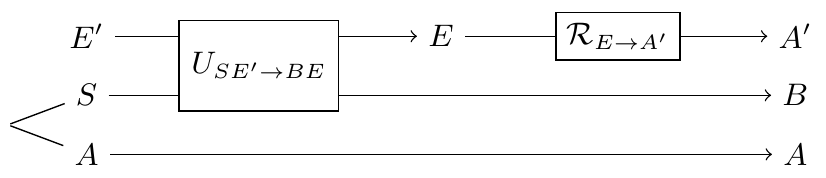}
    \caption{\textbf{Overall picture of our decoherence analysis for quantum 
      theory.} We interpret $\Hmin(A|E)_\rho$ as a measure for how close Eve can get to 
      maximal entanglement with Alice by applying a transformation $\R_{E \rightarrow 
      A'}$ on her share of the state $\rho_{ABE}$.
      \label{fig:overall-qm}}
  \end{center}
\end{figure}

The min-entropy is strictly more informative than the conditional von Neumann entropy in the following sense. In the iid limit (which stands for \emph{independent and identically distributed}), where many identically prepared systems go through the channel and end up in a state $\rho_{AB}^{\otimes n}$, the min-entropy converges to the conditional von Neumann entropy:
\begin{align}
  \lim_{n \to \infty} \frac{1}{n} \Hmin^\epsilon(A^n|B^n)_{\rho^{\otimes n}} = H(A|B)_\rho \ ,
\end{align}
where $\epsilon >  0$ is an arbtirary smoothing parameter.
This is known as the \emph{asymptotic equipartition property} \cite{TCR09}. Thus, in the limit of infinitely many channel uses, where the asymptotic quantity is relevant, the min-entropy reproduces the conditional von Neumann entropy. 

To gain some intuition for $\Hmin(A|E)_\rho$, we now have a look at some special cases. For these special cases, we assume that $\calH_S \simeq \calH_A \simeq \calH_B \simeq \calH_E$. Assume that initially, the state $\rho_{AS}$ is maximally entangled, i.e. $\rho_{AS} = \Phi_{AS}$ for some $\Phi_{AS} \in \Gamma_{AS}$ analogous to \eqref{gamma-def}. 
We think of the channel purification $U_{SE' \rightarrow BE}$ as being controlled by Eve. 
\begin{itemize}
  \item If the adversary Eve leaves system $S$ untouched, i.e. the channel $\Theta_{S 
    \rightarrow B}$ is the identity channel (or any other unitary channel), then 
    $\rho_{AE} = d_A^{-1} \id_A \otimes \rho_E$ for some state $\rho_E$ of system $E$. In 
    that case, $\Hmin(A|E)_\rho = \log d_A$, and we say that there is no decoherence.
  \item In the other extreme case, Eve snatches away the system $S$ and forwards an 
    uncorrelated system to Bob. In this case, $\rho_{AE} \in \Gamma_{AE}$ with 
    $\Gamma_{AE}$ analogous to \cref{gamma-def} (maximal entanglement between $A$ and 
    $E$). Then, $\Hmin(A|E)_\rho = - \log d_A$, and we say that we have full decoherence.
  \item As an intermediate case, we might consider the case where Eve interferes such 
    that she does not end up with maximal entanglement with Alice but such that she is 
    classically correlated with Alice in some basis, i.e. $\rho_{AE} = d_A^{-1} \sum_k
    |k\rangle\langle k|_A \otimes |k\rangle\langle k|_E$. In that case, $\Hmin(A|E)_\rho
    = 0$, and we speak of partial decoherence.
\end{itemize}

\section{Decoherence estimation through CHSH tests in quantum theory} \label{sec:qm}

\subsection{Introduction}

Our goal is to show that Alice and Bob can estimate the decoherence by performing a Bell experiment. We pose it as a feasibility problem: \emph{is it possible to observe certain statistics in a Bell experiment given a certain level of decoherence}? Solving this problem allows us to determine and plot the \emph{feasible region} in the space of suitably chosen parameters.

We look at the simplest Bell experiment, known as the Clauser-Horne-Shimony-Holt (CHSH) \cite{CHSH69} scenario. If $\rho_{AB}$ is the state that Alice and Bob share and $A_{j}, B_{k}$ for $j, k \in \{0, 1\}$ are the observables they perform, then the CHSH value equals
\begin{equation}
\label{chsh-expression}
\beta = \tr \big( [A_{0} \otimes B_{0} + A_{0} \otimes B_{1} + A_{1} \otimes B_{0} - A_{1} \otimes B_{1}] \rho_{AB} \big).
\end{equation}

As explained previously the min-entropy $\Hmin(A | E)$ defined in Eq.~\eqref{hmin-definition} captures the notion of decoherence between Alice and Bob (although note that high min-entropy corresponds to low decoherence and vice versa). Since the range of values that the min-entropy takes depends on the dimension of Alice's system (denoted by $d_{A}$), it is only meaningful to compare scenarios in which $d_{A}$ is fixed. For simplicity, we consider the simplest non-trivial scenario in which the subsystems held by Alice and Bob are qubits, $d_{A} = d_{B} = 2$.

We define the feasible region $\calS$ as follows. A pair of real numbers $(u, v)$, where $u \in [-1, 1]$ and $v \in [0, 2 \sqrt{2} ]$ belongs to $\calS$ if there exists a tripartite state $\rho_{ABE}$ and binary observables $A_{0}, A_{1}$ on $\calH_{A}$ and $B_{0}, B_{1}$ on $\calH_{B}$ such that
\begin{itemize}
\item subsystems $A$ and $B$ are qubits: $\dim \calH_{A} = \dim \calH_{B} = 2$
\item The conditional min-entropy of $A$ given $E$ equals $u$: $\Hmin(A | E) = u$.
\item The CHSH value given by Eq.~\eqref{chsh-expression} equals $v$: $\beta = v$.
\end{itemize}
First note that a CHSH value of $v \leq 2$ can be achieved using trivial measurements (namely $\{\id,0\}$) acting on an arbitrary state. Therefore, for $v \leq 2$ all values of $u \in [-1, 1]$ are allowed. For the remainder of the argument we implicitly assume that $v > 2$ and the following intuitive argument shows why certain pairs $(u, v)$ must indeed be forbidden. Consider a point $u \approx -1$ and $v > 2$. According to the operational meaning of the min-entropy \eqref{hmin}, $u \approx -1$ means that Eve can recover the maximally entangled state with Alice with fidelity close to unity, which clearly allows Alice and Eve to violate the CHSH inequality. On the other hand, since $v > 2$ Alice also observes a CHSH violation with Bob. This violates the monogamy relation for tripartite three-qubit states proved in Ref.~\cite{scarani01}, which states that Alice can violate the CHSH inequality with at most one party (even if she is allowed to use different measurements for different scenarios). This simple argument leads to the conclusion that the region $u \approx -1$ and $v > 2$ is forbidden. In the remainder of this section we show that the non-trivial part of the feasible region $\calS$ can be fully characterised by a single inequality.
\begin{thm}
\label{thm:feasible-region}
A pair of real numbers $(u, v)$ where $u \in [-1, 1]$ and $v \in (2, 2 \sqrt{2} ]$ belongs to the feasible region $\calS$ if and only if
\begin{equation}
u \geq f(v),
\end{equation}
where
\begin{equation}
\label{eq:f-definition}
f(v) := 3 - 2 \log \max_{c_{z}} \bigg( \, 2 \sqrt{ 1 + c_{z}} + \sqrt{ 1 - c_{z} + \frac{v}{\sqrt{2}} } + \sqrt{ 1 - c_{z} - \frac{v}{\sqrt{2}} } \, \bigg),
\end{equation}
where the maximization is taken over
\begin{equation}
-1 \leq c_{z} \leq 1 - \frac{v}{\sqrt{2}}.
\end{equation}
\end{thm}
\noindent While the definition of $f$ might seem complicated, it is straightforward to see that $f$ is monotonically increasing in $v$ and evaluating $f(v)$ numerically for a particular value of $v$ is straightforward since the function to be maximized is concave. The feasible region $\calS$ is plotted below.

\begin{figure}
  \begin{center}
    \includegraphics{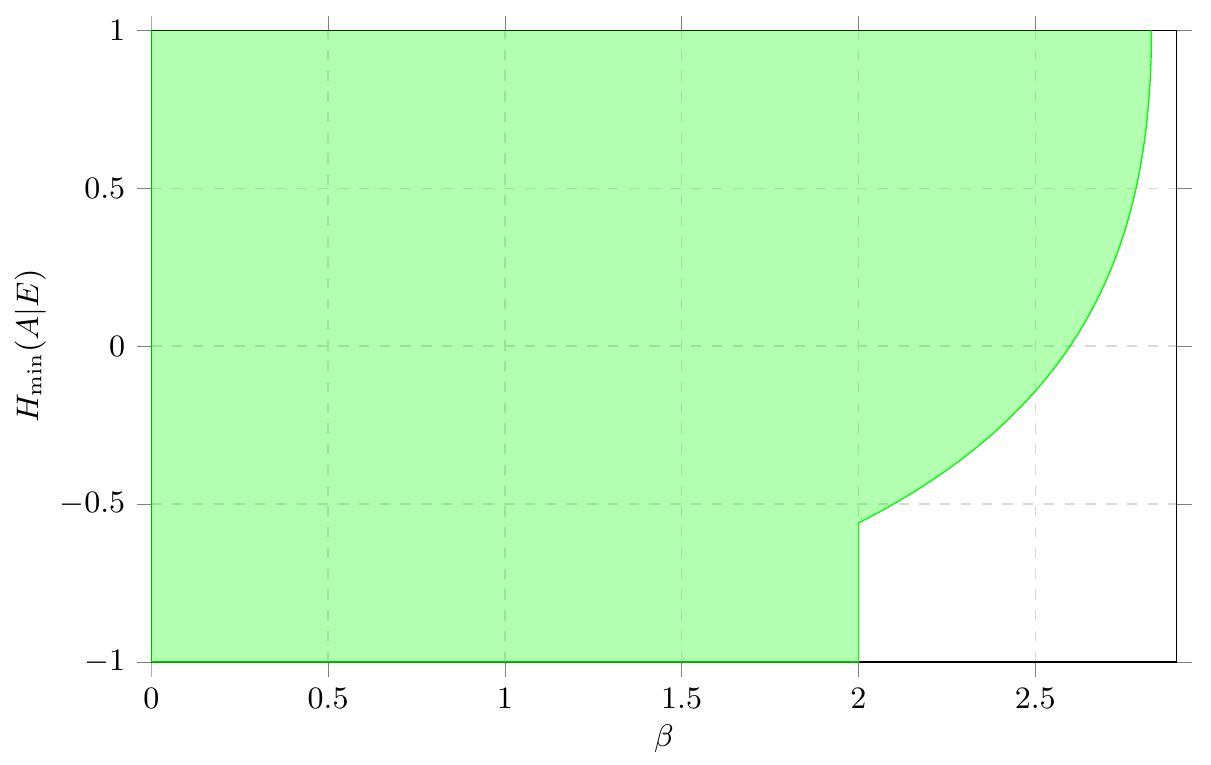}
  \end{center}
\end{figure}

\noindent The proof of \cref{thm:feasible-region} is conceptually simple, but it requires a wide array of technical tools, which we present in \cref{sec:preliminaries}. In \cref{sec:direct,sec:converse} we prove the direct and converse parts of Theorem~\ref{thm:feasible-region}, respectively.

%
\subsection{Preliminaries}
\label{sec:preliminaries}

\begin{defi}
  Let $\calH_A$, $\calH_B$ be Hilbert spaces of dimension $d$. A \dt{generalized Bell 
  basis} for $\calH_A \otimes \calH_B$ is a set $\{ \ket{\Phi_{j}} \}_{j = 1}^{d^{2}}$
  of $d^{2}$ pure states on $\calH_A \otimes \calH_B$ which satisfy
  \begin{align}
    &\tr_{A} \ketbraq{\Phi_{j}} = \frac{\id_B}{d} \,, \quad \tr_{B} \ketbraq{\Phi_{j}} = 
    \frac{\id_A}{d} \nbox{for $j=1,\ldots,d^2$ and}\\
    &\sum_{j = 1}^{d^{2}} \ketbraq{\Phi_{j}} = \id_A \otimes \id_B.
  \end{align}
  A state $\rho_{AB} \in \calS(\calH_A \otimes \calH_B)$ is called \dt{Bell-diagonal} if 
  it is diagonal in some generalized Bell basis, i.e. if there exists a probability 
  distribution $\{p_{j}\}_{j=1}^{d^2}$ such that
  \begin{align} \label{eq:bell-diagonal-definition}
    \rho_{AB} = \sum_{j = 1}^{d^{2}} p_{j} \ketbraq{\Phi_{j}}.
  \end{align}
\end{defi}

\begin{lemma}
\label{lem:positive-semi-definite}
Let $A \geq 0$ be a positive semi-definite operator, let $\Pi_{A}$ be the projector on its support and let $\ket{b}$ be a normalised vector. Then $A \geq \ketbraq{b}$ iff
\begin{equation}
\Pi_{A} \ket{b} = \ket{b} \nbox{and} \bramatketq{b}{A^{-1}} \leq 1.
\end{equation}
Note that since $A$ might not be invertible, $A^{- 1}$ is only defined on the support of $A$.
\end{lemma}
\subsubsection{Two-qubit states}
\noindent A two-qubit state written in the Pauli basis takes the form
\begin{align}
  \rho_{AB} = \frac{1}{4} \big( \id_A \otimes \id_B + \sum_{j} a_{j} \sigma_{j} \otimes 
  \id_B + \id_A \otimes \sum_{j} b_{j} \sigma_{j} + \sum_{j,k} T_{jk} \sigma_{j} \otimes 
  \sigma_{k} \big),
\end{align}
where all the summations go over $\{x, y, z\}$. It is known that for every state there exists a local unitary $U_{A} \otimes U_{B}$ which diagonalizes the correlation tensor (i.e.~ensures that $T_{jk} = 0$ for $j \neq k$) and since all the properties we consider are invariant under local unitaries we can make this assumption without loss of generality. We denote these diagonal entries $T_{xx}$, $T_{yy}$ and $T_{zz}$ by $c_{x}, c_{y}$ and $c_{z}$, respectively, which simplifies the expression to
\begin{align} \label{eq:two-qubit-state}
  \rho_{AB} = \frac{1}{4} \big( \id_A \otimes \id_B + \sum_{j} a_{j} \sigma_{j} \otimes 
  \id_B + \id_A \otimes \sum_{j} b_{j} \sigma_{j} + \sum_{j} c_{j} \sigma_{j} \otimes 
  \sigma_{j} \big).
\end{align}
Without loss of generality, we assume that $\abs{c_{x}} \geq \abs{c_{y}} \geq \abs{c_{z}}$ and $c_{x}, c_{y} \geq 0$. As shown in Ref.~\cite{horodecki96} every Bell-diagonal state of two qubits (up to local unitaries which, again, we can safely ignore) can be written as
\begin{align} 
  \rho_{AB} = \sum_{j = 1}^{4} p_{j} |\Phi_{j} \rangle \langle \Phi_{j}|,
\end{align}
where $\{p_{j}\}_{j=1}^4$ is a probability distribution and $|\Phi_{1,2} \rangle =\frac{ |00\rangle \pm |11\rangle }{\sqrt{2}}$ and $|\Phi_{3,4} \rangle =\frac{ |01\rangle \pm |10\rangle }{\sqrt{2}}$. It is easy to verify that
\begin{equation}
\label{eq:bell-diagonal-2qbits}
  \rho_{AB} = \frac{1}{4} \big( \id_A \otimes \id_B + \sum_{j} c_{j} \sigma_{j} \otimes 
  \sigma_{j} \big),
\end{equation}
where
\begin{align}
  c_{x} &= p_{1} - p_{2}  + p_{3} - p_{4} \,, \nonumber \\
  c_{y} &= - p_{1} + p_{2}  + p_{3} - p_{4} \,, \label{eq:correlation-coefficients} \\
  c_{z} &= p_{1} + p_{2}  - p_{3} - p_{4} \,. \nonumber
\end{align}

%
\subsubsection{Non-locality}
\begin{defi}
For a bipartite quantum state $\rho_{AB}$ the maximum CHSH value is defined as
\begin{equation}
\betamax(\rho_{AB}) := \max_{A_{0}, A_{1}, B_{0}, B_{1}} \tr \big[ ( A_{0} \otimes B_{0} + A_{0} \otimes B_{1} + A_{1} \otimes B_{0} - A_{1} \otimes B_{1} ) \rho_{AB} \big],
\end{equation}
where the maximisation is taken over all Hermitian, binary observables.
\end{defi}
\noindent Note that for all states $\betamax \geq 2$ and we say that the state violates the CHSH inequality if $\betamax > 2$. It was shown in Ref.~\cite{horodecki95} that if $\rho_{AB}$ is a state of two qubits then the value of $\betamax$ is fully determined by the correlation tensor. Adopting the convention $\abs{c_{x}} \geq \abs{c_{y}} \geq \abs{c_{z}}$ we have
\begin{equation}
\label{eq:maximal-violation}
\betamax(\rho_{AB}) =
\begin{cases}
2 &\nbox{if} c_{x}^{2} + c_{y}^{2} \leq 1,\\
2 \sqrt{c_{x}^{2} + c_{y}^{2}} &\nbox{otherwise.}
\end{cases}
\end{equation}
%
%
\subsubsection{Entropic measures of entanglement}
\noindent To derive a bound on the min-entropy $\Hmin(A|E)_\rho$, we will use a closely related quantity, namely the \emph{max-entropy}.
\begin{defi}
  For a bipartite quantum state $\rho_{AB}$ the \dt{conditional max-entropy} (or just 
  \dt{max-entropy}) is defined as
  \begin{equation}
	\label{eq:hmax-definition}
    \Hmax(A | B) = \max_{\sigma_{B}} \log d_{A} F^2(\rho_{AB}, \pi_{A} \otimes \sigma_{B})
    \,,
  \end{equation}
  where $\pi_{A}$ is the maximally mixed state on $A$ and the maximisation is taken over 
  all states on $B$.
\end{defi}
\noindent The proof uses the following known properties of the min- and max-entropies.

\begin{lemma}[Duality, \cite{KRS09}]
\label{lem:duality}
Let $\rho_{ABC}$ be a tripartite state. Then
\begin{equation*}
\Hmax(A|B)_\rho + \Hmin(A|C)_\rho \geq 0,
\end{equation*}
and the equality holds iff $\rho_{ABC}$ is pure.
\end{lemma}
\begin{lemma}[Data-processing inequality, \cite{Ren05}]
\label{lem:Data Processing Inequality}
For an arbitrary tripartite state $\rho_{ABC}$ we have
\begin{equation}
\Hmax(A|B) \geq \Hmax(A|BC).
\end{equation}
\end{lemma}
\begin{lemma}[Conditioning on classical information, Proposition 4.6 of~\cite{Tom12}]
\label{lem:Concavity of max-entropy}
Let $\rho_{ABK}$ be a tripartite state where $K$ is a classical register:
\begin{equation}
\rho_{ABK} = \sum_{k} p_{k} \, \tau_{AB}^{k} \otimes \ketbraq{k}.
\end{equation}
Then
\begin{equation}
\Hmax(A | BK)_{\rho_{ABK}} = \log \Big( \sum_{k} p_{k} \, 2^{\Hmax(A|B)_{\tau_{AB}^{k}}} \Big).
\end{equation}
\end{lemma}
\noindent Finally, we need an explicit expression for the max-entropy of a Bell-diagonal state. Note that by assumption $d_{A} = d_{B} = d$.
\begin{lemma}
\label{lem:max-entropy-bell-diagonal}
Let $\rho_{AB}$ be a Bell-diagonal state of form~\eqref{eq:bell-diagonal-definition}. Then the conditional max-entropy equals
\begin{equation}
\Hmax(A|B) = - \log d + 2 \log \Big( \sum_{j} \sqrt{p_{j}} \Big).
\end{equation}
\end{lemma}
\noindent To prove \cref{lem:max-entropy-bell-diagonal} we use the fact that the optimization problem which appears in the definition of the max-entropy \eqref{eq:hmax-definition} can be written as a semidefinite program (SDP)~\cite{vitanov13}. More specifically, given $\rho_{AB}$ we have $\Hmax(A | B) = \log \lambda$, where $\lambda$ is the value of the following SDP for $\rho_{ABC}$ being an arbitrary purification of $\rho_{AB}$
\begin{equation*}
\begin{aligned}
\texttt{PRIMAL}: \quad
& {\text{minimize}}
& & \mu \\
& \text{subject to}
& & \mu \id_{B} \geq \text{tr}_{A}(Z_{AB}) \\
&&& Z_{AB} \otimes \id_{C} \geq \rho_{ABC} \\
&&& Z_{AB} \in \calP(\calH_{AB})\\
&&& \mu \geq 0
\end{aligned}
\hspace{1cm}
\begin{aligned}
\texttt{DUAL}: \quad
&{\text{maximize}}
& &  \tr (\rho_{ABC} Y_{ABC})  \\
& \text{subject to}
& & \tr_{C}(Y_{ABC}) \leq \id_{A} \otimes \sigma_{B} \\
&&& \tr \sigma_{B} \leq 1 \\
&&& Y_{ABC} \in \calP(\calH_{ABC}) \\
&&& \sigma_{B} \in \calP(\calH_{B})
\end{aligned}
\end{equation*}
\noindent where $\calP(\calH)$ denotes the set of positive semi-definite operators acting on $\calH$. By providing feasible solutions for the \texttt{PRIMAL} and the \texttt{DUAL} we show that for Bell-diagonal states
\begin{equation}
\lambda = \frac{1}{d} \Big( \sum_{j} \sqrt{p_{j}} \Big)^{2}
\end{equation}
which is precisely the statement of \cref{lem:max-entropy-bell-diagonal}.
\begin{proof}
Let $\rho_{ABC} = \ketbraq{\psi_{ABC}}$ be a purification of $\rho_{AB}$, e.g.
\begin{equation}
| \psi_{ABC} \rangle = \sum_{j} \sqrt{p_{j}} \ket{ \Phi_{j} } \otimes \ket{j}.
\end{equation}
For the \texttt{PRIMAL} consider
\begin{gather}\label{eq:ZAB}
Z_{AB} = \Big( \sum_{j } \sqrt{p_{j}} \Big) \sum_{k} \sqrt{p_{k}} \ketbraq{\Phi_{k}},\\
\mu = \frac{1}{d} \Big( \sum_{j} \sqrt{p_{j}} \Big)^{2} .
\end{gather}
Clearly, $Z_{AB} \geq 0$, $\mu \geq 0$ and since $\tr_{A}(Z_{AB}) = \frac{1}{d} \big( \sum_{j} \sqrt{p_{j}} \big)^{2} \id_{B}$ the first constraint is easy to check. The last inequality we need to check is
\begin{equation} \label{eq: const2}
\Big( \sum_{j } \sqrt{p_{j}} \Big) \sum_{k } \sqrt{p_{k}} \ketbraq{ \Phi_{k} } \otimes \id_{C} \geq \rho_{ABC}.
\end{equation} 
We apply \cref{lem:positive-semi-definite} to $A = Z_{AB} \otimes \id_{C}$ and $\ket{b} = \ket{\psi_{ABC}}$. The projector on the support of $Z_{AB} \otimes \id_{C}$ equals
\begin{equation}
\Pi = \sum_{j : p_{j} > 0} \ketbraq{ \Phi_{j} } \otimes \id_{C}
\end{equation}
and it is easy to verify that $\Pi \ket{\psi_{ABC}} = \ket{\psi_{ABC}}$. Moreover, since $(Z_{ABC})^{-1} = (Z_{AB})^{-1} \otimes \id_{C} $ we have
\begin{gather}
\left(  \sum_{m } \sqrt{p_{m}} \bra{ \Phi_{m} } \otimes \bra{ m } \right) \left( \Big( \sum_{j} \sqrt{p_{j}} \Big)^{-1} \sum_{k : p_{k} > 0} \frac{1}{\sqrt{p_{k}}} \ketbraq{ \Phi_{k} } \otimes \id_{C}  \right) \left( \sum_{n} \sqrt{p_{n}} \ket{ \Phi_{n} } \otimes \ket{ n } \right)\\
= \Big( \sum_{j} \sqrt{p_{j}} \Big)^{-1} \left( \sum_{m} \sqrt{p_{m}} \bra{ \Phi_{m} } \otimes \bra{ m } \right) \left(\sum_{n : p_{n} > 0} \ket{ \Phi_{n} } \otimes \ket{ n } \right) = 1.
\end{gather}
Showing that $Z_{AB}$ and $\mu$ constitute a valid solution to the \texttt{PRIMAL} implies that $\lambda \leq \frac{1}{d} \big( \sum_{j} \sqrt{p_{j}} \big)^{2}$.

\noindent For the \texttt{DUAL} consider
\begin{gather}
Y_{ABC} =  \frac{1}{d} \sum_{jk} |\Phi_{j} \rangle \langle \Phi_{k}| \otimes |j\rangle \langle k|,\\
\sigma_{B} = \frac{\id_{B}}{d}.
\end{gather}
Note that $Y_{ABC}$ is proportional to a rank-1 projector. The first constraint gives
\begin{equation}
\tr_{C}(Y_{ABC}) = \frac{1}{d} \sum_{j} \ketbraq{\Phi_{j}} = \frac{1}{d} \; \id_{A} \otimes \id_{B} = \id_{A} \otimes \sigma_{B}
\end{equation}
and the remaining ones are easily verified to be true. The value of this solution equals $\tr (\rho_{ABC} Y_{ABC}) = \frac{1}{d} \big( \sum_{j} \sqrt{p_{j}} \big)^{2}$ which implies that $\lambda \geq \frac{1}{d} \big( \sum_{j} \sqrt{p_{j}} \big)^{2}$.
\end{proof}

\subsubsection{Sufficiency of considering Bell-diagonal states}
\noindent To prove the converse part of \cref{thm:feasible-region}, we will use the following argument, which is similar in spirit and inspired by the symmetrization argument presented in Ref.~\cite{ABGM07}.
%
%
\begin{lemma}
\label{lem:symmetrisation}
Let $\rho_{AB}$ be an arbitrary state of two qubits. Then, there exists a Bell-diagonal state $\sigma_{AB}$ which satisfies
\begin{equation}
\betamax(\rho_{AB}) = \betamax(\sigma_{AB}) \nbox{and} \Hmax(A|B)_{\sigma} \geq \Hmax(A | B)_{\rho}.
\end{equation}
\end{lemma}
\begin{proof}
We present an explicit construction of $\sigma_{AB}$ which meets the requirements. According to Eq.~\eqref{eq:two-qubit-state}, $\rho_{AB}$ can be written as
\begin{align}
  \rho_{AB} = \frac{1}{4} \big( \id_A \otimes \id_B + \sum_{j} a_{j} \sigma_{j} \otimes 
  \id_B + \id_A \otimes \sum_{j} b_{j} \sigma_{j} + \sum_{j} c_{j} \sigma_{j} \otimes 
  \sigma_{j} \big).
\end{align}
Moreover, consider the following random unitary channel
\begin{equation}
\Lambda(\rho_{AB}) = \frac{1}{4}\sum_{j = 1}^{4} (U_{j} \otimes U_{j}) \rho_{AB} (U_{j}^{\dagger} \otimes U_{j}^{\dagger}),
\end{equation}
where $U_{1} =\id$, $U_{2} = \sigma_{x}$, $U_{3} =\sigma_{y}$ and $U_{4} =\sigma_{z}$. It is easy to verify that for $j \in \{x, y, z\}$
\begin{equation}
\Lambda( \sigma_{j} \otimes \id_B ) = \Lambda( \id_A \otimes \sigma_{j} ) = 0
\end{equation}
because each Pauli operator commutes with identity and itself but anticommutes with the other two unitaries. This implies that $\sigma_{AB} = \Lambda(\rho_{AB})$ is Bell-diagonal. Moreover, one can check that the map preserves the correlation tensor, i.e.~for $j \in \{x, y, z\}$
\begin{equation}
\Lambda( \sigma_{j} \otimes \sigma_{j} ) = \sigma_{j} \otimes \sigma_{j},
\end{equation}
which implies that $\betamax(\rho_{AB}) = \betamax(\sigma_{AB})$. To check the last property consider the following state
\begin{equation}
\sigma_{ABK} = \frac{1}{4}\sum_{j = 1}^{4} (U_{j} \otimes U_{j}) \rho_{AB} (U_{j}^{\dagger} \otimes U_{j}^{\dagger}) \otimes \ketbraq{j}.
\end{equation}
By the data processing inequality, we have $\Hmax(A | B)_{\sigma} \geq \Hmax(A | BK)_{\sigma}$ and by conditioning on classical information we have
\begin{gather}
\Hmax(A | BK)_{\sigma} = \log \Big( \sum_{j = 1}^{4} \frac{1}{4} \cdot 2^{\Hmax(A | B)_{\tau^{j}}} \Big),\\
\nbox{where} \tau_{AB}^{j} = (U_{j} \otimes U_{j}) \rho_{AB} (U_{j}^{\dagger} \otimes U_{j}^{\dagger}).
\end{gather}
Since the max-entropy is invariant under local unitaries we have $\Hmax(A | B)_{\tau^{j}} = \Hmax(A | B)_{\rho}$ for $j \in \{x, y, z\}$ which implies that
\begin{equation}
\Hmax(A | B)_{\sigma} \geq \Hmax(A | BK)_{\sigma} = \Hmax(A | B)_{\rho}. 
\end{equation}
\end{proof}
\noindent The final technical lemma concerns the problem of maximizing the max-entropy of a Bell-diagonal state of two qubits whose maximal CHSH violation is fixed.
\begin{lemma}
\label{lem:hmax-beta-tradeoff}
Let $\rho_{AB}$ be a Bell-diagonal state of two qubits, whose maximal CHSH violation equals $\beta \in (2, 2 \sqrt{2} ]$. Then, the max-entropy of $\rho_{AB}$ satisfies the following inequality
\begin{equation}
\Hmax(A|B) \leq -f(\beta)
%
\end{equation}
for function $f$ defined in Eq.~\eqref{eq:f-definition}. Moreoever, there exists a state which saturates this inequality.
%
%
%
%
%
%
\end{lemma}
%
%
\begin{proof}
According to \cref{lem:max-entropy-bell-diagonal} the max-entropy of a Bell-diagonal state of two qubits equals
\begin{equation}
\Hmax(A|B) = - 1 + 2 \log \Big( \sum_{j = 1}^{4} \sqrt{p_{j}} \Big).
\end{equation}
Here, it is convenient to express the probabilities through the correlation coefficients $c_{x}, c_{y}, c_{z}$. Inverting Eqs.~\eqref{eq:correlation-coefficients} gives
\begin{gather}
p_{1} = \frac{1}{4} ( 1 + c_{x} - c_{y} + c_{z}), \quad p_{2} = \frac{1}{4} ( 1 - c_{x} + c_{y} + c_{z}),\\
p_{3} = \frac{1}{4} ( 1 + c_{x} + c_{y} - c_{z}), \quad p_{4} = \frac{1}{4} ( 1 - c_{x} - c_{y} - c_{z}),
\end{gather}
which allows us to write
\begin{equation}
\label{eq:hmax-g-function}
\Hmax(A|B) = - 3 + 2 \log g(c_{x}, c_{y}, c_{z}),
\end{equation}
where
\begin{equation}
\label{eq:g-definition}
g(c_{x}, c_{y}, c_{z}) = \sqrt{ 1 + c_{x} - c_{y} + c_{z} } + \sqrt{ 1 - c_{x} + c_{y} + c_{z} } + \sqrt{ 1 + c_{x} + c_{y} - c_{z} } + \sqrt{ 1 - c_{x} - c_{y} - c_{z} }.
\end{equation}
%
%
%
%
%
%
In the space of correlation coefficients the feasible set are the triples $(c_{x}, c_{y}, c_{z})$ for which the function $g(c_{x}, c_{y}, c_{z})$ is well-defined (the expressions under the roots must be non-negative). As before, we assume without loss of generality that $\abs{c_{x}} \geq \abs{c_{y}} \geq \abs{c_{z}}$ and $c_{x}, c_{y} \geq 0$.
%
%
%
%
Then, the maximal CHSH violation (we are only interested in states that violate the CHSH inequality) is given by Eq.~\eqref{eq:maximal-violation}
\begin{equation*}
\beta = 2 \sqrt{c_{x}^{2} + c_{y}^{2}}.
\end{equation*}
Since in our case $\beta$ is fixed, the angular parametrisation takes the form
\begin{equation*}
c_{x} = \frac{q}{ \sqrt{2} } \, \sin \Big( \phi + \frac{\pi}{4} \Big) \nbox{and} c_{y} = \frac{q}{ \sqrt{2} } \, \cos \Big( \phi + \frac{\pi}{4} \Big),
\end{equation*}
where $q = \frac{\beta}{\sqrt{2}}$ and $\phi \in [0, \pi/4]$ (which ensures $c_{x} \geq c_{y} \geq 0$). Note that
\begin{gather*}
c_{x} + c_{y} = q \cos \phi,\\
c_{x} - c_{y} = q \sin \phi.
\end{gather*}
It is easy to check that the allowed range of $c_{z}$ is
\begin{equation*}
q \sin \phi - 1 \leq c_{z} \leq 1 - q \cos \phi.
\end{equation*}
Note that we should also impose the condition $\abs{ c_{z} } \leq \abs{ c_{y} }$ but as it turns out the optimal solution will satisfy it even if we do not include it explicitly. To maximize the max-entropy it is sufficient to maximize function $g$ defined in Eq.~\eqref{eq:g-definition}, which in the angular parametrisation equals
\begin{equation}
\label{eq:g-angular}
g(\phi, c_{z}) = \sqrt{ 1 + c_{z} + q \sin \phi } + \sqrt{ 1 + c_{z} - q \sin \phi } + \sqrt{ 1 - c_{z} + q \cos \phi } + \sqrt{ 1 - c_{z} - q \cos \phi },
\end{equation}
over
\begin{equation}
\calR = \big\{ (\phi, c_{z}) : \phi \in [0, \pi/4], \; q \sin \phi - 1 \leq c_{z} \leq  1 - q \cos \phi \big\}.
\end{equation}
The maximum is achieved either in the interior (denoted by $\calR_{\textnormal{int}}$) or at the boundary. Let us start by ruling out the first option. Function $g$ is differentiable everywhere in $\calR_{\textnormal{int}}$ and the partial derivatives are
\begin{gather}
\label{eq:partial-g}
\frac{ \partial g }{\partial c_{z}} = \frac{1}{ 2 \sqrt{ 1 + c_{z} + q \sin \phi } } + \frac{1}{ 2 \sqrt{ 1 + c_{z} - q \sin \phi } } + \frac{-1}{ 2 \sqrt{ 1 - c_{z} + q \cos \phi } } + \frac{-1}{ 2 \sqrt{ 1 - c_{z} - q \cos \phi } },\\
\frac{ \partial g }{\partial \phi} = \frac{q \cos \phi}{ 2 \sqrt{ 1 + c_{z} + q \sin \phi } } + \frac{-q \cos \phi}{ 2 \sqrt{ 1 + c_{z} - q \sin \phi } } + \frac{-q \sin \phi}{ 2 \sqrt{ 1 - c_{z} + q \cos \phi } } + \frac{q \sin \phi}{ 2 \sqrt{ 1 - c_{z} - q \cos \phi } }.
\end{gather}
To prove that there is no maximum in the interior, it suffices to show that there is no $(\phi, c_{z}) \in \calR_{\textnormal{int}}$ such that both derivatives vanish $\frac{ \partial g }{\partial c_{z}} = \frac{ \partial g }{\partial \phi} = 0$. To do this we consider the following linear combination
\begin{equation*}
s(\phi, c_{z}) = 2 \sin \phi \cdot \frac{ \partial g }{\partial c_{z}} + \frac{2}{q} \cdot \frac{ \partial g }{\partial \phi} = \frac{\sin \phi + \cos \phi}{ \sqrt{ 1 + c_{z} + q \sin \phi } } + \frac{\sin \phi - \cos \phi}{ \sqrt{ 1 + c_{z} - q \sin \phi } } + \frac{-2 \sin \phi}{ \sqrt{ 1 - c_{z} + q \cos \phi } }
\end{equation*}
and show that $s(\phi, c_{z}) = 0$ has no solution in $\calR_{\textnormal{int}}$. Since the last term of $s(\phi, c_{z})$ is negative, a necessary condition for $s(\phi, c_{z}) = 0$ is that the sum of the first two terms is non-negative, which is equivalent to
\begin{equation*}
\frac{\sin \phi + \cos \phi}{ \sqrt{ 1 + c_{z} + q \sin \phi } } \geq \frac{\cos \phi - \sin \phi}{ \sqrt{ 1 + c_{z} - q \sin \phi } }.
\end{equation*}
This can be rearranged to give
\begin{equation*}
c_{z} \geq \frac{q}{2 \cos \phi} - 1,
\end{equation*}
which contradicts the second inequality in the definition of $\calR_{\textnormal{int}}$ as shown below.
\begin{gather}
c_{z} \geq \frac{q}{2 \cos \phi} - 1 \nbox{and} 1 - q \cos \phi > c_{z}\\
\implies 1 - q \cos \phi > \frac{q}{2 \cos \phi} - 1 \iff \frac{1}{2 \cos \phi} + \cos \phi < \frac{2}{q}.
\end{gather}
It is easy to check that the left-hand side of the final inequality is always at least $\sqrt{2}$, while the right hand side is always at most $\sqrt{2}$. This proves that the final (strict) inequality is always false, which implies that $s(\phi, c_{z}) = 0$ has no solutions in $\calR_{\textnormal{int}}$ and that $g(\phi, c_{z})$ has no maximum in $\calR_{\textnormal{int}}$.

The boundaries $c_{z} = q \sin \phi - 1$ and $c_{z} = 1 - q \cos \phi$ correspond to one of the expression under the roots being zero. Since the square root function has infinite slope at $0$, such solutions cannot be optimal. Therefore, the maximum must be achieved at the boundary $\phi = 0$. Combining Equations~\eqref{eq:hmax-g-function} and \eqref{eq:g-angular} and setting $\phi = 0$ leads directly to the statement of the lemma.

To show that the solution of the optimization problem satisfies $\abs{c_{z}} \leq \abs{c_{y}}$, it is sufficient to show that for $\phi = 0$ and $c_{z} = - c_{y} = - q / 2$ the partial derivative $\partial g / \partial c_{z}$ is strictly positive.
\end{proof}
\subsection{The direct part}
\label{sec:direct}
\noindent Here, we show (by an explicit construction) that points described by $v \in (2, 2 \sqrt{2}]$ and $f(v) \leq u \leq 1$ are allowed. \cref{lem:hmax-beta-tradeoff} shows that for $v \in (2, 2 \sqrt{2}]$ there exists a Bell-diagonal state of two qubits whose max-entropy equals
\begin{equation}
\Hmax(A|B) = - f(v).
\end{equation}
By duality (\cref{lem:duality}), if $\rho_{ABE}$ is an arbitrary purification, the conditional min-entropy equals
\begin{equation}
\Hmin(A|E) = f(v).
\end{equation}
In this example $u = f(v)$, which corresponds to a point lying precisely on the boundary defined in \cref{thm:feasible-region}. In order to obtain higher values of $u$ (all the way up to $1$), it suffices to apply noise of appropriate strength to subsystem $E$.
\subsection{The converse part}
\label{sec:converse}
\noindent Here, we show that every feasible point $(u, v)$ must satisfy $u \geq f(v)$. Consider a state $\rho_{ABE}$ for which $\Hmin(A | E)_{\rho} = u$ and which for some measurements achieves the CHSH value of $v$. Clearly, $\betamax(\rho_{AB}) \geq v$ and by Lemma~\ref{lem:duality} $\Hmax(A | B)_{\rho} \geq -u$. Applying the symmetrization argument (Lemma~\ref{lem:symmetrisation}) gives rise to a Bell-diagonal state $\sigma_{AB}$ such that $\Hmax(A|B)_{\sigma} \geq -u$ and $\betamax(\sigma_{AB}) \geq v$. By Lemma~\ref{lem:hmax-beta-tradeoff} these quantities must satisfy
\begin{equation}
\Hmax(A|B)_{\sigma} \leq -f \big( \betamax(\sigma_{AB}) \big),
\end{equation}
which implies that
\begin{equation}
u \geq - \Hmax(A|B)_{\sigma} \geq f \big( \betamax(\sigma_{AB}) \big) \geq f(v),
\end{equation}
where the last inequality follows from the fact that $f$ is monotonically increasing.

\section{Decoherence estimation through CHSH tests in GPTs} \label{sec:gpt}

In this section, we are going to develop a framework for decoherence analysis in analogy to \cref{sec:background}, but without assuming that nature is correctly described by quantum theory. Instead, we will work in a framework that makes only minimal assumptions about the probabilistic structure of measurements. This allows to make statements in cases where quantum theory might not be a correct description of nature. 

In \cref{sec:ass}, we define a framework for probabilistic theories that has become a standard one in the literature. Besides defining the core structure in \cref{sec:basic-framework}, we explain in \cref{sec:trip} how we extend this framework to make it suitable for analyzing tripartite states, in a way that allows us to make a decoherence analysis that is analogous to the quantum case. 

In \cref{sec:dec-quant-gpt}, we will define a decoherence quantity $\Dec(A|E)_\omega$ for GPTs as an analogue of the quantum min-entropy $\Hmin(A|E)_\rho$. This will be our quantity of interest for the decoherence analysis for GPTs. We will first motivate an expression for $\Dec(A|E)_\omega$ in \cref{sec:dec-mot}, inspired by expression \eqref{hmin} for the min-entropy in the quantum case. This expression will require us to say what a maximally entangled state in a GPT is. We will define it in \cref{sec:max-ent-gpt}.

\cref{sec:gpt-bounds} is devoted to finding a bound on our decoherence quantity in terms of the CHSH winning probability for Alice and Bob. This is a measurable quantity in the case where the channel is an \emph{iid} (for \emph{independent and identically distributed}) channel, meaning that it behaves identically in repeated uses of the channel without building up correlations amongst systems going through the channel in different uses of it. This is a practically relevant case, giving our bound a practical meaning. This bound allows us to infer non-trivial statements about decoherence from measured data when, apart from the iid assumption, we assume only very little about the behavior of nature. We approach our bound by first bounding our fidelity-based decoherence quantity by a trace distance-based quantity. We will then bound this trace distance-based quantity in terms of the CHSH winning probability for Alice and Bob by a quantity that can be expressed as a linear program.

Finally, in \cref{sec:gpt-results}, we show how our bound can be expressed as a linear program and present the numerical results. This is followed by a discussion of the physical interpretation of our numerical findings.

\subsection{The framework} \label{sec:ass}

\subsubsection{A basic framework for GPTs} \label{sec:basic-framework}

Frameworks for probabilistic theories in which quantum theory and classical theory can be formulated as special cases have already been considered some decades ago \cite{Mac63, Edw70, DL70}. After some period of oblivion, a seminal paper by Hardy \cite{Har01} caused a revival in the interest in such frameworks (see, for example, \cite{MM11, MMAPG12, CDP11, DB11, Udu12, PW13} and references therein). Today, they are generally refered to as frameworks for \emph{generalized probabilistic theories} \cite{Bar07}. 

We formalize our decoherence analysis for GPTs in the \emph{abstract state space} framework \cite{BW09, BBLW08, BGW09, BW11}. It is one rigorous formalization of what a generalized probabilistic theory is, amongst a few equivalent or closely related ones that can be found in the literature (see the references cited above). We prefer it for its concise and precise formulation. For the sake of brevity, we will not go far beyond the mere mathematical definitions related to abstract state spaces here. For a detailed introduction to abstract state spaces, see \cite{Pfi12}.

\begin{defi} \label{def:ass}
  An \dt{abstract state space} is a triple $(V, V^+, u)$, where $V$ is a finite-
  dimensional real vector space, $V^+$ is a cone\footnote{A subset $V^+ \subseteq V$ is a \emph{cone} in $V$ if
  \begin{itemize}
    \item[(C1)] $V^+ + V^+ \subseteq V^+$,
    \item[(C2)] $\alpha V^+ \subseteq V^+$ for all $\alpha \geq 0$,
    \item[(C3)] $V^+ \cap (-V^+) = \{0\}$,
  \end{itemize}
} 
in $V$ which is closed\footnote{We assume the standard topology on $V$, i.e. the only linear Hausdorff topology on $V$.} and generating\footnote{A cone $V^+ \subseteq V$ is \emph{generating} if $V^+ - V^+ = V$.} and $u \in V^*$ is a linear functional\footnote{For a finite-dimensional vector space $V$, we denote by $V^*$ the dual space of $V$, i.e. the vector space of linear functionals on $V$.} on $V$ such that $u(\omega) > 0$ for 
  all $\omega \in V^+ \setminus \{0\}$. The functional $u$ is called the \dt{unit effect}.
\end{defi}

\begin{defi} \label{def:ass-induced}
  For an abstract state space $(V, V^+, u)$, we define the following induced structure
  (see \cref{fig:ass-vis}):
  \newline
  The \dt{normalized states} are the elements of the set
  \begin{align}
    \Omega := \{ \omega \in V^+ \mid u(\omega) = 1 \} \,.
  \end{align}
  The \dt{subnormalized states} are the elements of the set
  \begin{align}
    \Omega^\leq := \{ \omega \in V^+ \mid u(\omega) \leq 1 \} \,.
  \end{align}
  The \dt{effects} are the elements of the set
  \begin{align}
    \mathcal{E} := \{e\in V^*\mid 0 \leq e(\omega) \leq 1 \ \forall \omega \in \Omega \} \,.
  \end{align}
  The \dt{measurements} are the elements of the set
  \begin{align}
    \calM := \left\{ M \subseteq \mathcal{E} \text{ finite} \ \middle\vert \ 
    \sum_{e \in M} e = u \right\} \,.
  \end{align}
\end{defi}

\begin{figure}[h!]
  \centering
    \includegraphics{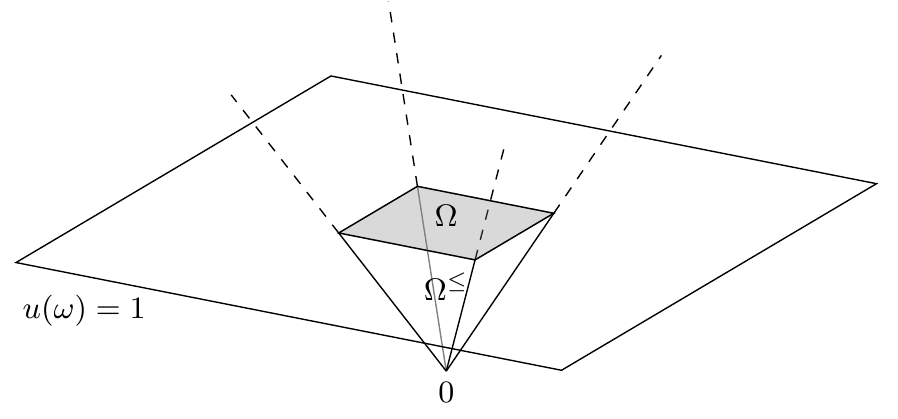}
%
%
%
%
%
  \caption{\textbf{Visualization of an abstract state space.} The set $V^+$ is a cone in a 
  vector space $V$, here shown as a cone with a square base. The normalized states 
  $\Omega$ are given by the intersection of $V^+$ with the plane on which the functional 
  $u$ takes the value 1. The set $\Omega^\leq$ consists of those elements of $V^+$ on 
  which the functional $u$ takes values between 0 and 1.
  \label{fig:ass-vis}}
\end{figure}

An effect respresents a measurement outcome. If a system in a state $\omega$ is measured with respect to a measurement $M = \{e_1, \ldots, e_n\}$, then $e_k(\omega)$ is the probability that the measurement yields the outcome associated with $e_k$. 

\begin{ex}[Quantum theory] \label{ex:qm}
  The probabilistic structure of measurements on a (finite-dimensional) quantum system can 
  be formulated as an abstract state space. For a quantum system with an associated 
  Hilbert space $\calH$, consider the abstract state space $(V, V^+, u) = (\Herm(\calH), 
  \Pos(\calH), \tr)$, where $V = \Herm(\calH)$ is the real vector space of Hermitian 
  operators on $\calH$, $V^+ = \Pos(\calH)$ is the cone of positive operators on $\calH$ 
  and $u = \tr$ is the trace on $\calH$. According to \cref{def:ass-induced}, this yields 
  the states $\Omega = \{ \rho \in \Pos(\calH) \mid \tr(\rho) = 1 \}$, which are precisely 
  the density operators on $\calH$. Analogously, $\Omega^\leq$ are the subnormalized 
  density operators. The effects are the functionals induced by POVM elements via the 
  trace, $\calE = \{ \tr(P \ \cdot \ ) \mid P \in \Pos(\calH), P \leq \id_{\calH} \}$. 
  Accordingly, the measurements are the sets of functionals that are induced by POVMs, 
  $\calM = \{ \{ \tr(P_k \ \cdot \ ) \mid P_k \in \{P_k\}_k \} \mid \{P_k\}_k \text{ is a 
  POVM} \}$. This precisely reproduces the structure of measurement statistics in quantum 
  theory. For further details, see \cite{Pfi12}. \hfill $\blacksquare$
\end{ex}

By our definition, $\calE$ is the set of \emph{all} linear functionals $e$ such that $0 \leq e(\omega) \leq 1$ for all $\omega \in \Omega$. The underlying assumption that every such linear functional represents a physical measurement outcome has been called the \emph{no-restriction hypothesis} \cite{Udu12}. A priori, there seems to be no immediate physical reason for this assumption, and some authors have argued about how to weaken this assumption \cite{JL13}. For our purposes here, it is not relevant whether the no-restriction hypothesis holds, and weakening the assumption complicates the definitions. Thus, we assume it for simplicity.

In \cref{sec:dec-quant-gpt}, we will define a decoherence quantity $\Dec(A|E)_\omega$ analogous to the quantum min-entropy $\Hmin(A|E)_\rho$. We will take our inspiration from expression \eqref{hmin} for the quantum min-entropy, which involves the fidelity as a measure of closeness of quantum states. Therefore, it is desirable to have a generalization of the fidelity to states in abstract state spaces. Such a generalization is easily found once it is noticed that the quantum fidelity of two states can be expressed as the \emph{Bhattacharyya coefficient} (or \emph{classical fidelity}) of the two probability distributions that the two states induce, minimized over all measurements. More precisely, the quantum fidelity satisfies \cite{NC00}
\begin{align} \label{eq:cl-to-q-f}
F(\rho, \sigma) = \min_{\{P_k\}_k} \sum_k \sqrt{\tr(P_k \rho)\tr(P_k \sigma)} \,,
\end{align}
where the minimization runs over all POVMs $\{P_k\}_k$ on the Hilbert space on which $\rho$ and $\sigma$ are defined. The sum in \eqref{eq:cl-to-q-f} is precisely the Bhattacharyya coefficient of the probability distributions that the POVM $\{P_k\}_k$ induces on the states $\rho$ and $\sigma$. This motivates us to define the fidelity for abstract state spaces as follows.

\begin{defi} \label{def:fidelity}
  Let $(V, V^+, u)$ be an abstract state space with normalized states $\Omega$ and 
  measurements $\calM$. For states $\omega, \tau \in \Omega$, we 
  define the \dt{fidelity} of $\omega$ and $\tau$ as
    \begin{align}
      F(\omega, \tau) := \inf_{M \in \calM} b(\omega, \tau | M) \,, \quad \text{where}
      \quad b(\omega, \tau | M) = \sum_{e \in M} \sqrt{e(\omega)} \sqrt{e(\tau)} \,.
    \end{align}
    The quantity $b(\omega, \tau | M)$ is the \dt{Bhattacharyya coefficient} (or 
    sometimes called the classical fidelity) of the probability distributions that the 
    measurement $M$ induces on the states $\omega$ and $\tau$. 
\end{defi}

The fidelity as defined in \cref{def:fidelity} precisely reduces to the quantum fidelity in the case where the abstract state space is a quantum state space. In addition to the fidelity, in \cref{sec:gpt-bounds} we will also consider a generalization of the quantum trace distance $D(\rho, \sigma) = \frac{1}{2} \tr \vert \rho - \sigma \vert$ in order to formulate a bound on $\Dec(A|E)_\omega$. Somewhat analogously to the fidelity, the quantum trace distance is equal to the \emph{total variation distance} (or \emph{classical trace distance}) between the two probability distributions that the two states induce, maximized over all measurements \cite{NC00}:
\begin{align}
  D(\rho, \sigma) = \max_{\{P_k\}_k} \frac{1}{2} \sum_k \vert \tr(P_k\rho) - \tr(P_k
  \sigma) \vert \,.
\end{align}
This motivates the following definition.

\begin{defi}
  Let $(V, V^+, u)$ be an abstract state space with normalized states $\Omega$ and 
  measurements $\calM$. For states $\omega, \tau \in \Omega$, the 
  \dt{trace distance} between $\omega$ and $\tau$ is given by
    \begin{align}
      D(\omega, \tau) := \sup_{M \in \calM} d(\omega, \tau | M) \,,\quad \text{where}
      \quad d(\omega, \tau | M) = \frac{1}{2} \sum_{e \in M} 
      \left\vert e(\omega) - e(\tau) \right\vert
    \end{align}
    The quantity $d(\omega, \tau | M)$ is the \dt{total variation distance} (or 
    sometimes called the classical trace distance) between the probability distributions 
    that the measurement $M$ induces on the states $\omega$ and $\tau$.
\end{defi}

Note that the fidelity and the trace distance take values between 0 and 1 for all states. For squares of the quantities $F$, $b$, $D$ and $d$, we will write the square sign right after the letter, e.g. we will write $F^2(\omega, \tau)$ instead of $(F(\omega, \tau))^2$.

\subsubsection{A tripartite framework for GPTs} \label{sec:trip}

In \cref{sec:dec-quant-gpt}, we will consider a tripartite situation for the decoherence analysis, analogous to \cref{sec:background}. This requires us to model a tripartite scenario mathematically since such a structure is not induced by an abstract state space $(V, V^+, u)$ alone. We need to specify it as additional structure. Our goal here is to do this with the weakest possible assumptions, resulting in a very general validity of the bounds we derive.

Instead of assuming individual state spaces for every party, we only consider their overall combined state space, modelled by an abstract state space $(V, V^+, u)$ and all its induced structure as in \cref{def:ass,def:ass-induced}. This has the advantage that we do not have to make assumptions about how individual state spaces combine to multipartite state spaces, keeping our assumptions weak.
For our purposes, the only structure that we need to add to an abstract state space $(V, V^+, u)$ to make it suitable for the description of a tripartite scenario are the local transformations that each individual party can perform. The local measurements of the three parties are then induced by these local transformations.

We consider three parties, which we call Alice ($A$), Bob ($B$) and Eve ($E$) as before. We begin our considerations by assuming that there are three sets $\calT_A$, $\calT_B$ and $\calT_E$, containing all the transformations that Alice, Bob and Eve can perform, respectively. By a transformation, we mean a linear map $T: V \rightarrow V$ which maps states to subnormalized states, i.e. $T(\Omega) \subseteq \Omega^\leq$ (or, equivalently, $T(V^+) \subseteq V^+$ and $(u \circ T)(\omega) \leq u(\omega)$ for all $\omega \in V^+$). We can consider the case where several transformations are applied because compositions of transformations are transformations again: If $T$, $T'$ are linear maps $V \rightarrow V$ which map $\Omega$ inside $\Omega^\leq$, then the same is true for the composition $T \circ T'$ (we denote the composition of maps by a $\circ$ symbol). 

We assume that the three parties act individually at spatially separated locations. Relativistic considerations lead to the consistency requirement that transformations performed by different parties must commute, e.g. if Alice performs a transformation $T_A \in \calT_A$ and Bob performs a transformation $T_B \in \calT_B$, then the total transformation must satisfy $T_A \circ T_B = T_B \circ T_A$. 

For our purposes, we do not need to specify the sets $\calT_A$, $\calT_B$ and $\calT_E$ any further; the only requirement is that transformations of distinct parties commute. The sets $\calT_A$, $\calT_B$ and $\calT_E$ \emph{define} the systems $A$, $B$ and $E$, i.e. we define the individual parties via the transformations that they can perform. This leads us to the following definition.

\begin{defi} \label{def:tripartite}
  A \dt{tripartite scenario} is a quadruplet 
  \begin{align}
    S_{ABE} = ((V, V^+, u), \calT_A, \calT_B, \calT_E) \,,
  \end{align} 
  where $(V, V^+, u)$ is an abstract state space, and where
  \begin{align}
    \calT_A, \calT_B, \calT_E \subseteq \left\{ T: V \rightarrow V \text{ linear} 
    \ \middle\vert \ T(\Omega) \subseteq \Omega^\leq \right\}
  \end{align}
  are such that for all $P, P' \in \{A, B, E\}$ with $P \neq P'$, it holds that 
  $T_P \circ T_{P'} = T_{P'} \circ T_P$ for all $T_P \in \calT_P$ and for all $T_{P'} \in 
  \calT_{P'}$. We call the elements of $\calT_A$, $\calT_B$ and $\calT_E$ 
  the \dt{local transformations} of $A$, $B$ and $E$, respectively.
\end{defi}

It is absolutely natural to define tripartite scenarios via commuting transformations rather than via a tensor product structure. In quantum theory, the two approaches are equivalent in finite dimensions (we will talk about this below). In more general infinite-dimensional cases, where it is not known whether the two approaches are equivalent, things are usually formalized in a commutative way rather than via tensor products (see \cite{SW87}, for example). Knowing about the equivalence in finite dimensions, we will formulate some quantum examples in the tensor product structure below.

\begin{ex}[A tripartite quantum scenario] \label{ex:qm-trip}
  One can formulate a tripartite situation in quantum theory as a \emph{tripartite 
  scenario.} Based on \cref{ex:qm}, consider the tripartite scenario
  \begin{align}
    &((\Herm(\calH), \Pos(\calH), \tr),\calT_A, \calT_B, \calT_E) \,, 
    \quad \text{where } \\
    &\calH = \calH_A \otimes \calH_B \otimes \calH_E \,, \\
    &\calT_A = \{ \R_A \otimes \id_B \otimes \id_E 
    \mid \R_A \text{ is a trace non-increasing CPM on } \Herm(\calH_A) \} \,, \\
    &\calT_B = \{ \id_A \otimes \R_B \otimes \id_E 
    \mid \R_B \text{ is a trace non-increasing CPM on } \Herm(\calH_B) \} \,, \\
    &\calT_E = \{ \id_A \otimes \id_B \otimes \R_E 
    \mid \R_E \text{ is a trace non-increasing CPM on } \Herm(\calH_E) \} \,,
  \end{align}
  where CPM stands for \emph{completely positive map}. Having tensor product form, the local 
  transformations of different parties commute. \hfill $\blacksquare$
\end{ex}

For our purposes, \cref{def:tripartite} is all the structure one needs to specify. The local measurements are induced by the local measurements. We formalize this via the noation of a local \emph{instrument} \cite{DL70}. To get an intuition for what an instrument is, consider a Stern-Gerlach experiment. A spin-1/2 particle enters a magnet and undergoes one of two transformations: It either gets deflected upwards or downwards. Which of the two transformations it undergoes is determined probabilistically. Then it hits a screen, which reveals which of the two transformations the particle has undergone. This way, a measurement has been performed in two stages: a probabilistic application of a transformation and a detection. The sum of the probabilities of detecting the particle at the top or the bottom of the screen is one. If the state of the particle is described by a state $\omega \in \Omega$ of an abstract state space, we may model this by a set of two transformations $\{ T_{\text{up}}, T_{\text{down}} \}$. Such a set is an instrument. The norm $u(T_{\text{up}}(\omega))$ is the probability that the particle is deflected upwards, and likewise for $u(T_{\text{down}}(\omega))$. Thus, $u$ can be seen to play the role of the screen, detecting the particle. The requirement that the particle must undergo one of the two deflections reads $u \circ T_{\text{up}} + u \circ T_{\text{down}} = u$. The transformation $T_{\text{up}}$ is the analogue of the transformation $\rho \mapsto P_{\text{up}} \rho P_{\text{up}}$ in quantum theory, where $P_{\text{up}}$ is the projector onto the spin-up state. Since $u$ is given by the trace in quantum theory, the probability for the upward-deflection to occur is given by $\tr(P_{\text{up}}\rho P_{\text{up}}) = \tr(P_{\text{up}}\rho)$, which is precisely the Born rule.

A \emph{local} instrument is such a set of transformations where all the transformations are the \emph{local} transformations of one party. This motivates the following definition.

\begin{defi}[Local instruments] \label{def:trip-ind}
  For a tripartite scenario $S_{ABE} = ((V, V^+, u), \calT_A, \calT_B, \calT_E)$ with 
  $\Omega$ as defined in \cref{def:ass-induced}, we define the \dt{local instruments} as 
  the elements of
  \begin{align}
    \calI_P := \left\{ I_P \subseteq \calT_P \text{ finite} \ \middle\vert \  
    \sum_{T_P \in I_P} u \circ T_P = u \right\} \quad \text{for } P \in \{A, B, E\} \,.
\end{align}
\end{defi}

\begin{ex}[Local instruments in a tripartite quantum scenario] \label{ex:qm-trip-2}
  Considering the tripartite scenario of \cref{ex:qm-trip}, we get that the local 
  instruments are given by
  \begin{align}
    \calI_P = \left\{ I_P \subseteq \calT_P \ \middle \vert \ 
    \sum_{T_P \in I_P} T_P \text{ is a TPCPM} \right\} \quad \text{for } P\in \{A,B,E\}\,.
    \tag*{$\blacksquare$}
  \end{align}
\end{ex}

\begin{rmk}[Local measurements] \label{rmk:loc-meas}
  The definition of local instruments gives us a notion of local measurements as well. 
  Consider a tripartite scenario $S_{ABE} = ((V, V^+, u), \calT_A, \calT_B, \calT_E)$ with 
  its set of measurements $\calM$. It is easily verified that for a local transformation 
  $T_A \in \calT_A$, the map $u \circ T_A$ is an effect (as defined in 
  \cref{def:ass-induced}). Likewise, for a local instrument $I_A \in \calI_A$, the set 
  $\{ u \circ T_A \mid T_A \in I_A \}$ is a measurement. We interpret it as a measurement 
  performed by Alice. We can also consider composite measurements where several parties 
  locally perform measurements. For local instruments $I_A \in \calI_A$ and $I_B \in 
  \calI_B$, for example, the set $\{u \circ T_A \circ T_B \mid T_A \in I_A, T_B \in I_B\}$ 
  is a measurement. We interpret it as a composite measurement where Alice and Bob each 
  perform local measurements, described by $I_A$ and $I_B$. The analogous holds for other 
  parties and combinations thereof.
\end{rmk}

\begin{ex}[Local measurements in a tripartite quantum scenario] \label{ex:qm-trip-3}
  Based on \cref{ex:qm-trip,ex:qm-trip-2}, we can say how local measurements look like in 
  a tripartite quantum scenario. A local effect of Alice is of the form
  \begin{align}
    \rho_{ABE} \mapsto \tr(\R_A \otimes \id_B \otimes \id_E (\rho_{ABE}))
  \end{align}
  for a trace non-increasing CPM $\R_A$ on $\Herm(\calH_A)$. However, for every such 
  CPM, there is a POVM element $P_A$ on $\calH_A$ such that\footnote{This can be seen from
  the Kraus representation of $\R_A$: $\tr(\R_A(\rho_A)) =
  \tr(\sum_k F_k \rho_A F_k^\dagger) = \tr(\sum_k F_k^\dagger F_k \rho_A) = 
  \tr(P_A \rho_A)$ for $P_A = \sum_k F_k^\dagger F_k$. (We omitted the other tensor 
  factors for brevity.)}
  \begin{align}
    \tr((P_A \otimes \id_B \otimes \id_E) \rho_{ABE}) = \tr(\R_A \otimes \id_B 
    \otimes \id_E (\rho_{ABE})) \,.
  \end{align}
  This recovers the Born rule. Analogously, a composite measurement where Alice and Bob 
  each perform local measurements consists of local effects of the form
  \begin{align}
    \rho_{ABE} \mapsto \tr(\R_A \otimes \R_B \otimes \id_E (\rho_{ABE}))
    = \tr((P_A \otimes P_B \otimes \id_E)\rho_{ABE})
  \end{align}
  for POVM elements $P_A$, $P_B$ on $\calH_A$, $\calH_B$. Thus, in our tripartite quantum 
  example, local measurements reduce to POVM measurements of product form. \hfill 
  $\blacksquare$
\end{ex}

In \cref{ex:qm-trip,ex:qm-trip-2,ex:qm-trip-3}, instead of choosing a tensor factorization for $\calH$ and setting the local transformations to be acting non-trivially on one tensor factor, we could have chosen sets of transformations that merely commute, without a tensor product structure. The question of whether the resulting measurement statistics in that case would be different from the case with the tensor factor structure is known as \emph{Tsirelson's problem} \cite{SW08, DLTW08}. More precisely, the question is the following. Let $\calH$ be a Hilbert space, let $\rho$ be a density operator on $\calH$, let $\{P_k\}_k$, $\{Q_l\}_l$ be POVMs on $\calH$ such that $P_k Q_l = Q_l P_k$ for all $k$, $l$. Tsirelson's problem is: Does there necessarily exist Hilbert spaces $\calH_A$, $\calH_B$, a density operator $\sigma$ on $\calH_A \otimes \calH_B$ and POVMs $\{R_k\}_k$ on $\calH_A$ and $\{S_l\}_l$ on $\calH_B$ such that $\tr(P_k Q_l \rho) = \tr((R_k \otimes S_l)\sigma)$ for all $k$, $l$? In the case where $\calH$ is finite-dimensional, the answer is known to be affirmative. For infinite-dimensional Hilbert spaces, the answer is still unknown. 

Thus, for finite-dimensional quantum systems, we can restrict ourselves to the case with the tensor product structure without loss of generality. For abstract state spaces, however, an analogous restriction might cause a loss of generality. The advantage of our weak definition of a tripartite scenario is that we do not need to know the answer to an equivalent of Tsirelson's problem for generalized probabilistic theories. The downside is that it makes defining an equivalent of the min-entropy more difficult. We will deal with this issue in the next subsection.

\begin{notation}
  From now on, whenever we speak of a tripartite scenario $S_{ABE}$, we implicitly assume 
  that all its parts and induced structures are denoted as in 
  \cref{def:ass,def:ass-induced,def:tripartite,def:trip-ind} without restating it, i.e.
  instead of writing ``Let $S_{ABE} = ((V, V^+, u), \calT_A, \calT_B, 
  \calT_E)$ be a tripartite scenario, let $\Omega$ be its set of normalized states, 
  \ldots'', we will only write ``Let $S_{ABE}$ be a tripartite scenario''.
\end{notation}

\subsection{A decoherence quantity for GPTs} \label{sec:dec-quant-gpt}

\subsubsection{Motivation of an expression that quantifies decoherence} \label{sec:dec-mot}

We are now going to motivate an expression for the central quantitiy $\Dec(A|E)_\omega$ for our decoherence analysis for GPTs. We take our inspiration from expression \eqref{hmin} for the quantum min-entropy, which we repeat here for the reader's convenience:
\begin{align}
  H_{\text{min}}(A|E)_\rho = - \log d_A \max_{\R_{E \rightarrow A'}} F^2(\Phi_{AA'}, \id_A \otimes \R_{E \rightarrow A'}(\rho_{AE})) \,. \tag{\ref{hmin} revisited}
\end{align}
There are two issues that prevent us from directly translating expression \eqref{hmin} into our framework. The first issue is that in \cref{sec:trip}, to keep our framework as general as possible, we have defined a tripartite scenario with an overall state space $(V, V^+, u)$ with tripartite states $\Omega$. We do not have notions of individual state spaces at hand. Thus, we do not have an analogue of a reduced state $\rho_{AE}$ or of a transformation $\R_{E \rightarrow A'}$ from one state space to another. 

The second issue is that we do not know what the analogue of a maximally entangled state $\Phi_{AA'}$ in our framework is. We resolve the first issue here in \cref{sec:dec-mot}, arriving at an expression for $\Dec(A|E)_\omega$. In \cref{sec:max-ent-gpt}, we will then define what a maximally entangled state is in our framework. 

Expression \eqref{hmin}, which involves the state $\rho_{AE}$ and TPCPMs $\R_{E \rightarrow A'}$, can be transformed to an expression in which both the state and the TPCPMs are purified (see \cref{fig:purif}). This expression will be our motivation for the expression for $\Dec(A|E)_\omega$. The maximization over TPCPMs from $E$ to $A'$ is replaced by a maximization over unitaries from $EE''$ to $A'A''$, where $E''$ and $A''$ are ancilla systems extending system $E$ and $A'$, respectively. This is precisely the purification (or Stinespring dilation) of a channel as in \cref{sec:background}. Since systems $EE''$ and $A'A''$ have the same dimension, we can identify their Hilbert spaces and regard the resulting Hilbert space as the Hilbert space of a system $E_{\text{tot}}$. This system involves all subsystems that the third party needs to control in order to bring itself as close as possible to maximal entanglement with Alice. Since $U_{E_{\text{tot}}}$ is a transformation on system $E_{\text{tot}}$ alone, we can translate it into our generalized framework.

\begin{figure}[h!]
  \centering
  \includegraphics{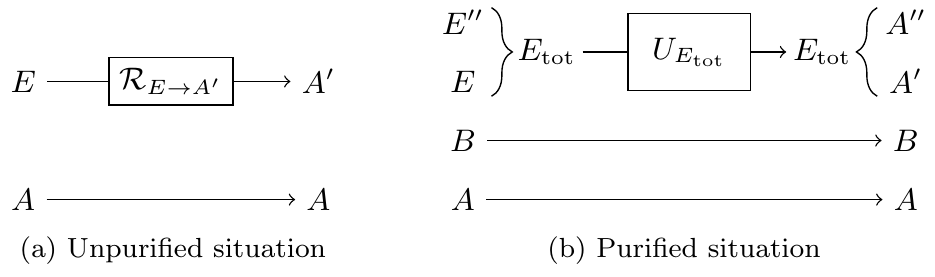}
  \caption{\textbf{Purification of \cref{hmin}.} Part (a) shows the system and 
  maps involved in expression \eqref{hmin} for the quantum min-entropy. In expression
  \eqref{eq:purif-hmin}, we purify this situation, as shown in (b), to arrive at a 
  situation with three parties $A$, $B$ and ${E_{\text{tot}}}$, and with a map $U_{E_{\text{tot}}}$ which acts on one 
  system ${E_{\text{tot}}}$ alone rather than mapping from one system to another. 
  \label{fig:purif}}
\end{figure}
  
The state $\rho_{AE}$ is replaced by a purification $\rho_{ABE}$. We choose the purifying system $B$ to be the channel's output system, which gives us the overall picture of our decoherence analysis as shown in \cref{fig:overall-gpt}.
\begin{figure}[htb]
  \centering
  \includegraphics{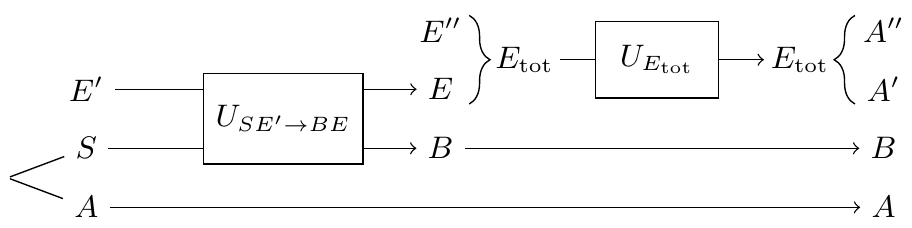}
  \caption{\textbf{Overall picture of our decoherence analysis for GPTs.} Since 
  the purifying system $B$ in expression \eqref{eq:purif-hmin} is not specified, we can 
  choose it such that it fits our situation for the decoherence analysis.
  \label{fig:overall-gpt}}
\end{figure}

The following gives a precise formulation of the purification of expression \eqref{hmin}. It can be proved using purification and Stinespring dilation. 
\textsl{
  Let $\calH_A$, $\calH_E$ be finite-dimensional Hilbert spaces of dimensions $d_A$, 
  $d_E$, respectively, let $\rho_{AE} \in \calS(\calH_A \otimes \calH_E)$. Then, 
  for any purification $\rho_{ABE} \in \calS(\calH_A \otimes \calH_B \otimes \calH_E)$ 
  of $\rho_{AE}$, any Hilbert spaces $\calH_{A'}$, $\calH_{A''}$ and $\calH_{E''}$ of 
  dimension $d_{A'} = d_A$, $d_{A''} = d_A d_E$ and $d_{E''} = d_A^2$, respectively,
  any maximally entangled state $\Phi_{AA'} \in \Gamma_{AA'}$ and any pure state 
  $|0\rangle\langle0|_{E''} \in \calS(\calH_{E''})$, it holds that
  \begin{align} \label{eq:purif-hmin}
    \Hmin(A|E)_\rho = -\log d_A \max_{U_{E_{\text{tot}}}} \max_{\sigma_{BA''}} 
    F^2(\Phi_{AA'} \otimes \sigma_{BA''}, (\id_{AB} \otimes U_{E_{\text{tot}}}) \rho_{AB{E_{\text{tot}}}} 
    (\id_{AB} \otimes U_{E_{\text{tot}}}^\dagger)) \,,
  \end{align}
  where $\rho_{AB{E_{\text{tot}}}} = \rho_{ABE} \otimes |0\rangle\langle0|_{E''}$ and where the 
  first maximization ranges over unitaries
  \begin{align}
    U_{E_{\text{tot}}}: \calH_E \otimes \calH_{E''} \rightarrow \calH_{A'} \otimes \calH_{A''} \,,
    \quad \text{where } \calH_{A'} \otimes \calH_{A''} \simeq \calH_E \otimes 
    \calH_{E''} =: \calH_{E_{\text{tot}}}
  \end{align}
  and the second maximization ranges over pure states $\sigma_{BA''} \in \calS(\calH_B 
  \otimes \calH_{A''})$.
}

Now we translate expression \eqref{eq:purif-hmin} into our generalized framework. We interpret the system $E_{\text{tot}}$ as the system controlled by Eve, and therefore rename $E_{\text{tot}} \rightarrow E$.
\begin{itemize}
  \item Since we want to arrive at an expression that does not make unnecessary 
    assumptions about the mathematical description of the physical situation, we avoid the 
    factor $d_A$ present in \eqref{eq:purif-hmin}. We look for a GPT analogue of 
    $\max_{U_{E_{\text{tot}}}} \max_{\sigma_{BA''}} F^2(\Phi_{AA'} \otimes \sigma_{BA''}, (\id_{AB} 
    \otimes U_{E_{\text{tot}}}) \rho_{AB{E_{\text{tot}}}} (\id_{AB} \otimes U_{E_{\text{tot}}}^\dagger))$, omitting $-\log d_A$. As a 
    consequence, we will have $\Hmin(A|E)_\rho = -\log d_A \Dec(A|E)_\rho$ in quantum 
    theory (see \cref{ex:qm-dec}). 
  \item We replace the maximization over all unitaries $U_{E_{\text{tot}}}$ acting on system ${E_{\text{tot}}}$ 
    by a supremum\footnote{We do not assume enough about $\calT_E$ to guarantee that the 
    maximum is achieved, so we replace it by a supremum.} 
    over all local transformations 
    $T_E \in \calT_E$.
    \footnote{One might raise the objection that in the quantum case, 
    \cref{ex:qm-trip}, the unitaries only correspond to those elements of $\calT_E$ which 
    bijectively map the space of density operators onto itself. It would be possible to 
    include this restriction, but we decide not to do so, for two reason: We want to keep 
    things simple, and we want to avoid the assumption that actions that the third party 
    can perform can be purified as in the quantum case.}
  \item We generalize the quantum fidelity to the fidelity in abstract state spaces as 
    defined in \cref{def:fidelity}.
  \item We replace the state $\rho_{ABE_{\text{tot}}} = \rho_{ABE} \otimes |0\rangle\langle0|_{E''}$ by 
    a state $\omega \in \Omega$. 
  \item If we look at the state $\Phi_{AA'} \otimes \sigma_{BA''}$, we see that it is a 
    state of maximal entanglement between Alice $(A)$ and Eve $(A'A'')$ in the 
    sense that by performing measurements with elements of the form $P_A \otimes P_{A'} 
    \otimes \id_B \otimes \id_{A''}$, they can get any statistics that two parties $A$ and 
    $A'$ would be able to get by performing local measurements on the maximally entangled 
    state $\Phi_{AA'}$. We translate this into our framework by assuming that there is a 
    set $\Psi_{AE}$ of ``states with maximal correlation between Alice and Eve''. 
    Instead of minimizing over states $\Phi_{AA'} \otimes \sigma_{BA''}$, we then minimize 
    over the set $\Psi_{AE}$. 
\end{itemize}
We postpone the discussion of how such a set $\Psi_{AE}$ looks like. We will give a definition of such a set in \cref{sec:max-ent-gpt} below. For now, we write down an expression for our decoherence quantity $\Dec(A|E)_\omega$ that depends on the choice of such a set $\Psi_{AE} \subseteq \Omega$. According to what we have just discussed, the expression is
\begin{align} \label{eq:first-dec}
  \sup_{T_E \in \calT_E} \sup_{\psi \in \Psi_{AE}} F^2(\psi, T_E(\omega)) \,.
\end{align}
We interpret the decoherence to be high when this quantity is high and vice versa, which is the opposite of $\Hmin(A|E)_\rho$ (see the end of \cref{sec:background}). Before we can define $\Dec(A|E)_\omega$, however, we need to specify what a maximally entangled state in a GPT is.

\subsubsection{Definition of maximal correlation in GPTs} \label{sec:max-ent-gpt}

The expression \eqref{eq:first-dec} for our decoherence quantity $\Dec(A|E)_\omega$ contains a maximization over a set $\Psi_{AE} \subseteq \Omega$ which we interpret to be the set of states with maximal correlation between Alice and Eve. We now define this set.

\begin{defi} \label{max-corr-def}
  For a tripartite scenario $S_{ABE}$, we define the set $\Psi_{AE}$ of \dt{states 
  with maximal correlation} between Alice and Eve by
  \begin{align}
    \Psi_{AE} := \left\{ \psi \in \Omega \ \middle\vert \ 
    \parbox{0.7\textwidth}{
    \textnormal{For every binary local instrument $I_A = \{T_A^0, T_A^1\} \in \calI_A$, 
    there is a binary local instrument $I_E = \{T_E^0, T_E^1\} \in \calI_E$ such that} 
    $(u \circ T_A^0 \circ T_E^0)(\psi) + (u \circ T_A^1 \circ T_E^1)(\psi) = 1$.
    } \right\} \,.
  \end{align} 
\end{defi}
\cref{max-corr-def} can be read as follows. The superscripts $0$ and $1$ of the elements of the instruments $I_A$ and $I_E$ stand for measurement outcomes, so $(u \circ T_A^0 \circ T_E^0)(\psi)$ or $(u \circ T_A^1 \circ T_E^1)(\psi)$ is the probability that Alice and Eve both get outcome $0$ or both get outcome $1$, respectively, when they measure with respect to $I_A$, $I_E$, respectively. Thus, the sum of these probabilities is the probability that Alice's and Eve's measurement outcomes are perfectly correlated. This means that for a state $\psi \in \Psi_{AE}$, it holds that for every binary measurement of Alice, there is a binary measurement for Eve such that their measurement outcomes are perfectly correlated.

A closer look at some subtleties is advisable here, both to avoid confusion and to see the advantages of the weak assumptions that define our framework. With reference to \cref{ex:qm-trip-3}, one may point out that that the set 
\begin{align}
  \left\{ \sigma \in \calS(\calH_A \otimes \calH_B \otimes \calH_E) \ \middle\vert \ 
  \parbox{0.6\textwidth}{
    \textnormal{For every binary POVM $\{P_A^0, P_A^1\}$ on $\calH_A$, 
    there is a binary POVM $\{P_E^0, P_E^1\}$ on $\calH_E$ such that} 
    \begin{align}
      \tr((P_A^0 \otimes \id_B \otimes P_E^0) \sigma) + 
      \tr((P_A^1 \otimes \id_B \otimes P_E^1) \sigma) = 1 \,. \nonumber
    \end{align}
  } \right\}
\end{align}
is empty. This may seem to make our definition of $\Psi_{AE}$ incompatible with quantum theory. Note, however, that the set
\begin{align}
  \left\{ \sigma \in \calS(\calH_A \otimes \calH_B \otimes \calH_E) \ \middle\vert \ 
  \parbox{0.6\textwidth}{
    \textnormal{For every binary projective measurement $\{P_A^0, P_A^1\}$ on 
    $\calH_A$, there is a binary projective measurement $\{P_E^0, P_E^1\}$ on 
    $\calH_E$ such that} 
    \begin{align}
      \tr((P_A^0 \otimes \id_B \otimes P_E^0) \sigma) + 
      \tr((P_A^1 \otimes \id_B \otimes P_E^1) \sigma) = 1 \,. \nonumber
    \end{align}
  } \right\}
\end{align}
is not empty as long as $\dim\calH_E \geq \dim\calH_A$. If, as in \cref{sec:dec-mot}, $\calH_E = \calH_{A'} \otimes \calH_{A''}$ with $\calH_{A'} \simeq \calH_A$, then this set contains all the states of the form $\Phi_{AA'} \otimes \sigma_{BA''}$ with $\Phi_{AA'} \in \Gamma_{AA'}$ as in \eqref{eq:purif-hmin}. The advantage of our weak definition of the local transformations is that it does not force to see $\calT_A$ as the analogue of the set of \emph{all} CPMs of the form $\R_A \otimes \id_B \otimes \id_E$, but that it can be considered to be the analogue of all such CPMs which induce a functional of the form $\sigma \mapsto \tr(P\sigma)$, where $P$ is a projector. \cref{ex:qm-trip-3} can be modified accordingly (see \cref{ex:qm-dec} below). This makes our definition of $\Psi_{AE}$ compatible with quantum theory. 

With \cref{max-corr-def} at hand, we are finally ready to define the decoherence quantity. 

\begin{defi} \label{def:dec-quant}
  Let $S_{ABE}$ be a tripartite scenario, let $\omega \in \Omega$. We define the 
  the \dt{decoherence quantity} of $\omega$ by
  \begin{align}
    \Dec(A|E)_\omega := \sup_{T_E \in \calT_E} \sup_{\psi \in \Psi_{AE}} 
    F^2(\psi, T_E(\omega))
  \end{align}
\end{defi}

\begin{ex} \label{ex:qm-dec}
  We consider a special case of a tripartite scenario in quantum theory. Consider
  \begin{align}
    &((\Herm(\calH), \Pos(\calH), \tr),\calT_A, \calT_B, \calT_E) \,, 
      \quad \text{where } \\
    &\calH = \calH_A \otimes \calH_B \otimes \calH_E \\
    &\calT_A = \left\{ \R_A \otimes \id_B \otimes \id_E \ \middle\vert \  
    \parbox{0.6\textwidth}{
      $\R_A$ is a trace non-increasing CPM on $\Herm(\calH_A)$ such that there is a 
      projector $P_A$ on $\calH_A$ with $\tr(P_A \rho_A) = \tr(\R_A(\rho_A))$ for all 
      $\rho_A \in \calS(\calH_A)$
    } \right\} \,,
  \end{align}
  and analogously for $\calT_B$ and $\calT_E$. In addition, we assume for simplicity 
  that $\calH_A \simeq \calH_E$. In this case,
  \begin{align}
    \Psi_{AE} &= \left\{ \sigma \in \calS(\calH_A \otimes \calH_B \otimes \calH_E) \ 
    \middle\vert \ \parbox{0.55\textwidth}{
      \textnormal{For every binary projective measurement $\{P_A^0, P_A^1\}$ on 
      $\calH_A$, there is a binary projective measurement $\{P_E^0, P_E^1\}$ on 
      $\calH_E$ such that} 
      \begin{align}
        \tr((P_A^0 \otimes \id_B \otimes P_E^0) \sigma) + 
        \tr((P_A^1 \otimes \id_B \otimes P_E^1) \sigma) = 1 \,. \nonumber
      \end{align}
    } \right\} \\
    &= \{ \Phi_{AE} \otimes \sigma_B \mid \Phi_{AE} \in \Gamma_{AE} \,, \ 
    \sigma_B \in \calS(\calH_B) \} \,,
  \end{align}
  where $\Gamma_{AE}$ is the set of maximally entangled states on $\calS(\calH_A \otimes 
  \calH_E)$ analogous to \eqref{gamma-def}. For a pure state $\rho_{ABE} \in \calS(\calH_A 
  \otimes \calH_B \otimes \calH_E)$, this gives us
  \begin{align}
    \Dec(A|E)_\rho &= \max_{\R_E} \max_{\Phi_{AE}} \max_{\sigma_B} 
    F^2(\Phi_{AE} \otimes \rho_B, \id_A \otimes \id_B \otimes \R_E(\rho_{ABE})) \\
    &= \max_{\R_E} F^2(\Phi_{AE}, \id_A \otimes \R_E(\rho_{AE})) \\
    &= \frac{1}{d_A} 2^{-\Hmin(A|E)_\rho} \,.
  \end{align}
  Hence, $\Hmin(A|E)_\rho = -\log d_A \Dec(A|E)_\rho$. \hfill $\blacksquare$
\end{ex}

\subsection{Bounds on the decoherence quantity for GPTs} \label{sec:gpt-bounds}

The goal of this subsection is to derive an upper bound on $\Dec(A|E)_\omega$ in terms of the CHSH winning probability of Alice and Bob. This is a practically relevant bound: On the premise that the channel behaves identically in multiple uses and does not build up correlations between different uses (such a channel is said to be \emph{iid}, for \emph{independent and identically distributed}), this winning probability can be estimated through repeated measurements on Alice's and Bob's side. What we show is that this estimate in turn gives a bound on $\Dec(A|E)_\omega$. In this section, we formulate this bound as a minimization problem which we solve and interpret in \cref{sec:gpt-results}.

In the following, we derive a lower bound on $-\log \Dec(A|E)_\omega$. We make the convention that $-\log 0 = \infty$, where $\infty$ is a symbol for which we accept the inequality $\infty \geq r$ for every real number $r$. This lower bound on $-\log \Dec(A|E)_\omega$ then gives us an upper bound on $\Dec(A|E)_\omega$. In a first step, we bound the fidelity-based quantity $-\log \Dec(A|E)_\omega$ by a trace distance-based quantity. This has the advantage that the resulting optimization problems which give us the bounds can be solved using linear programming.

\begin{prop} \label{geq-dsquared}
  Let $S_{ABE}$ be a tripartite scenario, let $\omega \in \Omega$. Then
  \begin{align} \label{eq:dgeq}
    -\log \Dec(A|E)_\omega \geq \inf_{T_E \in \calT_E} 
    \inf_{\psi \in \Psi_{AE}} D^2(\psi, T_E(\omega)) \,.
  \end{align}
\end{prop}

\noindent The following lemma is useful for the proof of \cref{geq-dsquared} below.

\begin{lemma} \label{log-lemma}
  For all $x \in (0,1]$, it holds that $-\log(x^2) \geq 2(1-x)$.
\end{lemma}

\begin{proof}
  We have that $-\log(x^2) = -2\log(x)$, so the claim is equivalent to
  \begin{align}
    (x-1) - \log(x) \geq 0 \quad \text{for all } x \in (0,1] \,.
  \end{align}
  The functions $F(x) = \log(x)$ and $G(x) = x-1$ are differentiable on $\mathbb{R}_{>0}$. 
  Thus, by 
  the fundamental theorem of calculus, it holds that for all $x \in \mathbb{R}_{>0}$,
  \begin{align}
    F(x) = F(1) + \int_1^x f(y) dy \,, \quad G(x) = G(1) + \int_1^x g(y) dy \quad 
    \text{where} \quad f(y) = \frac{d}{dy} F(y) \,, \quad g(y) = \frac{d}{dy} G(y) \,,
  \end{align}
  so for all $x \in (0,1]$, we have that
  \begin{align}
    (x-1) - \log(x) = G(x) - F(x) = \int_1^x g(y) - f(y) dy 
    = - \int_x^1 \underbrace{1 - \frac{1}{\ln(2)y}}_{< 0 \text{ for all } y \in (0,1]} dy 
    \geq 0 \,.
  \end{align}
  This proves the claim.
\end{proof}

\begin{proof}[Proof of \cref{geq-dsquared}]
  Since the right hand side of \eqref{eq:dgeq} is a finite real number, the inequality 
  trivially holds if $\Dec(A|E)_\omega = 0$ by the above convention. Thus, we assume in 
  the following that $\Dec(A|E)_\omega > 0$.
  We have that
  \begin{align}
    -\log \Dec(A|E)_\omega &= - \log \sup_{T_E \in \calT_E} \sup_{\psi \in \Psi_{AE}} 
    F^2(\psi, T_E(\omega)) \\
    &= - \log \left(\sup_{T_E \in \calT_E} \sup_{\psi \in \Psi_{AE}} 
    F(\psi, T_E(\omega)) \right)^2
  \end{align}
  For $x \in (0,1]$, it holds that $-\log x^2 \geq 2(1-x)$ (see \cref{log-lemma}). Thus, 
  since
  \begin{align}
    \sup_{T_E \in \calT_E} \sup_{\psi \in \Psi_{AE}} F(\psi, T_E(\omega)) \in (0,1] \,,
  \end{align}
  we get that
  \begin{align}
    -\log \Dec(A|E)_\omega &\geq 2 \left( 1 - \sup_{T_E \in \calT_E} \sup_{\psi \in
    \Psi_{AE}} F(\psi, T_E(\omega)) \right) \\
    &= \inf_{T_E \in \calT_E} \inf_{\psi \in \Psi_{AE}} 2(1-F(\psi, T_E(\omega))) \\
    &= \inf_{T_E \in \calT_E} \inf_{\psi \in \Psi_{AE}} \sup_{M \in \calM} 
    2(1-b(\psi, T_E(\omega)|M)) \,.
  \end{align}
  For the Bhattacharyya coefficient $b$ and the total variation distance $d$, it 
  has been shown \cite{Kra55} that for any two probability distributions distributions, it 
  holds that $2(1-b) \geq d^2$. Since this is true in particular for the two 
  probability distributions that the measurement $M$ induces on the states $\psi$ and 
  $T_E(\omega)$, we get that
  \begin{align}
    -\log \Dec(A|E)_\omega &\geq \inf_{T_E \in \calT_E} \inf_{\psi \in \Psi_{AE}} 
    \sup_{M \in \calM} d^2(\psi, T_E(\omega)|M) \\
    &= \inf_{T_E \in \calT_E} \inf_{\psi \in \Psi_{AE}} D^2(\psi, T_E(\omega)) \,,
  \end{align}
  as claimed.
\end{proof}

The idea that the fidelity and the trace distance are related is not new. In quantum theory, the \emph{Fuchs-van de Graaf inequalities} (FvdG) relate the two quantities \cite{FvdG99}. Inequality \eqref{eq:dgeq} is not completely analogous to the FvdG inequalities: It makes use of the logarithm in \eqref{eq:dgeq}, which allows to apply classical relations that lead to a stronger bound than with the application of the FvdG inequalities. 

For the bounds that we are going to derive, the notion of a \emph{non-signalling distribution} is central. Our bounds are essentially minimizations of functions over sets of non-signalling distributions $\Pr[a,b,c|x,y,z]_\omega$ and $\Pr[a,c|x,z]_\psi$ with certain additional properties. 

\begin{defi}
  A set of numbers $\Pr[a,b,c|x,y,z]_\omega \in [0,1]$, indexed by numbers 
  $a,b,c \in \{0,1\}$ which we call \dt{outcomes}, and numbers $x,y,z \in \{0,1\}$ 
  which we call \dt{settings}, is a \dt{non-signalling distribution} if
  \begin{align}
    &\text{normalization:} &&\sum_{a, b, c} 
    \Pr[a,b,c|x,y,z]_\omega = 1 
    \quad \text{for all } x,y,z \in \{0,1\} \,, \label{omega-norm} \\
    &\text{no-signalling:} &&\sum_{a} \Pr[a,b,c|0,y,z]_\omega 
    = \sum_{a} \Pr[a,b,c|1, y, z]_\omega \quad \text{for all } b,c,y,z \in \{0,1\} \,, 
    \label{omega-ns-1} \\
    & &&\sum_{b} \Pr[a,b,c|x,0,z]_\omega 
    = \sum_{b} \Pr[a,b,c | x,1,z]_\omega \quad \text{for all } a,b,x,y \in \{0,1\} \,, 
    \label{omega-ns-2}\\
    & &&\sum_{b} \Pr[a,b,c|x,y,0]_\omega 
    = \sum_{b} \Pr[a,b,c|x,y,1]_\omega \quad \text{for all } a,b,x,y \in \{0,1\} \,.
    \label{omega-ns-3}
  \end{align}
  Similarly, a set of numbers $\Pr[a,c|x,z]_\psi \in [0,1]$, indexed by outcomes 
  $a,c \in \{0,1\}$ and settings $x,z \in \{0,1\}$ is a non-signalling distribution if
  \begin{align}
    &\text{normalization:} &&\sum_{a,c} 
    \Pr[a,c|x,z]_\psi = 1 \quad \text{for all } x,z \in \{0,1\} \,,  \label{psi-norm} \\
    &\text{no-signalling:} &&\sum_{a} \Pr[a,c|0,z]_\psi 
    = \sum_{a} \Pr[a,c|1,z]_\psi \quad \text{for all } c,z \in \{0,1\} \,, 
    \label{psi-ns-1}\\
    & &&\sum_{c} \Pr[a,c|x,0]_\psi 
    = \sum_{c} \Pr[a,c|x,1]_\psi \quad \text{for all } a,x \in \{0,1\} 
    \label{psi-ns-2} \,.
  \end{align}
\end{defi}

The interpretation of \cref{omega-ns-1,omega-ns-2,omega-ns-3} is that it is impossible for each of the three parties to \emph{signal} to the other two parties by influencing their measurement statistics with the choice of the measurement setting. These one-party no-signalling constraints imply all the multi-party no-signalling constraints, saying that no collection of parties can signal to the remaining parties \cite{BLM05}, so we do not need to require these constraints separately.

Now we are going to formulate the bound on $-\log \Dec(A|E)_\omega$ in terms of the CHSH winning probability of Alice and Bob. Assume that Alice, Bob and Eve are in a situation described by a tripartite scenario $S_{ABE}$. Suppose that Alice and Bob have estimated that for the state $\omega \in \Omega$ that they are analyzing, their CHSH winning probability is at least $\lambda$ for some $\lambda \in [0,1]$. Formulated in our tripartite scenario language, this means that they have found out that for local instruments
\begin{align}
  I_A^0 = \{ T_A^{0|0}, T_A^{1|0} \} \in \calI_A, 
  &\qquad I_B^0 = \{T_B^{0|0}, T_B^{1|0}\} \in \calI_B, \label{eq:el-1} \\
  I_A^1 = \{ T_A^{0|1}, T_A^{1|1} \} \in \calI_A, 
  &\qquad I_B^1 = \{T_B^{0|1}, T_B^{1|1}\} \in \calI_B \label{eq:el-2} \,,
\end{align}
it holds that
\begin{align} \label{eq:lambda-prem-text}
  \frac{1}{4} \sum_{x,y} \sum_{\substack{a,b \\ a \oplus b = xy}} 
  \left(u \circ T_A^{a|x} \circ T_B^{b|y}\right)(\omega) \geq \lambda \,.
\end{align}
In that case, what can Alice and Bob infer about $-\log \Dec(A|E)_\omega$? We have seen in \cref{geq-dsquared} that this quantity is lower bounded by $\inf_{T_E \in \calT_E} \inf_{\psi \in \Psi_{AE}} D^2(\psi, T_E(\omega))$. Alice's and Bob's estimate on their CHSH winning probability can be translated into a bound on this quantity. This is shown by the following proposition.

\begin{prop} \label{lin-bound-prop}
  Let $S_{ABE}$ be a tripartite scenario, let $\omega 
  \in \Omega$ be a state. If the CHSH winning probability of Alice and Bob is at least 
  $\lambda$, i.e. if there are local instruments $I_A^0$, $I_A^1$, $I_B^0$ and $I_B^1$ as 
  in \eqref{eq:el-1} and \eqref{eq:el-2} and a $\lambda \in [0,1]$ such that
  \eqref{eq:lambda-prem-text} is satisfied, then 
  \begin{align}
    &\inf_{T_E \in \calT_E} \inf_{\psi \in \Psi_{AE}} 
    D(\psi, T_E(\omega)) \geq \min_{\substack{x,z \in \{0,1\} \\ 
    \Pr[a,b,c|x,y,z]_\omega \in \calD_\omega(\lambda) \\ 
    \Pr[a,c|x,z]_\psi \in \calD_\psi}}
    \frac{1}{2} \sum_{a,c} \left\vert \Pr[a,c|x,z]_\psi 
    - \sum_b \Pr[a,b,c|x,y,z]_\omega \right\vert \,,
  \end{align}
  where $\calD_\omega(\lambda)$ is the set of 
  non-signalling distributions for Alice, Bob and Eve such that Alice and Bob have a 
  CHSH winning probability of at least $\lambda$, i.e.
  \begin{align}
    &\calD_\omega(\lambda) = \left\{ \Pr[a,b,c|x,y,z]_\omega \ \middle\vert \
    \parbox{0.55\textwidth}{
      $\Pr[a,b,c|x,y,z]_\omega$ is a non-signalling distribution such that 
      \begin{align}
        \frac{1}{4} \sum\limits_{x,y} \sum\limits_{\substack{a,b \\ a \oplus b = xy}} 
        \sum\limits_c \Pr[a,b,c|x,y,z]_\omega \geq \lambda \,,  \nonumber
      \end{align}
    } \right\} \,. \label{eq:d-omega}
  \end{align}
  and where $\calD_\psi$ is the set of non-signalling distributions for Alice 
  and Eve such that their measurement outcomes are perfectly correlated when they 
  choose the same measurement setting, i.e.
  \begin{align}
    \calD_\psi
    = \left\{ \Pr[a,c|x,z]_\psi \ \middle\vert \ \parbox{0.4\textwidth}{
      $\Pr[a,c|x,z]_\psi$ is a non-signalling distribution such that 
      $\Pr[a=c|x=z]_\psi = 1$. \label{eq:corr-d}
    } \right\} \,.
  \end{align}
\end{prop}

\cref{lin-bound-prop} reduces our problem of lower bounding the decoherence quantity for GPTs to an optimization over non-signalling distributions. This allows us to use linear programming techniques, which in similar ways have been used in \cite{Ton09} to answer questions about non-signalling distributions. 

We need the following lemma for the proof of \cref{lin-bound-prop} below.

\begin{lemma} \label{infd-geq}
  Let $S_{ABE}$ be a tripartite scenario, let $\omega, \psi \in \Omega$. Then, for all 
  local instruments $(I_A, I_B, I_E) \in \calI_A \times \calI_B 
  \times \calI_E$, it holds that
  \begin{align}
    \inf_{T_E \in \calT_E} D(\psi, T_E(\omega)) 
    \geq \frac{1}{2} \inf_{T_E \in \calT_E} 
    \sum_{\substack{T_A \in I_A \\ U_E \in I_E}} 
    \abs{(u \circ T_A \circ U_E)(\psi) - \sum_{T_B \in I_B} (u \circ T_A \circ T_B \circ 
    U_E \circ T_E)(\omega)} \,.
  \end{align}
\end{lemma}

\begin{proof}
  It is sufficient to show that for all $\omega, \psi \in \Omega$, for all $T_E \in 
  \calT_E$ and for all $(I_A, I_B, I_E) \in \calI_A \times \calI_B \times \calI_E$,
  \begin{align}
    D(\psi, T_E(\omega)) 
    \geq \frac{1}{2} \sum_{\substack{T_A \in I_A \\ U_E \in I_E}} 
    \left\vert (u \circ T_A \circ U_E)(\psi) - \sum_{T_B \in I_B}
    (u \circ T_A \circ T_B \circ U_E \circ T_E)(\omega) \right\vert \,.
  \end{align}
  This is what we are going to show now. Let $\omega, \psi \in \Omega$, let $T_E \in 
  \calT_E$, let $(I_A, I_B, I_E) \in \calI_A \times \calI_B \times \calI_E$. Then
  \begin{align}
    D(\psi, T_E(\omega)) 
    = \sup_{M \in \calM} d(\psi, T_E(\omega) | M) 
    = \frac{1}{2} \sup_{M \in \calM} \sum_{e \in M} 
    \abs{e(\psi)  - e(T_E(\omega))} \,. \label{d-delta-sum}
  \end{align}
  If instead of taking the supremum over $\calM$, we only evaluate the expression for a 
  particular element of $\calM$, we get a lower bound on \eqref{d-delta-sum}. We choose 
  the element (c.f. \cref{rmk:loc-meas})
  \begin{align}
    \{u \circ T_A \circ U_E \mid T_A \in \calI_A, U_E \in \calI_E \} \in \calM \,.
  \end{align}
  Hence, 
  \begin{align}
    D(\psi, T_E(\omega)) 
    &\geq \frac{1}{2} \sum_{\substack{T_A \in \calI_A \\ U_E \in 
    \calI_E}} \abs{(u \circ T_A \circ U_E)(\psi) - (u \circ T_A \circ U_E \circ T_E)
    (\omega)} \,.
  \end{align}
  By the definition of a local instrument, $u = \sum_{T_B \in I_B} u \circ T_B$. Thus,
  \begin{align}
    D(\psi, T_E(\omega)) 
    &\geq \frac{1}{2} \sum_{\substack{T_A \in \calI_A \\ U_E \in 
    \calI_E}} \abs{(u \circ T_A \circ U_E)(\psi) - \sum_{T_B \in I_B} (u \circ T_B \circ 
    T_A \circ U_E \circ T_E)(\omega)} \\
    &= \frac{1}{2} \sum_{\substack{T_A \in \calI_A \\ U_E \in 
    \calI_E}} \abs{(u \circ T_A \circ U_E)(\psi) - \sum_{T_B \in I_B} (u \circ T_A \circ 
    T_B \circ U_E \circ T_E)(\omega)} \,,
  \end{align}
  where in the last equality, we made use of the fact that transformations of different 
  parties commute.
\end{proof}

\begin{proof}[Proof of \cref{lin-bound-prop}]
  It is sufficient to show that for every $\psi \in \Psi_{AE}$, the claimed inequality 
  holds without the minimization over $\Psi_{AE}$, i.e. 
  \begin{align} \label{suff-to-show}
    \inf_{T_E \in \calT_E} 
    D(\psi, T_E(\omega)) \geq \min_{\substack{x,z \in \{0,1\} \\ 
    \Pr[a,b,c|x,y,z]_\omega \in \calD_\omega(\lambda) \\ 
    \Pr[a,c|x,z]_\psi \in \calD_\psi(\Psi_{AE}, \lambda)}}
    \frac{1}{2} \sum_{a,c} \left\vert \Pr[a,c|x,z]_\psi 
    - \sum_b \Pr[a,b,c|x,y,z]_\omega \right\vert \,.
  \end{align}
  By means of \cref{infd-geq}, we know that for all $x, y \in \{0,1\}$ and every 
  $I_E \in \calI_E$,
  \begin{align}
    \inf_{T_E \in \calT_E} D(\psi, T_E(\omega)) 
    &\geq \frac{1}{2} 
    \inf_{T_E \in \calT_E}
    \sum_{\substack{a \in \{0,1\} \\ U_E \in I_E}} \left\vert 
    (u \circ T_A^{a|x} \circ U_E)(\psi) 
    - \sum_{b \in \{0,1\}} (u \circ T_A^{a|x} \circ T_B^{b|y} \circ U_E \circ T_E)(\omega) 
    \right\vert \,.
  \end{align}
  Let $\psi \in \Psi_{AE}$, let $I_E^0 = \{U_E^{0|0}, U_E^{1|0}\}$, $I_E^1 = \{U_E^{0|1}, 
  U_E^{1|1}\} \in \calI_E$ be local instruments for Eve such that
  \begin{align}
    &(u \circ T_A^{0|0} \circ U_E^{0|0})(\psi) + (u \circ T_A^{1|0} \circ U_E^{1|0})(\psi) 
    = 1 \,, \label{corr-1} \\
    &(u \circ T_A^{0|1} \circ U_E^{0|1})(\psi) + (u \circ T_A^{1|1} \circ U_E^{1|1})(\psi) 
    = 1 \,, \label{corr-2}
  \end{align}
  which exist according to the definition of $\Psi_{AE}$ 
  (\cref{max-corr-def}). It holds that for every $x, y, z \in \{0,1\}$, 
  \begin{align}
    \inf_{T_E \in \calT_E} D(\psi, T_E(\omega)) 
    &\geq \frac{1}{2} 
    \inf_{T_E \in \calT_E}
    \sum_{a,c \in \{0,1\}} \left\vert 
    (u \circ T_A^{a|x} \circ U_E^{c|z})(\psi) - 
    \sum_{b \in \{0,1\}} (u \circ T_A^{a|x} \circ T_B^{b|y} \circ U_E^{c|z} \circ T_E)
    (\omega) \right\vert \,. \\
    &= \frac{1}{2} \inf_{T_E \in \calT_E}
    \sum_{a,c \in \{0,1\}} \left\vert 
    \Pr[a,c|x,z]_\psi - \sum_{b \in \{0,1\}} \Pr[a,b,c|x,y,z]_\omega \right\vert \,,
  \end{align}
  where
  \begin{align}
    &\Pr[a,c|x,z]_\psi = (u \circ T_A^{a|x} \circ U_E^{c|z})(\psi) \,, \\
    &\Pr[a,b,c|x,y,z]_\omega = (u \circ T_A^{a|x} \circ T_B^{b|y} \circ U_E^{c|z} \circ 
    T_E)(\omega) \,.
  \end{align}
  Hence,
  \begin{align}
    \inf_{T_E \in \calT_E} D(\psi, T_E(\omega))
    &\geq \frac{1}{2} \min_{x,y,z} \inf_{T_E \in \calT_E}
    \sum_{a,c \in \{0,1\}} \left\vert 
    \Pr[a,c|x,z]_\psi - \sum_{b \in \{0,1\}} \Pr[a,b,c|x,y,z]_\omega \right\vert \,.
  \end{align}
  $\Pr[a,c|x,z]_\psi$ forms a non-signalling distribution: For the normalization, note 
  that for all $x,z \in \{0,1\}$, we have that
  \begin{align}
    \sum_{a,c} \Pr[a,c|x,z]_\psi &= \sum_{a,c} 
    (u \circ T_A^{a|x} \circ U_E^{c|z})(\psi) = \underbrace{\left(\sum_a u \circ T_A^{a|x}
    \right)}_{u} \circ \left( \sum_c U_E^{c|z} \right) (\psi) \\
    &= \left( \sum_c u \circ U_E^{c|z} \right) (\psi) = u (\psi) = 1 \,.
  \end{align}
  For the no-signalling condition, note that for all $c,z \in \{0,1\}$, it holds that
  \begin{align}
    \sum_a \Pr[a,c|0,z]_\psi &= \left(\left( \sum_a u \circ T_A^{a|0}\right) \circ 
    U_E^{c|z}\right) (\psi) = (u \circ U_E^{c|z})(\psi) 
    = \left(\left( \sum_a u \circ T_A^{a|1}\right) \circ U_E^{c|z}\right) (\psi) \\
    &= \sum_a \Pr[a,c|1,z]_\psi \,,
  \end{align}
  and that for all $a,x \in \{0,1\}$, it holds that
  \begin{align}
    \sum_c \Pr[a,c|x,0]_\psi &= \sum_c (u \circ T_A^{a|x} \circ U_E^{c|0})(\psi)
    = \sum_c (u \circ U_E^{c|0} \circ T_A^{a|x})(\psi) 
    = \left( \left( \sum_c u \circ U_E^{c|0} \right) \circ T_A^{a|x} \right) (\psi) \\
    &= (u \circ T_A^{a|x})(\psi) 
    = \left( \left( \sum_c u \circ U_E^{c|1} \right) \circ T_A^{a|x} \right) (\psi)
    = \sum_c \Pr[a,c|x,1]_\psi \,,
  \end{align}
  where in the second equality, we made use of the fact that local transformations of 
  different parties commute. Analogously, for every $T_E \in \calT_E$, 
  $\Pr[a,b,c|x,y,z]_\omega$ is a non-signalling distribution. Moreover, 
  $\Pr[a,b,c|x,y,z]_\omega$ satisfies
  \begin{align}
    \frac{1}{4} \sum\limits_{x,y} \sum\limits_{\substack{a,b \\ a \oplus b = xy}} 
    \sum\limits_c \Pr[a,b,c|x,y,z]_\omega 
    &= \frac{1}{4} \sum\limits_{x,y} \sum\limits_{\substack{a,b \\ a \oplus b = xy}} 
    \sum\limits_c (u \circ T_A^{a|x} \circ T_B^{b|y} \circ U_E^{c|z} \circ T_E)(\omega) \\
    &= \frac{1}{4} \sum\limits_{x,y} \sum\limits_{\substack{a,b \\ a \oplus b = xy}} 
    \underbrace{\left(\sum_c (u \circ U_E^{c|z} \circ T_E\right)}_{u} \circ T_A^{a|x} 
    \circ T_B^{b|y})(\omega) \\
    &\geq \lambda \,, 
  \end{align}
  where the inequality is one of the assumptions of the proposition. Furthermore, 
  $\Pr[a,c|x,1]_\psi$ satisfies
  \begin{align}
    &\Pr[0,0|0,0]_\psi + \Pr[1,1|0,0]_\psi 
    = (u \circ T_A^{0|0} \circ U_E^{0|0})(\psi) + (u \circ T_A^{1|0} \circ U_E^{1|0})
    (\psi) = 1 \,, \\
    &\Pr[0,0|1,1]_\psi + \Pr[1,1|1,1]_\psi 
    = (u \circ T_A^{0|1} \circ T^{0|1})(\psi) + (u \circ T_A^{1|1} \circ U_E^{1|1})
    (\psi) = 1 \,
  \end{align}
  (where we made use of \eqref{corr-1} and \eqref{corr-2}), which we may abbreviate as
  $\Pr[a=c|x=z]_\psi = 1$. Thus, for every $T_E \in \calT_E$, we have that
  $\Pr[a,b,c|x,y,z]_\omega \in \calD_\omega(\lambda)$ and $\Pr[a,c|x,z]_\psi \in 
  \calD_\psi$. Thus,
  \begin{align} \label{xyz-corr}
    \inf_{T_E \in \calT_E} D(\psi, T_E(\omega))
    &\geq \frac{1}{2} \min_{x,y,z} \inf_{\substack{
    \Pr[a,b,c|x,y,z]_\omega \in \calD_\omega(\lambda) \\ 
    \Pr[a,c|x,z]_\psi \in \calD_\psi(\Psi_{AE}, \lambda)}}
    \sum_{a,c \in \{0,1\}} \left\vert 
    \Pr[a,c|x,z]_\psi - \sum_{b \in \{0,1\}} \Pr[a,b,c|x,y,z]_\omega \right\vert \,.
  \end{align}
  Since $\Pr[a,b,c|x,y,z]_\omega$ satisfies the no-signalling property, the right hand 
  side of \eqref{xyz-corr} is independent of $y$, so the minimization only needs to be 
  performed over $x$ and $z$. Moreover, the infimum over the sets $\calD_\omega(\lambda)$ 
  and $\calD_\psi$ is a minimum because it is the infimum 
  of a continuous function over a convex polytope, which is always attained (see 
  \cref{sec:gpt-results} for more details). This completes the proof.
\end{proof}

\begin{cor}[The bound]
  Let $S_{ABE}$ be a tripartite scenario, let $\omega \in \Omega$ be a state. If the CHSH 
  winning probability of Alice and Bob is at least $\lambda$ (in the above sense), then
  \begin{align}
    \Dec(A|E)_\omega \leq 2^{-\delta^2(\lambda)} \,, 
    \label{eq:corr-bound}
  \end{align}
  where
  \begin{align}
    \delta(\lambda) = \min_{\substack{x,z \in \{0,1\} \\ 
    \Pr[a,b,c|x,y,z]_\omega \in \calD_\omega(\lambda) \\ 
    \Pr[a,c|x,z]_\psi \in \calD_\psi}}
    \frac{1}{2} \sum_{a,c} \left\vert \Pr[a,c|x,z]_\psi 
    - \sum_b \Pr[a,b,c|x,y,z]_\omega \right\vert \,, \label{eq:corr-min}
  \end{align}
\end{cor}

\begin{proof}
  This is a direct consequence of \cref{geq-dsquared,lin-bound-prop}.
\end{proof}

\subsection{Evaluation of the bound and results} \label{sec:gpt-results}

\subsubsection{Formulation of the bound as a linear program}

In this subsection, we evaluate the bound \eqref{eq:corr-bound}. To this end, we rewrite \eqref{eq:corr-min} in terms of linear programs.

\begin{lprog} \label{lprog:corr}
  The bound $\delta(\lambda)$, which is a function $\delta: [0,1] 
  \rightarrow [0,1]$, is given as follows. For all $\lambda \in [0,1]$, the value 
  $\delta(\lambda)$ is the solution of the linear program
  \begin{align}
  \begin{linprog}{minimize}{$\delta(\lambda)$}
    &$\Pr[a,b,c|x,y,z]_\omega \in \calD_\omega(\lambda)$ \\ 
    &$\Pr[a,c|x,z]_\psi \in \calD_\psi$ \\ 
    &$\delta(\lambda) \geq \sum_{a,c} \delta_{ac}^{xz} \ 
    \forall x,z \in \{0,1\}$ \\
    &$\delta_{ac}^{xz} \geq \frac{1}{2} \left( \Pr[a,c|x,z]_\psi - \sum_b 
    \Pr[a,b,c|x,0,z]_\omega \right) \geq -\delta_{ac}^{xz}\ \forall a,c,x,z \in \{0,1\}$
  \end{linprog}
  \end{align}
\end{lprog}
\noindent This is a linear program in 97 variables:
\begin{align}
  &\big\{ \Pr[a,b,c|x,y,z]_\omega \big\}_{a,b,c,x,y,z \in \{0,1\}} &&\phantom{{}+{}} 64 
  \text{ variables} \nonumber \\
  &\big\{ \Pr[a,c|x,y]_\psi \big\}_{a,c,x,z \in \{0,1\}} &&+16 \text{ variables} 
  \nonumber \\
  &\big\{ \delta_{ac}^{xz} \big\}_{a,c,x,z \in \{0,1\}} &&+16 
  \text{ variables} \nonumber  \\
  &\delta(\lambda)
  &&+1 \text{ variable} \nonumber \\
  & &&=97 \text{ variables} \nonumber
\end{align}
We have already written out the constraints for these 97 variables as (in)equalities. The third and fourth line are already written as such in the program description, and for the first two lines, we refer to the following:
\begin{center} 
  \begin{tabular}{ll} 
    constraint & (in)equalities \\
    \hline 
    $\Pr[a,b,c|x,y,z]_\omega \in \calD_\omega(\lambda)$ & \eqref{omega-norm} to 
    \eqref{omega-ns-3} and inequality in \eqref{eq:d-omega} \\
    $\Pr[a,c|x,z]_\psi \in \calD_\psi(\lambda)$  
    &\eqref{psi-norm} to \eqref{psi-ns-2} and equation in \eqref{eq:corr-d}
  \end{tabular}
\end{center}
The inequalities define a convex polytope over which the convex function $\delta(\lambda)$ is minimized, so the minimum is attained. It is straightforward to bring these inequalities into the standard form of linear programming. We solved the resulting linear program using standard linear programming routines in Mathematica and Octave.

\subsubsection{Solution of the linear program and discussion of the results}

\definecolor{blueish}{HTML}{A3D9FF}
\definecolor{yellowish}{HTML}{FFE0B5}

\begin{figure}[h!]
\begin{center}
  \includegraphics{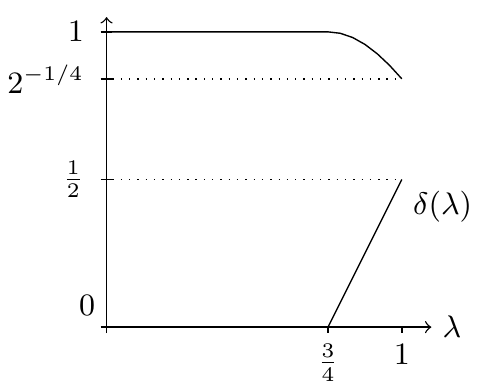}
  \caption{\textbf{Plot of the result.} The bound $2^{-\delta^2(\lambda)}$ is 
  non-trivial precisely when the CHSH winning probability for Alice and Bob is 
  non-classical, i.e. $\lambda > 3/4$. \label{fig:result}}
\end{center}
\end{figure}

We plot the result in \cref{fig:result}. The bound $2^{-\delta^2(\lambda)}$ is non-trivial for values $\lambda \in (3/4,1]$. This is a very satisfactory result as one cannot expect the bound to be non-trivial for $\lambda \in [0,3/4]$: A CHSH winning probability of at least $\lambda \in [0,3/4]$ for Alice and Bob is always compatible with $\Dec(A|E)_\omega = 1$. To see this, note that the requirement for a state to yield a CHSH winning probability for Alice and Bob of at least $ \lambda \in [0,3/4]$ is trivial: Alice and Bob can choose trivial measurements that always yields 1 as an outcome, independently of the state. More precisely, in our tripartite scenario language, we can express this as follows. Certainly, there are tripartite scenarios in which the identity map $\id_V$ and the zero map $0_V$ are in $\calT_A$, $\calT_B$ and $\calT_E$.\footnote{Every tripartite scenario can be turned into such by adding $\id_V$ and $0_V$ to the sets of local transformations. In fact, it would be physical to assume that each set of local transformations contains $\id_V$ and $0_V$.} 
For such tripartite scenarios, the condition (c.f. \eqref{eq:el-1} to \eqref{eq:lambda-prem-text} in \cref{sec:gpt-bounds})
\begin{align} 
  \frac{1}{4} \sum_{x,y} \sum_{\substack{a,b \\ a \oplus b = xy}} 
  \left(u \circ T_A^{a|x} \circ T_B^{b|y}\right)(\omega) \geq \lambda \in [0,3/4] 
  \quad \text{for some } 
  \begin{array}{l}\{ T_A^{0|0}, T_A^{1|0} \}, \{ T_A^{0|1}, T_A^{1|1} \} \in \calI_A, \\ 
  \{T_B^{0|0}, T_B^{1|0}\}, \{T_B^{0|1}, T_B^{1|1}\} \in \calI_B\end{array} 
\end{align}
is always satisfied because for all $\omega \in \Omega$,
\begin{align}
  \frac{1}{4} \sum_{x,y} \sum_{\substack{a,b \\ a \oplus b = xy}} 
  \left(u \circ T_A^{a|x} \circ T_B^{b|y}\right)(\omega) = \frac{3}{4} \quad \text{for } 
  \begin{array}{l}T_A^{0|0} = T_A^{0|1} = T_B^{0|0} = T_B^{0|1} = 0_V \,, \\T_A^{1|0} = 
  T_A^{1|1} = T_B^{1|0} = T_B^{1|1} = \id_V \,. \end{array}
\end{align}
This means that the requirement that the CHSH winning probability for Alice and Bob is at least $\lambda \in [0,3/4]$ does not exclude the case $\omega \in \Psi_{AE}$. In that case, $\Dec(A|E)_\omega = \sup_{T_E \in \calT_E} \sup_{\psi \in \Psi_{AE}} F^2(\psi, T_E(\omega)) = 1$.

\section{A test for gravitational decoherence}

\subsection{An optomechanical setting and its model for gravitational decoherence}
\label{sec:grav-deco-model}

The objective here is to create two entangled photonic qubits in which one photon is prepared in an opto-mechanical system that is itself subject to gravitational decoherence --- if there is any --- and the other photon is prepared in an identical cavity except the mirrors are fixed and cannot move. This model is a modification of the model first proposed by Bouwmeester \cite{Bouwmeester2003} in which an itinerant single photon pulse is injected into a cavity rather than created intra-cavity as here. Our modification avoids the problem that the time over which the photons interact with the mechanical element is stochastic and determined by the random times at which the photons enter and exit the cavity through an end mirror. In the new scheme, the cavities are assumed to have almost perfect mirrors --- very narrow line width (see for example\cite{NIST}).

The intracavity single photon Raman source is described in Nisbet-Jones, et al. \cite{NDLK11}.  In this scheme (see Fig. \ref{fig1})  a control pulse can quickly and efficiently prepare a cavity mode in a single photon state by driving a Raman transition between two hyperfine levels we label as $|g\rangle, |e\rangle$.
 In our scheme there are two optical cavities otherwise identical except in one of the cavities  a mechanical element can respond to the radiation pressure force of light.

We will assume that we can prepare the atomic sources in an arbitrary entangled state $|g,e\rangle+|e,g\rangle$, for example, using the trapped ion schemes of  Monroe \cite{DM10}. In addition we will assume that we can make arbitrary rotations in the $g,e$ subspace of each source and also make fast efficient single shot readout of the state of each source, for example using fluorescence shelving. This means we can readout the atomic qubit in each cavity in any basis. 

The write laser implements the Hamiltonian $H_{w}=i\hbar\Omega(t)(a^\dagger|e\rangle\langle g|-a |g\rangle\langle e|)/2$. This is a rotation in the state space 
$\{ |g\rangle|0\rangle,|e\rangle|1\rangle \}$. We can thus prepare arbitrary states of the form $\cos\theta/2|g\rangle|0\rangle+\sin\theta/2 |e\rangle|1\rangle$, where $\theta$ is determined by the pulse area. We will refer to the case of $\theta =\pi$ as a $\pi$-pulse.   Note that if the source is in the excited state $|e\rangle$ and the cavity is in the vacuum, no photon is excited.  

Starting with the cavities in the vacuum state the protocol proceeds as follows:
\begin{enumerate}
  \item prepare the source atoms in the state $|g,e\rangle+|e,g\rangle$.
  \item apply the write laser with a $\pi$-pulse
  \item free evolution of the OM systems for a time $T$
  \item apply the write laser with a $\pi$-pulse
  \item readout the atomic state in each cavity. 
\end{enumerate}
At the end of Step 2, the state of the sources and the cavities is $|\psi_2\rangle=|e,e\rangle\otimes(|1,0\rangle+|0,1\rangle)$ where $|n,m\rangle=|n\rangle\otimes|m\rangle$ with each factor being a photon number eigenstate. 

\begin{figure}[h!] 
   \includegraphics[scale =0.7]{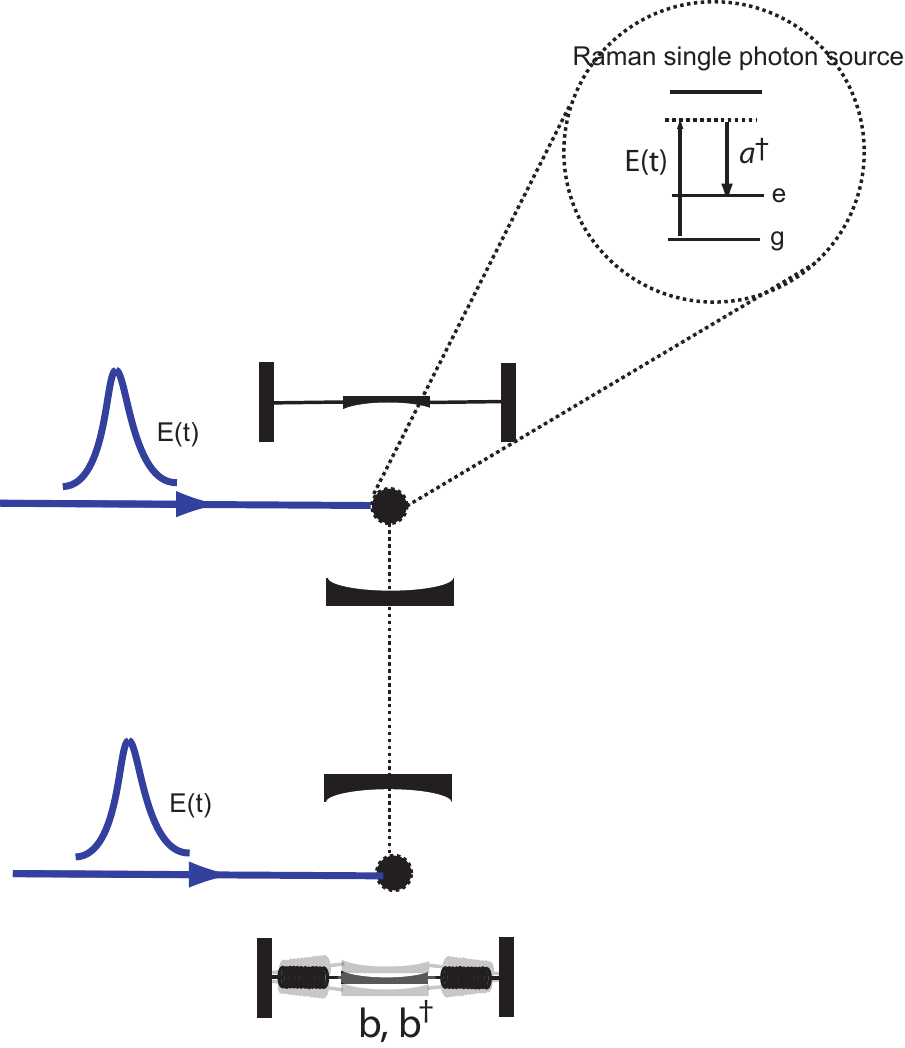} 
      \centering
   \caption{Two cavities each contain a Raman single photon source controlled by an external laser `write field' $E(t)$. The Raman sources are first prepared in an entangled state. Only one cavity contains a mechanical element coupled by radiation rouser to the cavity field. }
   \label{fig1}
\end{figure}

\subsubsection{Gravitational decoherence.}
We will use Diosi's theory of gravitational decoherence\cite{Dio89}. This is equivalent to the decoherence model introduced in Kafri et al. \cite{KTM14}.
One mirror of the opto-mechanical cavity is free to move in a harmonic potential with frequency $\omega_m$. The master equation for a massive particle moving in a harmonic potential, including gravitational decoherence is 
\begin{equation}
\label{grav-me}
\frac{d\rho}{dt}=-i\omega_m[b^\dagger b,\rho]-\Lambda_{\text{grav}}[b+b^\dagger,[b+b^\dagger,\rho]]
\end{equation}
where 
\begin{equation}
b=\sqrt{\frac{m\omega_m}{2\hbar}}\hat{x}+i\frac{1}{\sqrt{2\hbar m\omega_m}}\hat{p}
\end{equation}
with $\hat{x},\hat{p}$ the usual canonical position and momentum operators. The gravitational decoherence rate $\Lambda_{\text{grav}}$ is given by 
\begin{equation}
\Lambda_{\text{grav}}=\frac{2\pi}{3} \frac{G\Delta}{\omega_m}
\end{equation}
with $G$ the Newton gravitational constant and $\Delta$ the density of the mechanical element. As one might expect $\Lambda_{\text{grav}}$ is quite small, of the order of $10^{-8}$ s$^{-1}$ for suspended mirrors ( as in LIGO) with $\omega_m\sim 1$. 

Form a phenomenological perspective the effect of gravitational decoherence is analogous to a Browning heating effect. To see this we note that the average
vibrational quantum number increases diffusively
\begin{equation}
\frac{d\langle b^\dagger b\rangle}{dt} = 2\Lambda_{\text{grav}}
\end{equation}
Indeed, one could simulate this effect by adding a stochastic driving force to the mechanical element via the  stochastic Hamiltonian
\begin{equation}
H_s=\frac{dI}{dt}(b+b^\dagger)
\end{equation}
where  $I(t)$ satisfies an Ito stochastic differential equation,
\begin{equation}
dI(t)=\sqrt{4\Lambda_{\text{grav}}}\ dW(t)
\end{equation}
where $dW(t)$ is the Weiner increment. Averaging over all histories of the stochastic driving force gives the final term in Eq. \ref{grav-me}.  

In the absence of mechanical dissipation, there is no steady state. In reality the mechanical quality factor, $Q=\omega_m/\gamma_m$, is finite leading to a steady state with mean phonon number given by 
\begin{equation}
\label{mean-number}
\langle b^\dagger b\rangle_{ss}=\frac{2\Lambda_{\text{grav}}}{\gamma_m}
\end{equation}
This of course assumes that there is no additional mechanical heating (regular thermodynamic kind): hardly a realistic assumption. This adds a large (comparatively) additional term to $\Lambda_{\text{grav}}$ so that we find (for $k_BT >>\hbar\omega_m$),  
\begin{equation} \label{eq:lambda-goesto}
  \Lambda_{\text{grav}}\rightarrow \Lambda_{\text{grav}}
  + \Lambda_{\text{heat}} \,, \quad \text{where } \Lambda_{\text{heat}} 
  = \frac{k_BT}{\hbar Q} \,.
\end{equation}
Given the incredibly large quality factor of $Q=10^{10}$, one would need to cool the mechanical element to nano-Kelvin for the thermodynamical heating to be of the order of the gravitational heating. 

\subsubsection{Optomechanical probe of gravitational decoherence.}
The optomechanical Hamiltonian in cavity-one is
\begin{equation}
\label{OM-ham}
H_{\text{om}}=\hbar\omega_m b^\dagger b+ \hbar g_0 (b+b^\dagger)
\end{equation}
 $g_0$ is the single photon optomechanical coupling rate. Typically $g_0\sim 1$ s$^{-1}$ for the sorts of cavities we are considering here.    This is about the same order of magnitude as $\omega_m$. In new field OM cavity technologies, $g_0$ can be as high as $10^3$ s$^{-1}$ however in such cases the mechanical frequency is also typically much higher $\sim$ tens of MHz.  The interaction time is $T$ which is short compared to the cavity decay time (which we neglect). 
We will assume that the mechanics starts in a thermal state, the steady state of the system subject to gravitational decoherence, heating and dissipation. This is given by
\begin{equation}
\rho_{\text{om}} =\frac{1}{\pi\bar{n}}\int d^2\alpha\  e^{-\frac{|\alpha|^2}{\bar{n}}}\ |\alpha\rangle_b\langle \alpha|
\end{equation} 
where $\bar{n}=\langle b^\dagger b\rangle_{ss}$ is the steady state mean phonon number given in Eq. \ref{mean-number}. 

It is simplest to work in an interaction picture defined by the mechanical free dynamics,
\begin{equation}
  H_{\text{om,I}}= \hbar g_0 (be^{-i\omega_m t}+b^\dagger e^{i\omega_m t})
  \end{equation}
  The corresponding unitary evolution operator is 
  \begin{equation}
  U(t) = e^{\beta(t)b^\dagger-\beta^*(t) b}
  \end{equation}
  where 
  \begin{equation}
  \beta(t)=\frac{g_0}{\omega_m}(e^{-i\omega_m t}-1)
  \end{equation}
  The initial state for the OM interaction is the state at the end of Step 2 
  \begin{equation}
  \rho_{\text{om}}(0) =\frac{1}{2}\left (  |1,0\rangle\langle 1,0|+|0,1\rangle\langle 0,1| +|1,0\rangle\langle 0,1|+|0,1\rangle\langle 1,0|\right )\otimes\rho_{m}
  \end{equation}
  where $\rho_m$ is the state of the mechanical element at the start of the protocol, a thermal state. 
  We can ignore the state of the atomic sources at this stage as they do not participate in the OM interaction. 
  
  The state of the optomechanical system after an interaction time $T$ is given by
  \begin{equation}
  \rho_{\text{om}}(t)=\frac{1}{2}\left (  |1,0\rangle\langle 1,0| \rho_{m}+|0,1\rangle\langle 0,1| U(t) \rho_{m}U^\dagger(t)
                      +|1,0\rangle\langle 0,1| \rho_{m}U^\dagger(t)+|0,1\rangle\langle 1,0|U(t)\rho_{m}\right )
  \end{equation}
  The reduced state of the cavity fields is given by tracing out the mechanical degree of freedom,
  \begin{eqnarray}
  \rho_f(t) =  \frac{1}{2}\left ( |1,0\rangle\langle 1,0|+|0,1\rangle\langle 0,1| +R^*|1,0\rangle\langle 0,1|+R|0,1\rangle\langle 1,0|\right ) \label{eq:rho-f}
  \end{eqnarray}
  where 
  \begin{equation} \label{eq:r}
  R=e^{-(1+2\bar{n})|\beta(t)|^2/2}
  \end{equation}
  where 
  \begin{eqnarray}
    \bar{n} & = &  \overline{n}_{\text{grav}} + \overline{n}_{\text{heat}} \label{eq:n} \\
    & := & \frac{2\Lambda_{\text{grav}}}{\gamma_m}+\frac{2\Lambda_{\text{heat}}}{\gamma_m}
    \label{eq:lambda-sum}
  \end{eqnarray}
  with (as above)
  \begin{equation}
  \Lambda_{\text{grav}} = \frac{2\pi}{3} \frac{G\Delta}{\omega_m} \,, \quad
  \Lambda_{\text{heat}} = \frac{k_B T}{\hbar Q} 
  \end{equation}
  and 
  \begin{equation}
  |\beta(t)|^2=\frac{4g_0^2}{\omega_m^2} \sin^2(\omega_m t/2)
  \end{equation}
  
  Continuing with the protocol from Step 4, now results in the state of the atom-field system
  \begin{equation}
  \rho_{af}(t) =\rho_a(t)\otimes |00\rangle\langle 00|
  \end{equation} where 
  \begin{equation}
  \rho_a(t)=\frac{1}{2}(|ge\rangle\langle ge|+|eg\rangle\langle eg|+R^*|eg\rangle\langle ge|+R|ge\rangle\langle  eg|)
  \end{equation}
The suppression of coherence due to the thermal state of the mechanics has been transferred to a reduction of entanglement in the atomic sources. A readout of the atomic sources will reveal this through either state tomography or via a reduction in a CHSH correlation for a Bell-type experiment.

The function $R(t)$ is a periodic function of time. At each period of the motion it returns to its initial value of zero and the cavity field state would return to the fully entangled state it was in after Step 2. If we chose $T=2\pi/\omega_m$ then the protocol will return the atomic system to the same entangled state in which it began. This is because we have ignored the heating of the mechanics over the period $T$ so the only way decoherence enters is through the initial thermal excitation of the mechanics. In effect the protocol is a thermometer.  We thus see that for maximum effect we need to ensure $g_0 \gg \omega_m$. On the other hand, gravitational heating requires a small value of $\omega_m$ and typically such OM systems have $g_0/\omega_m \ll 1 $. Perhaps technical advances will enable OM systems with long mechanical periods and large single photon coupling. Of course this will also require sub hertz cavity line widths.  In the (exceptionally) optimistic case we can take $T\sim 1 $ nK,  $\omega_m \sim 1$ s$^{-1}$, $\gamma_m\sim 10^{-10}$ s$^{-1}$ so that $Q\sim 10^{10}$.

\subsection{An experimental test of the model} \label{sec:exp-test}

In \cref{sec:grav-deco-model} above, an optomechanical setting has been described. Making some assumptions about how gravitational decoherence influences the optomechanical system, a model has been given that describes how the state of the optomechanical system changes over time. In this section, we consider this model for the state of the optomechanical system as given and analyze it using our decoherence test formalism. We calculate the amount of decoherence that would be introduced to the optomechanical system if the model was correct. We compare this to the amount of decoherence that one would observe if there was no such gravitational decoherence, determining the difference between the two predictions. We devise an experiment that aims at estimating the actual amount of decoherence at a point in time when this difference is maximal. This turns the optomechanical experiment into a test that allows to falsify the above model for gravitational decoherence if it was wrong. This shows that the decoherence testing formalsim presented in this work can be applied in situations where the physical process is unknown. It allows to subject proposed models of the process to a consistency check.

We first present the predicted values of $\Dec(A|E)_\rho$ of the optomechanical system for the two cases where gravitational decoherence is present or absent, respectively, for some example parameters of the experiment. We then calculate the CHSH value $\beta$ that one would have to measure in order to falsify the model for gravitational decoherence. 

The main quantity of interest in our analysis is the decoherence quantity $\Dec(A|E)_\rho$ for the state $\rho_{AB} = \rho_f(t)$ desribed in equation \eqref{eq:rho-f}. This can be calculated using \cref{lem:max-entropy-bell-diagonal}. It turns out to be
\begin{align}
  \Dec(A|E)_\rho &= \frac{1}{4} \left( 1+\sqrt{1-R^2} \right) \\
  &= \frac{1}{4} \left( 1+\sqrt{1-\exp\left(-4(1+2\overline{n})
  \frac{g_0^2}{\omega_m^2} \sin^2 \left(\frac{\omega_m t}{2}\right) \right)} \right) \,.
\end{align}
If the above model is correct and gravitational decoherence occurs, both the gravitational interaction and the mechanical heating contribute to the average vibrational quantum number, i.e. we have
\begin{align}
  \overline{n} &= \overline{n}_{\text{grav}} + \overline{n}_{\text{heat}} \\
  &= \frac{4\pi G}{3} \frac{1}{\gamma_m \omega_m} \Delta 
  + \frac{2k_B}{\hbar} \frac{1}{\omega_m} T
\end{align}
and thus
\begin{align} \label{eq:gd-included}
  \Dec(A|E)_\rho = \frac{1}{4} \left( 1+\sqrt{1-\exp\left(-4\left(1+2\left(\frac{4\pi G}{3} 
  \frac{1}{\gamma_m \omega_m} \Delta + \frac{2k_B}{\hbar} \frac{1}{\omega_m} T\right)\right)
  \frac{g_0^2}{\omega_m^2} \sin^2 \left(\frac{\omega_m t}{2}\right) \right)} \right)
  \quad \text{(gr.\ dec.\ included).}
\end{align}
If gravitational decoherence is absent, then only the mechanical heating contributes to the average vibrational quantum number, i.e. we have
\begin{align}
  \overline{n} &= \overline{n}_{\text{heat}} \\
  &= \frac{2k_B}{\hbar} \frac{1}{\omega_m} T
\end{align}
and thus
\begin{align} \label{eq:gd-neglected}
  \Dec(A|E)_\rho = \frac{1}{4} \left( 1+\sqrt{1-\exp\left(-4\left(1+2\left(\frac{2k_B}
  {\hbar} \frac{1}{\omega_m} T\right)\right)
  \frac{g_0^2}{\omega_m^2} \sin^2 \left(\frac{\omega_m t}{2}\right) \right)} \right) 
  \quad \text{(gr.\ dec.\ neglected).}
\end{align}
Figure \ref{fig:decs} shows how the decoherence quantity in equation \eqref{eq:gd-included} as a function of time varies for different materials of the mechanical element and different temperatures, compared to the case where there is no gravitational decoherence as in equation \eqref{eq:gd-neglected}.

\begin{figure}[htb]
  \begin{center}
  \includegraphics{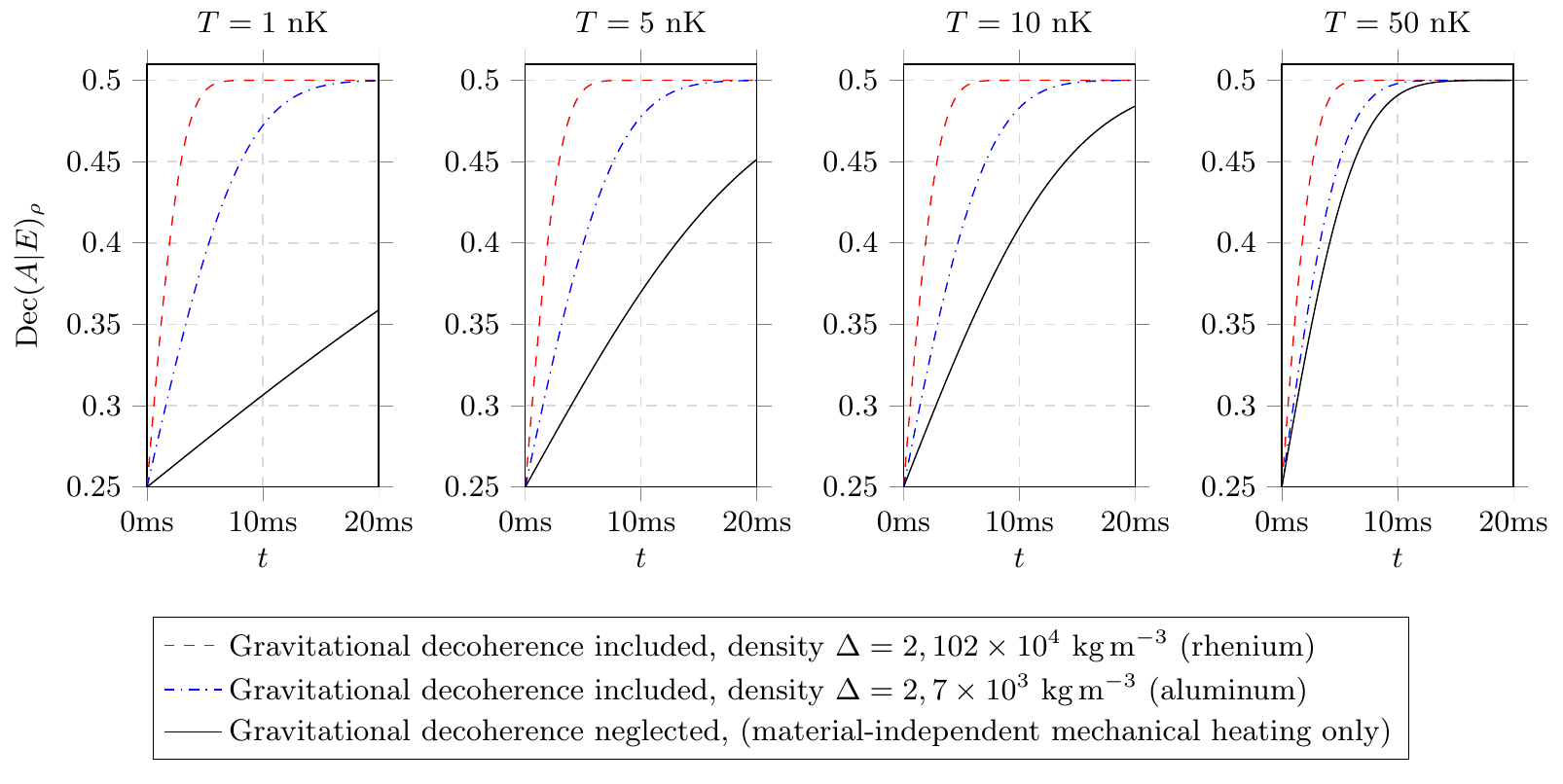}
    \caption{\textbf{Predicted values of the decoherence quantity in the optomechanical 
      experiment.} 
      The decoherence quantity $\Dec(A|E)_\rho$ as in Equation \eqref{eq:gd-included} is 
      plotted as a function of time for different temperatures and two different materials 
      of the mechanical element. In addition, $\Dec(A|E)_\rho$ is plotted for the case 
      where there is no gravitational decoherence, Equation \eqref{eq:gd-neglected}. The calculations have been made for the 
      experimental parameters $g_0 = 1 \ \text{s}^{-1}$, 
      $\omega_m = 1 \ \text{s}^{-1}$ and $\gamma_m = 10^{-10} \ \text{s}^{-1}$.
    }
      \label{fig:decs}
  \end{center}
\end{figure}

In order to rule out the model for gravitational decoherence, one needs to measure a CHSH value $\beta$ which is incompatible with the value of $\Dec(A|E)_\rho$ given in \eqref{eq:gd-included}. The minimal value $\beta_{\text{fals}}$ of $\beta$ that needs to be measured for this falsification can be calculated using \Cref{thm:feasible-region}: Using MATLAB, we numerically evaluated the quantum bound on $\Dec(A|E)_\rho$, which is given as a point-wise maximization problem in \cref{thm:feasible-region}. We inverted the resulting set of data points and interpolated a function from the resulting data using Mathematica. The resulting function takes a value of $\Dec(A|E)_\rho$ as its input and outputs the minimal $\beta$ that needs to be exceeded in a measurement in order to rule out the given value of $\Dec(A|E)_\rho$. Thus, applying this function to the curves of the $\Dec(A|E)_\rho$ values of the gravitational decoherence model in Figure \ref{fig:decs} yields the curves for $\beta_{\text{fals}}$. The results are plotted in \cref{fig:betas} for the same materials and temperatures as above.

\begin{figure}[htb]
  \begin{center}
    \includegraphics{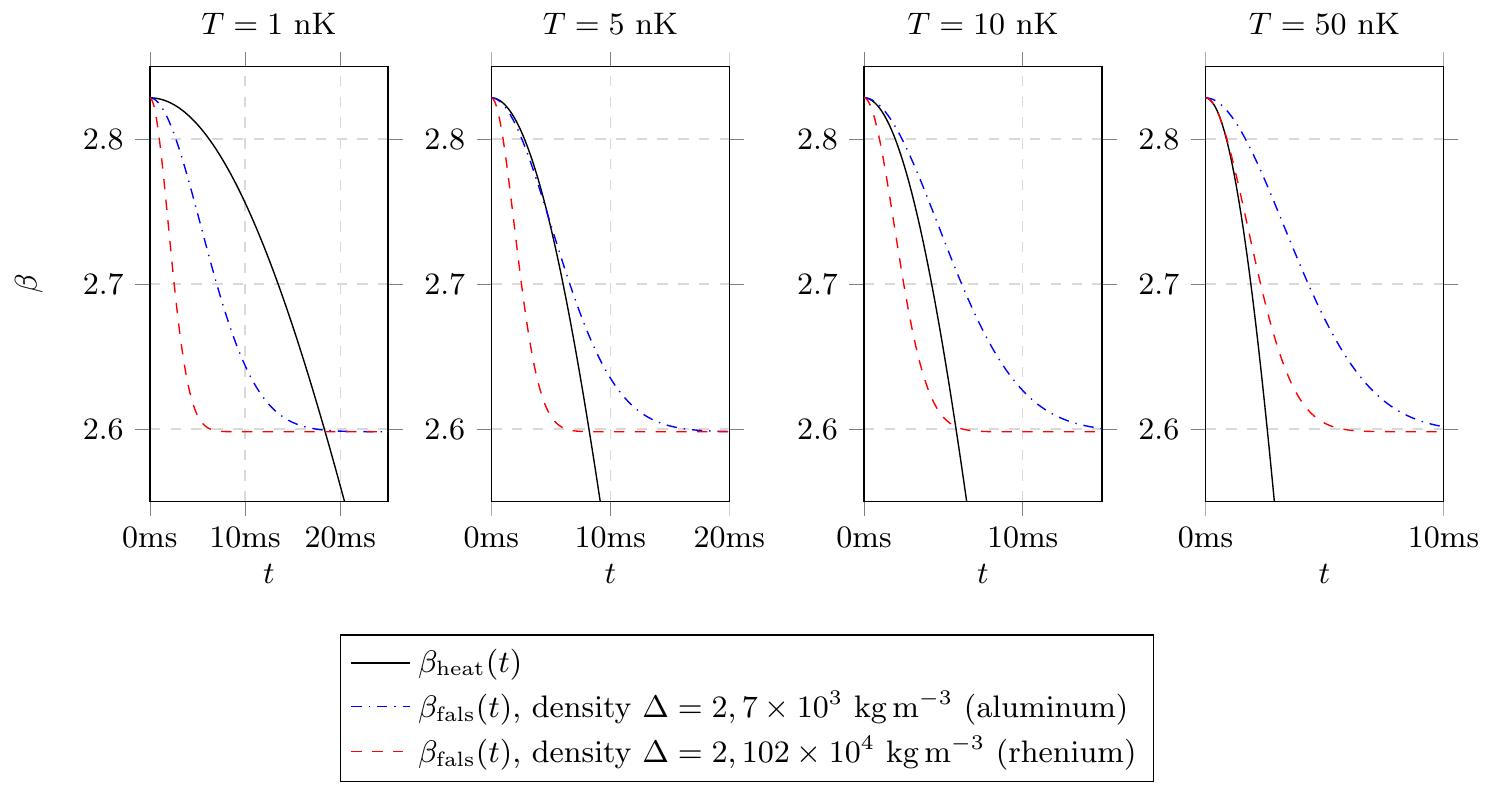}
    \caption{\textbf{Minimal CHSH values for the falsification of the gravitational 
      decoherence model.} 
      The quantity $\beta_{\text{fals}}$, which is the minimal value that needs to be exceeded
      in the measurement of the CHSH value $\beta$ in order to rule out the gravitational 
      decoherence model, is plotted as a function of time for the same materials and 
      temperatures as above. In addition, the value $\beta_{\text{mech}}$ is plotted, which is 
      the CHSH value that can actually be measured using the standard CHSH measurement in the 
      case where gravitational decoherence is absent and only mechanical heating contributes 
      to the decoherence.}
\label{fig:betas}
  \end{center}
\end{figure}

In order to determine whether it is promising to measure a value of $\beta$ that lies above $\beta_{\text{fals}}$, we need to determine the value $\beta_{\text{mech}}$ of $\beta$ which is predicted in the case where gravitational decoherence is absent, i.e. where we only have mechanical heating. We can do that exactly: Equation \eqref{eq:rho-f} gives us an expression for the state, which we consider for the value of $R$ given in the case of mechanical heating only, $\overline{n} = \overline{n}_{\text{mech}}$. Then we calculate the value of $\beta_{\text{mech}}$ for the case where the measurements are taken to be the standard CHSH measurements
\begin{align*}
  &A_0 = \sigma_x \,, &&A_1 = \sigma_z \,, \\
  &B_0 = \frac{\sigma_x - \sigma_z}{\sqrt{2}} \,, &&B_1 = \frac{\sigma_x + \sigma_z}{\sqrt{2}}
  \,,
\end{align*}
where $\sigma_x$, $\sigma_z$ are the Pauli $x$- and $z$-operator, respectively. The resulting curves are shown in \cref{fig:betas} as solid curves. It turns out that for the relevant time interval (where $\beta_{\text{mech}}$ is larger than either of the $\beta_{\text{fals}}$), the curve of $\beta_{\text{mech}}$ for the standard CHSH measurements is almost identical to the curve one would get for the optimal measurements for each time $t$. The latter can be calculated using a formula presented in \cite{HHH95}. This is an experimentally desirable fact: Using a fixed measurement independent of the measurement time is almost optimal.

The most promising measurement time for a falsification of the gravitational decoherence model is given by the time when $\beta_{\text{mech}}$ (that one may hope to actually measure) is high but $\beta_{\text{fals}}$ (which one needs to exceed) is low. Thus, the optimal measurement time can be calculated as the time $t_{\text{max}}$ that maximizes the gap function
\begin{align}
  g(t) := \beta_{\text{mech}}(t) - \beta_{\text{fals}}(t) \,.
\end{align}
This gap function is depends on the density $\Delta$ of the mechanical element and its temperature $T$. One can see that temperatures that look promising for a falsification measurement when looking at the $\Dec(A|E)_\rho$ values in Figure \ref{fig:decs} turn out to be too warm when looking at the experimentally relevant analysis of the $\beta$ values in \cref{fig:betas}. As an example, we have calculated the optimal measurement times for $T = 1 \ \text{nK}$ for the densities of aluminum and rhenium. They are visualized in \cref{fig:gaps}. If there is no gravitational decoherence, one needs to measure values of $\beta$ that are $\sim 0.1$ close (aluminum) or $\sim 0.2$ close (rhenium) to the value that one can maximally measure using the standard CHSH measurements, in order to exclude gravitational decoherence.

\begin{figure}[htb]
  \begin{center}
    \includegraphics{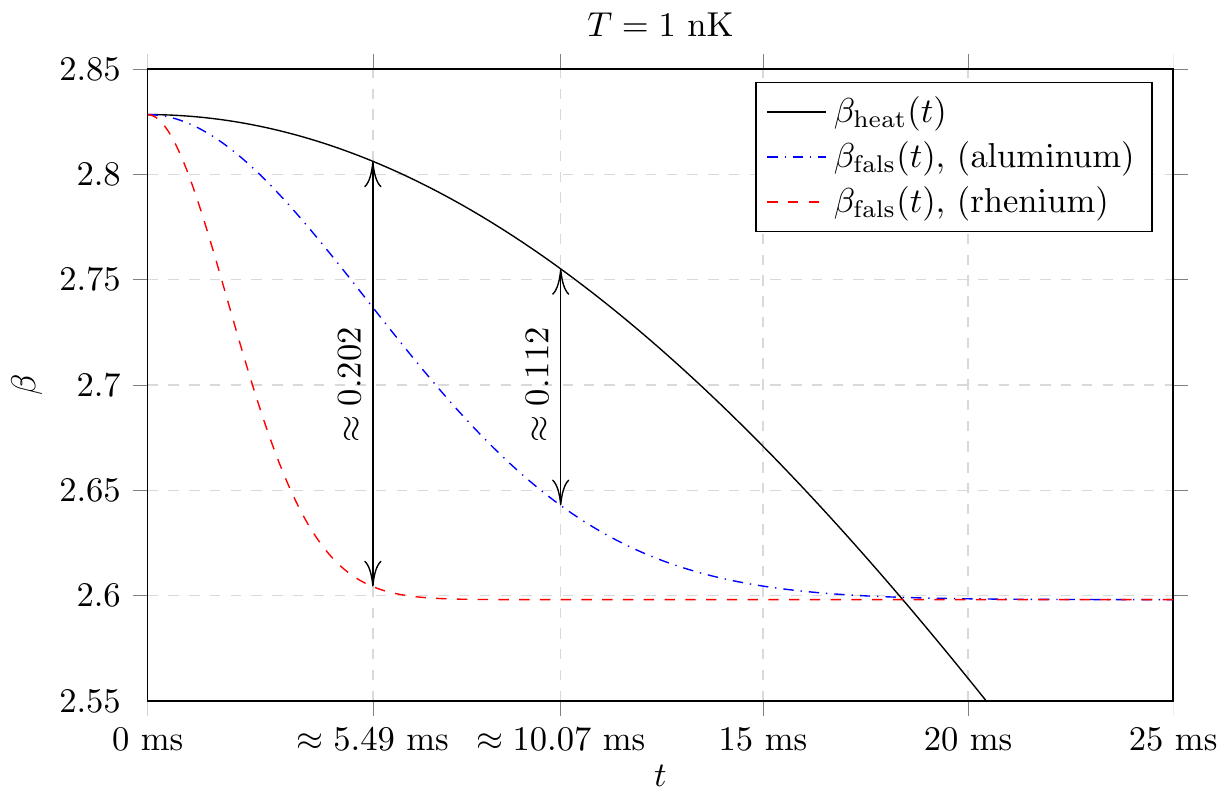}
    \caption{\textbf{Optimal measurement times for ruling out the gravitational decoherence 
      model.} 
      The three plots are identical to the ones in the leftmost box in \cref{fig:betas}, i.e.
      for $T = 1 \ \text{nK}$. In addition, the time $t_{\text{max}}$ at which the gap $g(t)$ 
      between $\beta_{\text{mech}}$ and $\beta_{\text{fals}}$ is maximal is indicated for the 
      two cases where the material of the mechanical element has the density of aluminum or 
      rhenium.}
 \label{fig:gaps}
  \end{center}
\end{figure}

\clearpage

\end{appendix}

\end{document}